\documentclass[11pt,letterpaper]{article}
\usepackage{mehdis}

\newif\ifnotes\notestrue
\ifnotes
\usepackage{color}
\definecolor{mygrey}{gray}{0.50}
\newcommand{\notename}[2]{{\textcolor{mygrey}{\footnotesize{\bf (#1:} {#2}{\bf ) }}}}

\newcommand{\pnote}[1]{{\endnote{#1}}}
\else
\newcommand{\notename}[2]{{}}

\newcommand{\pnote}[1]{}
\fi

\begin{document}

\title{
		\LARGE Classical algorithms, correlation decay, and complex zeros of partition functions of quantum many-body systems \\
}
\author{Aram Harrow \thanks{Center for Theoretical Physics, MIT. \href{mailto:aram@mit.edu}{aram@mit.edu}}\and Saeed Mehraban \thanks{IQIM, Caltech.
 \href{mailto:mehraban@caltech.edu}{mehraban@caltech.edu}} \and Mehdi Soleimanifar \thanks{Center for Theoretical Physics, MIT. \href{mailto:mehdis@mit.edu}{mehdis@mit.edu}}}

\date{\today}

\maketitle

\begin{abstract}
Various statistical properties of quantum many-body systems in thermal equilibrium such as the free energy, entropy, and average energy can be obtained from the \emph{partition function}. The problem of estimating the partition function has been the subject of numerous studies in statistical physics, computer science, and machine learning.  The aim of this work is to present a new classical algorithm for estimating the partition function of quantum systems. We achieve this by studying the connection between the \emph{hardness of approximating} the partition function and the \emph{thermal phase transition}. In particular, we show the following:

\begin{enumerate}
     \item[(1)] We demonstrate a quasi-polynomial time classical algorithm that estimates the partition function of quantum systems above the phase transition point. The running time of this algorithm relies heavily on the locus of the complex zeros of the partition function. Intriguingly, these complex zeros are known to mark where the phase transition occurs. By a result of \cite{sly_hardcore}, in the worst case, the same problem is $\NP$-hard below this point. Together with our work, this shows that the transition in the phase of a quantum system is also accompanied by a transition in the hardness of approximation.
    \item [(2)]

    We show that in a system of $n$ particles at temperatures above the phase transition point, where the complex zeros are far from the real axis, the correlations between two observables whose distance is $\Omega(\log n)$ decay exponentially. We can improve the factor of $\log n$ to a constant when the Hamiltonian has commuting terms or is on a $1$D chain. Previously, the decay of correlations was only proved for translationally-invariant $1$D systems \cite{araki_1d} or at very high temperatures \cite{Kastoryano_locality}.
     
    \item[(3)] We find a deterministic quasi-polynomial time approximation algorithm for the $\mathrm{XXZ}$ model in the \emph{ferromagnetic} regime at any temperature over arbitrary graphs. Previously, a randomized algorithm was known only for the ferromagnetic $\mathrm{XY}$ model \cite{Bravyi_ferro}.
\end{enumerate}
This work is the first rigorous study of the connection between the complex zeros of the partition function and the decay of correlations in quantum many-body systems and extends a seminal work of Dobrushin and Shlosman on classical spin models \cite{Dobrushin1}. On the algorithmic side, our result extends the scope of a recent approach due to Barvinok for solving classical counting problems \cite{Barvinok_book} to quantum many-body problems.
\end{abstract}

\newpage
{
  \hypersetup{linkcolor=black,linktoc=all}
  \tableofcontents
  \newpage
}
%%%%%%%%%%%%%%%%%%%%%%%%%%%%%%%%%%%%%%%%%%%%%%%%%%%%%%%%%%%%%%%%%%%
\section{Introduction}\label{sec:intro}
At low temperatures, the main characteristics of many-body systems in condensed matter physics or quantum chemistry are captured in the structure of the ground state of their Hamiltonian. The computational complexity of estimating the ground state energy has been extensively studied through numerous works. In particular, it has been shown that in the worst case, for many physically relevant systems including even a two-local Hamiltonian on a one-dimensional ($1$D) chain, estimating the ground state energy is \QMA-complete \cite{aharonov_line}. On the other hand, there is a host of classical algorithms for efficiently estimating the ground state energy in certain restricted examples like a gapped Hamiltonian on a $1$D chain \cite{arad_1d} or a dense interaction graph \cite{brandao2013product}.

While at low temperatures the system is in the vicinity of the ground space, at finite temperatures, the state of the system is a mixture of different excited states. In thermal equilibrium, a quantum system characterized by a local Hamiltonian $H$ is in the Gibbs (or thermal) state $\r=\exp(-\b H)/Z_{\b}(H)$, where $\b$ is the inverse of temperature and $Z_{\b}(H)=\Tr[\exp(-\b H)]$ is the \emph{partition function} of the system. A natural equivalent to the ground energy at finite temperatures is the \emph{free energy} which is defined as $F_{\b}(H)=-1/\b \log Z_{\b}(H)$. Many useful statistical properties of the system including the free energy and entropy can be obtained from the partition function and its derivatives. However, exactly evaluating the partition function is known to be $\sharpP$-hard. Hence in order to characterize the finite-temperature behavior of the system, it is crucial to have efficient algorithms that \emph{approximate} this quantity. 

Our starting point for finding such approximation algorithms is based on the observation that the phenomenon of the thermal phase transition is an obstacle for finding \emph{efficient} algorithms. Consider a quantum many-body system that consists of $n$ qudits interacting according to a local Hamiltonian $H$. As the temperature of this system increases, meaning $\b\rightarrow 0$, the Gibbs state $\rho$ approaches the maximally mixed state $\iden/d^n$. Thus, in this case, finding the partition function is trivial since $Z_{\b=0}(H)=d^n$. On the other hand, this problem becomes significantly harder at lower temperatures. In particular, as $\b\rightarrow \infty$, the Gibbs state approaches the ground space of the Hamiltonian $H$ and the free energy $F_{\b}(H)$ approaches the ground energy which is known to be $\QMA$-hard to estimate. Hence, we see that the computational hardness of estimating the partition function (or equivalently the free energy) depends on the inverse temperature $\b$ and goes through a transition from being trivial to $\QMA$-hard as $\b$ increases. 

In statistical physics, however, another transition occurs as $\b$ increases, namely, the transition in the \emph{phase} of the system. At the thermal phase transition point, certain physical properties of the system undergo an abrupt change. An example of such a transition is when a magnetic material that consists of a network of interacting spins goes from the \emph{ferromagnetic} to the \emph{paramagnetic} phase. In the ferromagnetic phase, most spins are pointing in the same direction and their net magnetic effect is non-zero, whereas in the paramagnetic phase, the spins are distributed equally in opposite directions making their net magnetic effect zero. This transition does \emph{not} happen gradually as $\b$ varies. On the contrary, the phase of the system changes suddenly at some critical inverse temperature $\b_c$ known as the phase transition point. 

Does the computational hardness of estimating the partition function also undergo an abrupt change at the same transition point? This question has been studied in the context of the \emph{classical Ising} or \emph{hard-core model}, and the answer is known to be affirmative. For these systems, there are efficient algorithms for estimating the partition function when $\b<\b_c$ \cite{Weitz,Sinclair_antiferro} whereas by a result of Sly and Sun \cite{sly_sun,sly_hardcore} the same problem is $\NP$-hard for $\b>\b_c$.

Hence, it appears that the thermal phase transition poses a barrier to obtaining efficient algorithms, and we need a framework for characterizing this phenomenon. There are at least two methods for such purpose. One, which is the basis of our algorithm, stems from analyzing the locus of the complex zeros of the partition function. Another seemingly different method involves the decay of long-range order in the Gibbs state of the system. In this work, we study the interface between these two methods and their algorithmic implications. In particular, we find a quasi-polynomial time approximation algorithm for the partition function for temperatures far from the complex zeros and show that the correlations in the Gibbs state decay exponentially in the same temperature range. The following section summarizes our results. 

\subsection{Our main results}
\subsubsection{The complex zeros of the partition function}

In general, the partition function can be written as $Z_{\b}(H)=\sum_k \exp(-\b E_k)$, where each $E_k$ is an eigenvalue of the Hamiltonian $H$. If $\b$ is real, the terms $\exp(-\b E_k)$ are all strictly positive, and hence the partition function $Z_{\b}(H)$ is strictly positive itself. However, this changes when $\b$ is allowed to be complex. In that case, the terms $\exp(-\b E_k)$ acquire complex phases that when added together might cancel each other and make the partition function zero. We call the solutions of $Z_{\b}(H)=0$ for $\b\in \bbC$ the \emph{complex zeros} of the partition function. 

The significance of these zeros becomes more clear if one looks at the free energy $F_{\b}(H)$. The zeros of $Z_{\b}(H)$ are the \emph{singularities} of $\log Z_{\b}(H)=-\b F_{\b}(H)$. Since $Z_{\b}(H)\neq0$ when $\b$ is real, we see that all these singularities are located in the complex plane and the free energy is analytic near the real axis. As the number of particles $n$ grows, the number and location of these points change. Perhaps rather surprisingly, some of these singularities approach the real axis in the limit of a large number of particles, $n\rightarrow \infty$. The point on the real axis where these zeros converge in the large $n$ limit is called the critical inverse temperature and denoted by $\b_c$ (see \fig{zeros_fisher}). This critical temperature separates different phases of matter and important quantities such as the free energy become non-analytic in the vicinity of $\b_c$. The study of these complex zeros in connection with the phase transition phenomenon in classical Ising models was initiated by Lee and Yang \cite{Lee-yang} and later extended by Fisher \cite{Fisher}. This approach is one of the few rigorous methods available in the theory of phase transitions. 

One can go beyond partition functions and consider complex roots of other high-degree polynomials that appear in combinatorics such as estimating the permanent of a matrix. Recently, there has been a surge of interest in studying these complex zeros in theoretical computer science due to their algorithmic applications. In particular, a new approach introduced by Barvinok \cite{Barvinok_book} directly connects the locus of the complex zeros to approximation algorithms for counting problems. In this work, we extend the scope of this method by applying it to quantum many-body systems.

We first state the condition on the location of zeros that we use in our approximation algorithm. Under this condition, it is guaranteed that the inverse temperature $\b$ at which the partition function is estimated is connected to $\b=0$ by a path in the complex plane that avoids the complex zeros along its way with a significant margin. Even though this algorithm works for any such path, we restrict our attention to the physically-relevant case when this zero-free region contains the real $\b$-axis. Hence, we define:

\begin{definition}\label{def:informal vicinity}
The $\d$-neighborhood of the interval $[0,\b]$ for some $\b\in \bbR^+$ is a region of the complex plane defined as $\Omega_{\d,\b}=\{z\in \bbC: \exists z'\in [0,\b], |z-z'|\leq \d \}$ (see \fig{zeros_fisher} for an example of such a region).
\end{definition}
\begin{definition}[Informal version of \condreftwo{simple absence of zeros}{analyticity after measurement}]\label{def:informal absence of zeros}
For a system of $n$ particles with a local Hamiltonian $H$, we define:
\begin{enumerate}
    \item A $\d$-neighborhood $\Omega_{\d,\b}$ of the interval $[0,\b]$ (see \defref{informal vicinity}) is called zero-free if $\d$ is some constant and $\forall\b'\in \Omega_{\d,\b}$ the partition function $Z_{\b'}(H)\neq 0$ and moreover, $|\log Z_{\b'}(H)|\leq O(n)$. 
    \item Equivalently, the free energy $F_{\b}(H)$ is called $\d$-analytic along $[0,\b]$ if $\Omega_{\d,\b}$ is a zero-free region. 
    \end{enumerate}
\end{definition}

We now state our first result.

\begin{thm}[Informal version of \thmref{extrapolation algorithm}]\label{thm:informal extrapolation algorithm}
There is a deterministic classical algorithm that takes a local Hamiltonian $H$ and a number $\e$ as inputs, runs in time $n^{O(\log(n/\e))}$, and outputs a value within $\e$-multiplicative error of the partition function $Z_{\b}(H)$ at inverse temperature $\b$ as long as the free energy is $\d$-analytic along the $[0, \b]$ line (see \defref{informal absence of zeros}).
\end{thm}

The critical point $\b_c$ where the zero-free region ends has been precisely determined for some specific systems such as the \emph{classical} Ising model. In general, though, it is a hard problem to exactly find this point given an arbitrary Hamiltonian. One can compare this with when a $1$D Hamiltonian is assumed to have a constant gap. Under this condition, there is an efficient algorithm for estimating the ground energy. However, it has been shown that validating this condition, i.e. determining if a Hamiltonian is gapped or not, is undecidable in the worst case \cite{cubitt_undecidablegapp}. 

In our next result, we find a constant \emph{lower} bound on the critical point $\b_c$. We show that there is a zero-free \emph{disk} of radius $\b_0$ around $\b=0$ for some constant $\b_0\leq \b_c$. We prove this for geometrically-local Hamiltonians in which the local terms act on neighboring qudits that are located on a $D$-dimensional lattice $\L\subset \mathbb{Z}^D$.

\begin{thm}[Informal version of \thmref{t_higH_temp}]\label{thm:informal high t disk of no zeros}
There exists a real constant $\b_0$ such that for all $\b\in \bbC$ with $|\b| \leq \b_0$, the partition function $Z_{\b}(\L)$ of a geometrically-local Hamiltonian $H$ does not vanish, and furthermore, $\big|\log|Z_{\b}(\L)|\big|\leq O(n)$.
\end{thm}

\begin{figure}[t!]
\centering
	\includegraphics[width=.48\textwidth]{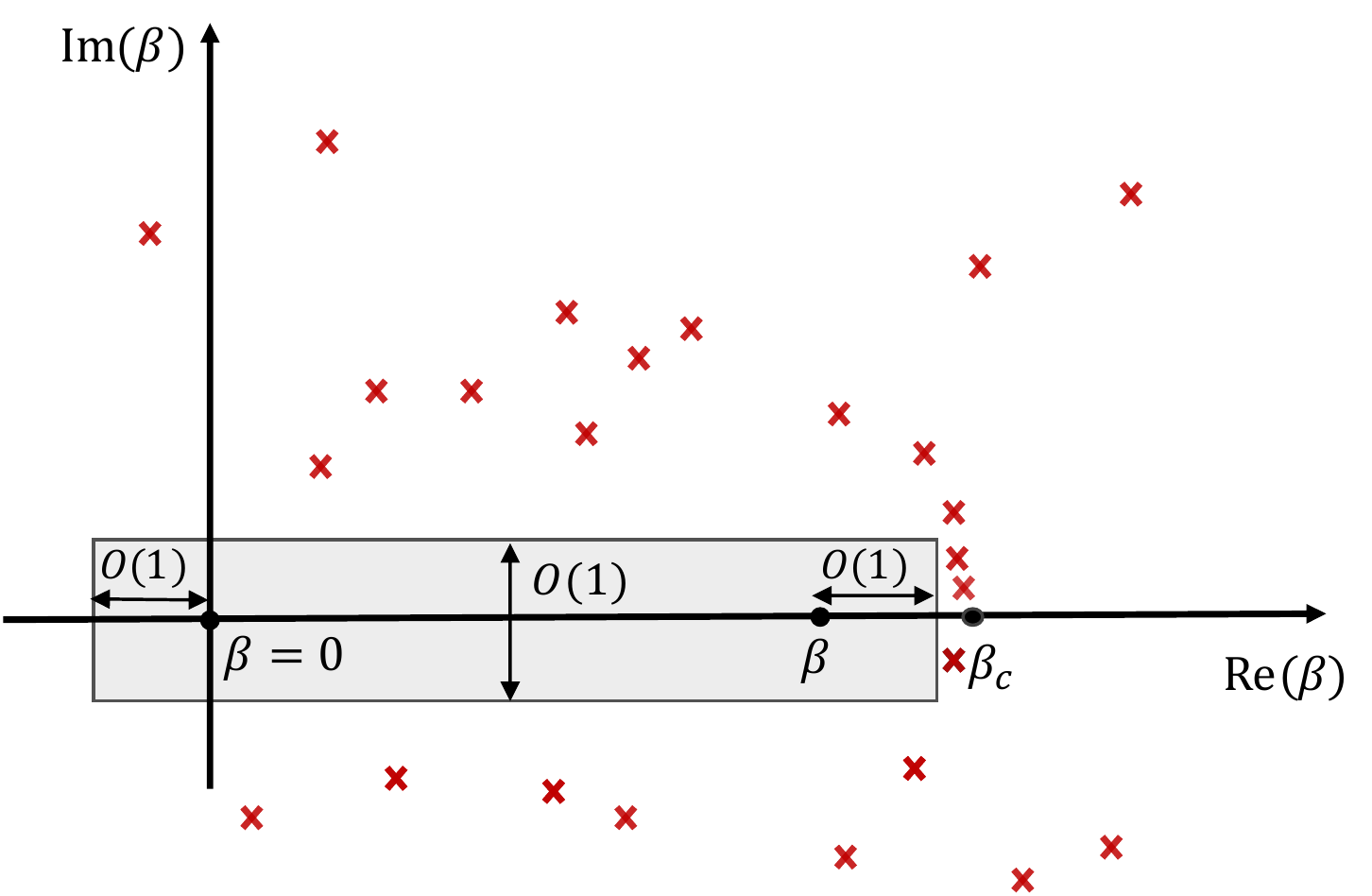}
	\caption{The location of complex zeros of the partition function, the critical point $\b_c$, and the zero-free region near the real axis (as in \defref{informal absence of zeros}). The free energy is analytic in this region.
 \label{fig:zeros_fisher}
}
\end{figure}
\subsubsection{The decay of correlations in the Gibbs state}
Another signature of the thermal phase transition is the appearance of long-range order in the system. In the example of a magnetic system, below the phase transition in the ferromagnetic phase (also called the \emph{ordered} phase), distant spins are correlated and point in the same direction, whereas in the paramagnetic phase (also known as the \emph{disordered} phase), the correlations between disjoint parts of the system decay exponentially with their distance. More precisely, we define the exponential decay of correlations as

\begin{definition}[Informal version of \condref{Exponential decay of correlations}]\label{def:informal decay of cor}
The Gibbs state $\rho_{\b}(H)$ of a geometrically-local Hamiltonian $H$ at inverse temperature $\b$ exhibits an exponential decay of correlations if for any two disjoint observables $O_1$ and $O_2$ there exist constants $\xi$ and $c$ such that 
\ba\label{eq:a3}
\big|\Tr\left[ \rho_{\b}(H) O_1 O_2 \right]-\Tr\left[ \rho_{\b}(H) O_1 \right]\Tr\left[ \rho_{\b}(H) O_2 \right] \big| \leq c\norm{O_1}\norm{O_2} e^{-\dist(O_1,O_2)/\xi}.
\ea
\end{definition}

Besides its physical significance, this property also has algorithmic applications and has been studied both in classical \cite{Weitz} and quantum \cite{Kastoryano_locality, Brandao_gibbs_preparing} settings. What is the relation between this notion of the phase transition and the complex zeros of the partition function? Note that the former involves correlations in the system at a \emph{real} temperature while the latter concerns the \emph{complex} temperature features of the partition function. Could it be that these two apparently distinct characterizations are indeed equivalent? 

In a seminal work \cite{Dobrushin1}, Dobrushin and Shlosman proved that for translationally-invariant classical systems, the decay of correlations is actually equivalent to the analyticity of the free energy and the existence of a zero-free region. Recently, a more refined version of this equivalence was proved for the classical Ising model \cite{Liu_zeros}.

The same question has been open for quantum systems. Our next two results suggest an affirmative answer. Our first result shows that the absence of complex zeros around some real $\b$ implies the exponential decay of correlations at that $\b$. 

\begin{thm}[Informal version of results in \secref{Zero-free region implies the exponential decay of correlations}]\label{thm:informal decay from zero}
Let $\r_{\b}(H)$ be the Gibbs state of a geometrically-local Hamiltonian at inverse temperature $\b$ in the zero-free region $\Omega_{\d,\b}$ given in \defref{informal absence of zeros}. This state has the decay of correlation property as in \defref{informal decay of cor} in any of the following cases: 
\begin{itemize}
    \item [(i.)] The distance between the observables $O_1$ and $O_2$ is at least $\Omega(\log n)$,
    \item [(ii.)] The Hamiltonian $H$ is the sum of mutually commuting local terms, or
    \item [(iii.)] The Hamiltonian $H$ is defined on a $1$D chain.
\end{itemize}
\end{thm}
The class of commuting Hamiltonians include important examples such as stabilizer Hamiltonians like the Toric code, Color code, or Levin-Wen model \cite{levin_wen}. 
 
Proving the converse of \thmref{informal decay from zero} turns out to be more challenging. Nevertheless, we can give evidence for this direction by generalizing the result of \cite{Dobrushin1} to classical systems that are not translationally invariant, and also quantizing certain steps in the proof.  

\begin{thm}[Informal version of \thmref{The decay of correlations implies the absence of zeros}] \label{thm:informal zeros from decay}
Let $H$ be a geometrically-local Hamiltonian of a classical spin system, i.e. the local terms $H_i$ are all diagonal in the same product basis. For this system, the exponential decay of correlations given in \defref{informal decay of cor} implies the absence of zeros near the real axis as in \defref{informal absence of zeros}.
\end{thm}

The importance of fully establishing this equivalence between the decay of correlations and the absence of zeros is twofold. On one hand, this can be thought of as an improvement on \thmref{informal high t disk of no zeros}. This means we can prove the analyticity of the free energy not only below the lower bound $\b_0$ that we found, which might be smaller than the exact value $\b_c$, but also for any $\b$ at which the decay of correlations holds. On the other hand, this equivalence allows us to use the locus of zeros to extend the range of $\b$ where the system exhibits the decay of correlations from a constant (by a result of \cite{Kastoryano_locality}) to the critical point $\b_c$. Overall, this equivalence rigorously confirms the physical intuition that a quantum system enters the disordered phase at the point where the free energy becomes analytic. 

\subsubsection{Two-local Hamiltonians and Lee-Yang zeros}

For our last result, we switch gears and focus on a specific family of $2$-local Hamiltonians. We again use the idea of extrapolation, but this time, our extrapolation parameter instead of $\b$ is the strength of the external magnetic field applied to the system in the $z$-direction. The physical motivation is that when the system is subject to a large external field in a specific direction (the $z$-direction in our case), all spins align themselves in that direction, and estimating the properties of the system becomes trivial. On the other hand, as we move to smaller fields, the other interaction terms between the particles gain significance, making the problem non-trivial. Our result is an approximation algorithm for the quantum $\mathrm{XXZ}$ model with the following Hamiltonian:

\begin{definition}\label{def:informal XXZ}
The anisotropic $\mathrm{XXZ}$ Hamiltonian on an interaction graph $G=(V,E)$ is given by
\ba\label{eq:a4}
H(\mu)=-\sum_{(i,j)\in E} \left(J_{ij}(X_i X_j+Y_i Y_j)+J_{ij}^{zz}Z_i Z_j\right)-\mu \sum_{i\in V}Z_i.
\ea
\end{definition}
We find an approximation algorithm for this model. This is stated in the following theorem. 
\begin{thm}[Informal version of \thmref{alg for xxz}]\label{thm:informalalg for xxz}
There is a deterministic algorithm that runs in $n^{O(\log (n/\e))}$ time and outputs an $\e$-multiplicative approximation to the partition function of the anisotropic $\mathrm{XXZ}$ model (see \defref{informal XXZ}) in the ferromagnetic regime, i.e. when $J_{ij}^{zz}\geq |J_{ij}|$, and $\mu$ is an arbitrary constant.
\end{thm}

\subsection{Sketch of our techniques}\label{sec:Sketch of our techniques}

\paragraph{Sketch of the proof for \thmref{informal extrapolation algorithm}} The basis of our algorithm in \thmref{informal extrapolation algorithm} is the following observation. It is computationally easy to find the partition function and its derivatives at $\b=0$. Note that in a system of $n$ qudits, $Z_{\b=0}(H)=d^n$ and its derivatives are
\ba\label{eq:a1}
\frac{d^k Z_{\b}(H)}{d \b^k}\Bigr|_{\b=0}=(-1)^k \Tr[H^k].
\ea

Since the local Hamiltonian $H$ equals $\sum_{i=1}^m H_i$ for some $m=\poly(n)$, its $k$th power $H^k$ is also the sum of $n^{O(k)}$ many local terms, i.e.
\ba
H^k=\sum_{j=1}^{n^{O(k)}} H^{(k)}_j,
\ea
where $H^{(k)}_j$ is a product of $k$ local terms $H_i$. Each of the new terms $H^{(k)}_j$ acts on a region that is at most $k$ times larger than the support of the original terms $H_i$ which is still some constant. We can find $\Tr[H^k]$ by adding $n^{O(k)}$ many terms like $\Tr[H^{(k)}_j]$, which allows us to compute the derivatives \eqref{eq:a1} in time bounded by $n^{O(k)}$. 

How can the solution at $\b=0$ be used to estimate the one at some non-zero $\b$? We use a technique due to Barvinok \cite{Barvinok_permanent,barvinok_clique} that has been applied to similar counting problems. The idea is to \emph{extrapolate} this solution at $\b=0$ to find $Z_{\b}(H)$ at some non-zero $\b$ where the problem is non-trivial. The extrapolation is done simply by using a \emph{truncated Taylor expansion} of $\log Z_{\b}(H)$ at $\b=0$. Since our goal is to find the partition function with some $\e$-multiplicative error, it is sufficient to estimate $\log Z_{\b}(H)$ within $\e$-additive error.

The main barrier to the reliability of this algorithm is establishing the fast convergence of the Taylor expansion. Such a Taylor expansion is only valid when $\log Z_{\b}(H)$ remains a complex-analytic function, meaning the extrapolation is done along a path contained in the zero-free region. This is precisely the condition stated in \defref{informal absence of zeros}. Under this assumption, the Taylor theorem along with the bound $|\log Z_{\b}(H)|\leq O(n)$ that we get from being in the zero-free region give
\ba
\left|\log Z_{\b}(H)- \sum_{k=0}^{K-1} \frac{1}{k!}\frac{d^k \log Z_{\b}(H)}{d \b^k}\Bigr|_{\b=0}\right|\leq c_1 n e^{-c_2 K}
\ea
for some constants $c_1,c_2$ (see \propref{taylor_bounded} in the body for details). The running time of computing the terms in this expansion is dominated by that of finding the derivatives which, as mentioned earlier, takes time $n^{O(K)}$. To get an additive error of $\e$ for $\log Z_{\b}(H)$, it suffices to choose $K=O(\log(n/\e))$ resulting in a quasi-polynomial time algorithm.

The running time of this algorithm depends exponentially on the distance between the zeros and the extrapolation path. This allows us to clearly see why our algorithm fails beyond the phase transition point. If we try to extrapolate to $\b\geq\b_c$, we need to find a zero-free region that avoids the ``armor" of zeros that are concentrated around the real axis at $\b_c$. This results in a zero-free region with a vanishing width. Hence, the running time blows up, which matches our expectation from the $\NP$ hardness result above $\b_c$ \cite{sly_sun}.

\paragraph{Sketch of the proof for \thmref{informal decay from zero}} The technique used in the proof of \thmref{informal decay from zero} is inspired by the extrapolation idea of \thmref{informal extrapolation algorithm} and also the proof of the similar statement for the classical systems due to \cite{Dobrushin1}. 

For any given disjoint observables $O_1$ and $O_2$, we define a function $f(\b)$ that measures the correlation between them. This function is defined in a slightly different way than the \emph{covariance} form in \eqref{eq:a3} and is tuned to have specific properties. In particular, we show that at $\b=0$, the value of this function is zero, i.e. $f(0)=0$. This is expected intuitively since the system is in the maximally mixed state at $\b=0$ and particles are distributed independently at random. However, we further show that the low order derivatives of this function up to $O(\dist(O_1,O_2))$ are all zero at $\b=0$, i.e. 
\ba
\frac{d^k f(\b)}{d\b^k}\Bigr|_{\b=0}=0,\quad \text{for\ } k=0,1,\dots, O(\dist(O_1,O_2)).
\ea
Hence, this function looks very flat around the origin. Additionally, we prove that $f(\b)$ is an analytic function in the zero-free region. Finally, we show that this together with the constraints on the derivatives imply that the value of $f(\b)$, which shows how correlated $O_1$ and $O_2$ are, remains exponentially small when moving from the origin to a constant $\b$. 

This gives us an upper bound $\propto n\exp(-\dist(O_1,O_2)/\xi)$ on the amount of correlation. The extra factor of $n$ makes this bound exponentially small when $\dist(O_1,O_2)=\Omega(\log n)$.

\begin{rem}
Even with the extra factor of $n$, our bound remains useful for algorithmic applications such as in \cite{Brandao_gibbs_preparing}. There one needs to split the system into computationally tractable smaller pieces and solve the problem for those pieces locally. The error of this strategy can be bounded using the exponential decay of correlations. To keep this error less than $1/\poly(n)$, one needs to choose the distances to be $O(\log n)$ which is the regime that our result covers. 
\end{rem}

We are able to remove the constraint $\dist(O_1,O_2)=\Omega(\log n)$ in certain instances. This includes when the Hamiltonian consists of commuting terms or when it is defined on a $1$D chain. In both cases, using either the commutativity of local terms or the quantum belief propagation \cite{hastings_belief_propagation} (refer to \propref{Quantum belief propagation} in the body for the precise statement), we show that by removing the interaction terms acting on particles that are far from the observables $O_1$ and $O_2$, the correlations between $O_1$ and $O_2$ do not change by much. Hence, the system size reduces to the number of particles in the vicinity of the two observables. This number replaces the prefactor $n$ we had before and is negligible compared to the exponential factor $\exp(-\dist(O_1,O_2)/\xi)$. Thus, for these systems, the decay of correlations holds even when $\dist(O_1,O_2)$ is a constant. 

\paragraph{Sketch of the proofs for \thmref{informal high t disk of no zeros} and \thmref{informal zeros from decay}} We first introduce a core idea which plays a central role in the proofs of both \thmref{t_higH_temp} and \thmref{informal decay from zero}. For ease of notation, we denote the partition function of a geometrically-local Hamiltonian $H$ defined over a $D$-dimensional lattice $\L\subset \mathbb{Z}^D$ by $Z_{\b}(\L)$. The particles are located on the vertices of this lattice. 

In \thmref{informal high t disk of no zeros}, our goal is to show that $Z_{\b}(\L)\neq 0$ inside a disk of radius $\b_0$, i.e. for $\b\in \bbC$ where $|\b|\leq \b_0$ for some constant $\b_0$. We consider a series of sublattices $\emptyset=\L_0\subset \L_1\subset \L_2\subset \dots \subset \L_n=\L$ such that each sublattice $\L_i$ has one fewer vertex than $\L_{i+1}$. By convention, we let $Z_{\b}(\emptyset)=1$. As long as the sublattice $\L_i$ has only a constant number of particles, we can always ensure $Z_{\b}(\L_i)\neq 0$ by choosing $\b$ to be a sufficiently small constant. One might worry that by adding more particles, the partition function vanishes. Our main contribution is to prove this does not happen. We do so by showing that the partition function after involving new particles cannot become smaller than a constant fraction of the partition function before adding the particles. In other words, we show there exists a constant $c> 1$ such that
\ba
|Z_{\b}(\L_{i+1})|\geq c^{-1} |Z_{\b}(\L_{i})|,\quad i\in\{1,2,\dots,n-1\}.\label{eq:a5}
\ea
By repeatedly applying this bound, we obtain the following exponentially small (yet sufficiently large for our purposes) lower bound on the partition function of the whole system
\ba
|Z_{\b}(\L)|\geq c^{-n} |Z_{\b}(\L_1)|.
\ea
This leads to the bound given in \thmref{informal high t disk of no zeros}. This lower bound is obtained using a method known as the \emph{cluster expansion}. These expansions are widely used in statistical physics to study the high temperature behavior of classical and quantum many-body systems. The cluster expansion we use is due to Hastings \cite{Hastings_solving_gapped_locally,Kastoryano_locality}, which represents the operator $\exp(H)$ as sum of products of local terms $H_i$. This allows us to express $Z_{\b}(\L_{i+1})$ in terms of $Z_{\b}(\L_{i})$ plus some small correction terms that account for the interaction terms acting on the added particle. Our main contribution is to use an inductive proof to connect such a decomposition to the lower bound \eqref{eq:a5} (see the proof of \thmref{t_higH_temp} in the body for details). 

A similar strategy is used in the proof of \thmref{informal zeros from decay} which closely follows the proof of the same statement for translationally-invariant systems in \cite{Dobrushin1}. We essentially show a similar bound to \eqref{eq:a5} on how much the partition function can shrink after adding new particles. Here, instead of cluster expansions, we use the exponential decay of correlations to show such a lower bound. However, notice that the decay of correlations is a property of the system at a \emph{real} $\b$, whereas we want to bound the absolute value of the partition function at some \emph{complex} $\b$. There are multiple steps in the proof before we can get around this issue. 

One crucial step is to reduce the proof of the analyticity of the free energy to a condition that roughly speaking (see \propref{fraction with projectors} for the details) states that changing the value of a spin in the system only causes a small relative change in the partition function of the system even for complex $\b$. We prove this by isolating the effect of this spin flip from the rest of the system using the decay of correlations. This requires removing the imaginary part of $\b$ for all the interactions in the vicinity of the flipped spin and bounding the resulting error. 

This overall approach involves a subtle use of the boundary conditions in the spin system. In the quantum case, this means applying local projectors to the Gibbs state before evaluating the partition function. These projectors can in general be entangled which makes using this proof technique more challenging for quantum systems.

\paragraph{Sketch of the proof for \thmref{informalalg for xxz}} Thus far we have only considered complex zeros of the partition function as a function of $\b$. These are often called Fisher zeros \cite{Fisher}. One can, however, fix $\b$ and consider the partition function as a function of other parameters in the Hamiltonian. When that parameter is the strength of the external magnetic field denoted by $\mu$, these zeros are called Lee-Yang zeros \cite{Lee-yang}. In a pioneering result, Lee and Yang showed that for ferromagnetic systems, the locus of these zeros can be exactly determined and they are all on the imaginary axis in the complex $\mu$-plane.

A generalization of this theorem has been proved for a class of $2$-local quantum systems including the anisotropic Heisenberg model \cite{Suzuki_XYZ}. The result follows by mapping the quantum system to a classical spin system and applying a Lee-Yang type argument to the classical model. 

Knowing the location of the complex zeros, we use the extrapolation algorithm to estimate the solution at a constant $\mu$ by finding the low-order derivatives of the partition function at $\mu=0$. We can apply this to the quantum $\mathrm{XXZ}$ model given in \eqref{eq:a4}. 

\subsection{Previous work }

\subsubsection{Classical statistical physics and combinatorial counting}

The Gibbs distribution and partition function appear naturally in combinatorial optimization, statistical physics, and machine learning. In particular, the classical \emph{Ising model} has been studied extensively within these areas. These studies have cultivated in various probabilistic and deterministic approximation algorithms for this model and its variants. In the following, we summarize some of these results.

Most notable and the first rigorously proven efficient algorithm for the Ising model is the result of Jerrum and Sinclair \cite{Jerrum_ising} that uses a Markov chain Monte Carlo (MCMC) sampling algorithm to estimate the partition function in the ferromagnetic regime on arbitrary graphs. More generally, it has been shown that one can set up Markov chains for sampling from the Gibbs distribution that mix rapidly \emph{if} and \emph{only if} the correlations decay exponentially. This is known as the equivalence of mixing in time and mixing in space \cite{weitz2004mixing} .  

Another approach uses the decay of correlations in the Gibbs distribution. This property essentially allows one to decompose the interaction graph of the system into smaller computationally tractable pieces, and then combine the results of the computation on those pieces to find the overall partition function. In contrast to the MCMC approach, algorithms based on the decay of correlations can be deterministic -- even in the regime where no MCMC algorithm is known. This approach, for instance, has lead to efficient deterministic algorithms for the hard-core model up to the hardness threshold \cite{Weitz} and the antiferromagnetic Ising model \cite{Sinclair_antiferro}.

There is a recent conceptually different approach to estimating the partition function, which is the basis of this work. This approach views the partition function as a high-dimensional polynomial and uses the truncated Taylor expansion to extend the solution at a computationally easy point to a non-trivial regime of parameters. Since its introduction \cite{Barvinok_book}, this method has been used to obtain deterministic algorithms for various interesting problems such as the ferromagnetic and antiferromagnetic Ising models \cite{Liu_ising,Regts_ising} on bounded graphs.

As mentioned before, the equivalence of the analyticity of the free energy and the decay of correlations was first proved by Dobrushin and Shlosman \cite{Dobrushin1}. The Fisher zeros of the classical Ising model and their relation to the correlation decay was also recently studied in \cite{Liu_zeros}.  

\subsubsection{Quantum many-body systems}
The problem of estimating the partition function and correlation decay in quantum systems has also been studied in the past. We review some of these results here. 

There are various results (e.g., \cite{poulin2009sampling,chowdhury_gibbs_sampling}) that estimate the partition function by sampling from the Gibbs state using a quantum computer (also known as quantum Gibbs sampling). The best known bound on the running time of these algorithms is exponential in the number of particles. This running time can be reduced if we assume other conditions. For example, \cite{kastoryano_commuting} shows that a strong form of the decay of correlations implies an efficient quantum Gibbs sampler for commuting Hamiltonians. If in addition to the decay of correlations, we add the decay of quantum conditional mutual information, then this result can be extended to non-commuting Hamiltonians \cite{Brandao_gibbs_preparing}. Turning these quantum algorithms into classical ones results in an $n^{\polylog(n)}$ running time. Although we cannot directly compare these results with our algorithm due to different conditions that are imposed, the $n^{O(\log n)}$ running time that we achieve outperforms that of these algorithms.

Considering the success of approximation schemes for the classical statistical problems, it is desirable to import those results to evaluate the thermal properties of interacting quantum many-body systems. This indeed can be done for some models like the quantum transverse field Ising model \cite{bravyi_classicalmapping} or the quantum $\mathrm{XY}$ model \cite{Bravyi_ferro} in the ferromagnetic regime using what is called the \emph{quantum-to-classical mapping}.

Establishing the decay of correlations in the Gibbs state has also been studied in quantum settings. In particular, it has been shown that the Gibbs state has this property in the $1$D translationally invariant case \cite{araki_1d} or above some constant temperature in higher dimensions \cite{Kastoryano_locality}. Thus, in these regimes, there exist efficient representations for the state of the system using a tensor network ansatz like matrix product states or projected entanglement pair states \cite{Hastings_solving_gapped_locally,Kastoryano_locality,Cirac_gibbs_peps}. However, this does not necessarily imply an efficient algorithm that finds and faithfully manipulates these tensor networks.

The decay of conditional mutual information is another property of the Gibbs state that has been rigorously proved for $1$D systems \cite{Brandao_gibbs_cmi} and conjectured for higher dimensions. This result has been used to find algorithmic schemes for preparing the Gibbs state on a quantum computer \cite{Brandao_gibbs_preparing} or estimating the free energy in $1$D \cite{Kim_gibbs,gibbs_one_dim}. A recent result of \cite{Kohtaro_cmi_cluster} uses cluster expansions along with a technique very similar to the one we use in \thmref{informal decay from zero} (i.e. showing the low-order derivatives of the correlation function are zero) to establish the decay of conditional mutual information above some constant temperature. 

\subsection{Discussion and open questions}

Our work raises many questions that we leave for future work. Here we mention some of them.

\renewcommand\labelitemi{$\vcenter{\hbox{\tiny$\bullet$}}$}

\begin{enumerate}
\item
Perhaps the most immediate problem is to fully establish (or refute) the connection between the decay of correlations and the absence of zeros. There are at least two directions to pursue.

\begin{enumerate}
    \item It would be interesting to prove the exponential decay of correlations in the zero-free region of non-commuting Hamiltonians in higher dimensions. Currently we can only show this when the distance of the observables is $\Omega(\log n)$. It seems for this to work, the region of applicability of certain tools such as the quantum belief propagation needs to be extended to the complex regime.

    \item Establishing the absence of zeros in quantum systems when the correlations decay exponentially is also open. A first step might be to prove this for commuting Hamiltonians or $1$D chains. In \secref{Exponential decay of correlations gives the absence of complex zeros}, we have already extended some parts of the proof of this statement for the classical systems to commuting Hamiltonians, but it seems to complete the proof, a more careful analysis of the entangled boundary conditions is required.
\end{enumerate}

\item While we focus on the covariance form of the correlations \eqref{eq:a3}, one can also consider quantum conditional mutual information (qCMI) as a measure of correlations. Using the absence of zeros to prove the decay of qCMI is another interesting question. This would extend the result of \cite{Kohtaro_cmi_cluster} to lower temperatures down to the phase transition point. Since the approach of \cite{Kohtaro_cmi_cluster} resembles some of the techniques we use, this looks like a promising direction. 

\item Is there some range of temperatures or Hamiltonian parameters that a quantum computer cannot efficiently sample from the Gibbs state but the extrapolation technique still works? At least, when the parameter of interest is temperature, this depends on the fate of the previous questions we mentioned, i.e. showing that the decay of correlations and qCMI are necessary for the absence of zeros. The result of \cite{Brandao_gibbs_preparing} implies an efficient quantum sampler under the same conditions. Are there other parameters besides temperature for which one can show a separation between these notions?

\item Is it possible to improve the lower bound we obtained for the critical point $\b_c$ in \thmref{informal high t disk of no zeros} without using other conditions such as the decay of correlations? In general, what is the computational hardness of determining the thermal phase transition point $\b_c$? 

\item Can the running time of our algorithm be improved for specific systems to polynomial time? This has been achieved for the classical Ising model \cite{Liu_ising,Regts_ising} by relating the derivatives of the partition function to combinatorial objects that can be efficiently counted. 

\item Can we use the extrapolation idea to avoid the \emph{sign problem}? The easy regime, which includes the starting point of the extrapolation, could be a regime of parameters where the Hamiltonian is sign-free and MCMC algorithms yield a good estimate, whereas the end point is where the sign problem exists. A candidate parameter for extrapolation is the chemical potential. There are important physical systems such as lattice gauge theories for which at zero chemical potential the partition function is sign-free while there is a severe sign problem for non-zero chemical potentials.

\item Barvinok's approach has been used to obtain approximation algorithms for other problems related to quantum computing \cite{Mehraban_permanent,mann2018_approximation,bravyi2019_meanvalues}. Are there other relevent applications for this method?
\end{enumerate}

\section{Preliminaries and notation}
\subsection{Local and geometrically-local Hamiltonians}
Consider a $D$-dimensional lattice $\L \subset \mathbb{Z}^D$ containing $n$ sites with a $d$-dimensional particle (qudit) on each site. The Hilbert space is $\cH=\Ot_{i \in \L} \cH_i$ where $\cH_i$ is the local Hilbert space of site $i$. For a region $A\subseteq\L$, we denote its size by $|A|$ and its complement by $\bar{A}$. The diameter of $A$ is defined to be $\diam(A)=\sup\{\dist(x,y): x,y\in A\}$. The interaction of these particles is described by a local Hamiltonian $H$ that has the following form:
\ba\label{eq:local_decom}
H=\sum_{X\subset \L} H_X.
\ea
Each term $H_X$ is a Hermitian operator with operator norm at most $h$ that is acting non-trivially only on the sites in $X$. We denote this by writing $\supp(H_X)=X$. The local terms $H_X$ do not necessarily commute with each other. Similarly, we define $H_A=\sum_{X \subseteq A} H_X$ to be the Hamiltonian restricted to a region $A\subseteq \L$. We denote the number of local terms in the Hamiltonian by $m$ and often also write $H=\sum_{i=1}^m H_i$. The $1$-norm of an operator $O$ is denoted by $\norm{O}_1$ and its operator norm by $\norm{O}$. 

In order to impose geometric locality on the interactions between the particles, we consider the interactions that satisfy the following condition. 
\begin{definition}[Geometrically-local Hamiltonians]\label{def:geometrically-local Hamiltonians}
A Hamiltonian $H=\sum_{X\subset \L} H_X$ such that $\supp(H_X)=0$ when $\diam(X)>R$ or $|X|> \k$ is called a $(\k,R)$-local Hamiltonian. We call $\k$ the locality and $R$ the range of $H$. We use the words geometrically-local and $(\k,R)$-local interchangeably when $\k,R$ are kept constant. 
\end{definition}
This should be contrasted with the case where $\supp(H_{X})=0$ when $|X|>\k$ but there is no restriction on $\diam(X)$. In order to distinguish between these two, we use the terms \emph{geometrically-local} versus \emph{local} throughout this paper. We also focus mostly on geometrically-local Hamiltonians with a \emph{finite range} $R$, but most of our results also apply to Hamiltonians with interactions that decay fast enough, for example, with some exponential rate. 
\begin{rem}
In general, the locality $\k$ of a geometrically-local Hamiltonian on a $D$-dimensional lattice $\L$ can be bounded as $\k\leq O(D^{R})$, which is the size of a ball of diameter $R$. Nevertheless, we treat both $\k$ and $R$ as independent parameters in this paper.
\end{rem}

For the Hamiltonians we consider, the sum $|\sum_{X\cap \{x_0\} \neq \emptyset}H_X|$ is bounded from above by a constant like $O(hD^{\k R})$, but in general, this is a loose bound and we introduce the \emph{growth constant} as an independent parameter such that:
 
\begin{definition}[Growth constant]\label{def:growth constant}
Given a geometrically-local Hamiltonian $H$, the growth constant $g$ is defined such that $|\sum_{X\cap \{x_0\} \neq \emptyset}H_X|\leq gh$ for all sites $x_0\in \L$. 
\end{definition}
 
Given a $(\k,R)$-local Hamiltonian $H$, we denote the boundary of a region $A \subseteq \L$ by $\partial A$ and define it as $\partial A=\{v\in \L\setminus A: \exists v'\in A,\ \dist(v-v')\leq R\}$. Defined this way, the boundary of $A$ is a subset of $\bar{A}$.

For local Hamiltonians with $\k=2$, we define an \emph{interaction graph} which is an undirected graph $G=(V,E)$ with a qudit on each vertex $v\in V$ and an edge $(i,j)$ between any two vertices $i,j$ that are acted on by a local term in the Hamiltonian. For qubits, $d=2$ and such a Hamiltonian is of the following form:

\ba\label{eq:u2_local}
H=-\sum_{\substack{(i,j)\in E\\a,b\in\{x,y,z\}}} J_{ij}^{ab} \s_a \ot \s_b- \sum_{\substack{i\in V \\a\in \{x,y,z\}}} h^{a}_i \s_a,
\ea
where $J_{ij}^{ab}, h_i^a \in \bbR$ are the interaction coefficients and $\s_a\in\{X,Y,Z,\iden\}$ are Pauli matrices.

\subsection{Quantum thermal state and partition function}
The free energy of state $\r$ at inverse temperature $\b$ is defined as 
\ba
F_{\b}(\r)=\Tr(H\r)-\frac{1}{\b}S(\r)\nonumber,
\ea
where $S(\r)=-\Tr(\r\log \r)$ is the von Neumann entropy of $\r$ (here and throughout this paper, we assume $\log$ denotes the natural logarithm). In thermal equilibrium, the free energy of the system is minimized. Using the non-negativity of the relative entropy $S(\r \|\frac{e^{-\b H}}{Z(\b)})\geq 0$, one can see that 
\ba \label{eq:convex_free_energy}
\min_{\r} F_{\b}(\r)&=\min_{\r}\Tr(H\r)-\frac{1}{\b}S(\r)\\\nonumber
&=-\frac{1}{\b}\log(Z_{\b}(\L)), 
\ea
where $Z_{\b}(H)=\Tr[\exp(-\b H)]$ is the partition function of the system at inverse temperature $\b$. When dealing with spin systems on a lattice, we often denote the partition function of the system by $Z_{\b}(\L)$ rather than $Z_{\b}(H)$. 

Furthermore, the state that achieves this minimization, known as the Gibbs (or thermal) state, is given by
\ba
\r_{\b}(H)=\frac{\exp(-\b H)}{Z_{\b}(H)}.
\ea
We often need to consider the Gibbs state after some measurement has been performed on a local region of the lattice. The post-selected state $\r_{\b}(H|N)$ associated with a positive operator $N$ is given by
\ba
\r_{\b}(H|N)=\frac{\sqrt{N}\exp(-\b H)\sqrt{N}}{\Tr[\exp(-\b H)N]}.
\ea
\subsection{Quantum belief propagation}
Suppose certain local terms in Hamiltonian $H$ are removed and consider the Gibbs state before and after this change. We would like to relate these Gibbs states by applying a local operator on the old state to obtain the new one. This has been addressed before in \cite{hastings_belief_propagation} under the name \emph{quantum belief propagation}. We only mention this result without the proof and refer the reader to \cite{hastings_belief_propagation,Brandao_gibbs_cmi} for the derivation and more details. 

\begin{prop}[Quantum belief propagation]\label{prop:Quantum belief propagation}
Let $H$ be a geometrically-local Hamiltonian on lattice $\L$. Consider a sublattice $C\subset \L$. We denote the terms in $H$ acting on both $C$ and $\bar{C}$ by $H_{\partial C}$.
There exists a quasi-local operator $\eta$ such that 
\ba\label{eq:p0}
e^{-\b H}=\eta e^{-\b (H-H_{\partial C})} \eta ^{\dag},
\ea
where $\norm{\eta} \leq \exp(\beta/2 \norm{H_{\partial C}})$. Moreover, there exists a truncation of $\eta$ denoted by $\eta_{\ell}$ supported non-trivially only on $\partial C$ and sites within distance $\ell$ from $\partial C$ such that for some positive constants $\a_1,\a_2$,
\ba\label{eq:p1}
\Norm{\eta - \eta_{\ell}} \leq e^{\a_1 |\partial C|-\a_2 \ell}.
\ea
\end{prop}

\subsection{Tools from complex analysis}\label{sec:interpolation}
Given a function that is analytic in a region of the complex plane, i.e. it is complex differentiable, we are interested in approximating the function in that region with a low-degree polynomial. Conventional methods in complex analysis allow us to achieve this using a Taylor expansion around a point inside that region. 

\begin{definition}[Taylor expansion of analytical functions]
Given a complex function $f(z)$ that is analytic in a region $A$, the Taylor expansion of $f(z)$ around a point $z_0 \in A$ is a power series $\sum_{k=0}^{\infty} a_k (z-z_0)^k$, where for $\forall k=0,1,\dots$
\ba
a_k=\frac{1}{k!}\frac{d^kf(z_0)}{dz^k}=\frac{1}{2\pi i}\oint_C \frac{f(w)}{(w-z_0)^{k+1}}dw\label{eq:4}
\ea
for an arbitrary contour $C$ around $z_0$ inside the region $A$. 
\end{definition}

In \secref{Extrapolating from high external fields and Lee-Yang zeros}, we map the partition function of a quantum system to that of a classical system. As we increase the precision of the mapping, we get a family of classical systems with increasing size that in the limit of an infinite number of particles have the same partition function as the quantum system. The following theorem guarantees that the zero-free region of the classical ensemble coincides with that of the original quantum system. 
\begin{thm}[Multivariate Hurwitz's theorem]\label{thm:Multivariate Hurwitz' theorem}
If a sequence of multivariate functions $f_1,f_2,f_3,\dots$ are analytic and non-vanishing on a connected open set $D\subset \bbC^n$ and converge uniformly on compact subsets of $D$ to $f$, then f is either non-vanishing on $D$ or is identically zero.
\end{thm}

The proof can be found in standard complex analysis textbooks \cite{gamelin_book_complex}.

\section{Algorithm for estimating the partition function}\label{sec:Algorithm for estimating the partition function}
In this section, we provide more details about the approximation algorithm that we presented in \secref{intro}. 

\begin{definition}
An approximation algorithm for the partition function $Z_{\b}(H)$ takes as input the description of the local Hamiltonian $H$, the inverse temperature $\b$, and a parameter $\e$ and gives an estimate $\tilde{Z}_{\b}(H)$ with $\e$-multiplicative error, i.e.
\ba
\left|\tilde{Z}_{\b}(H)-Z_{\b}(H)\right|\leq \e Z_{\b}(H).
\ea
This is, up to unimportant constants, equivalent to finding an $\e$-additive error for $\log Z_{\b}(H)$ or $F_{\b}(H)$.
\end{definition}

We now make a connection between analyticity of functions and approximation algorithms precise. Similar propositions were first proved by \cite{Barvinok_book} for bounded degree polynomials.

Suppose we want to estimate the value of a complex function $f(z)$. We consider two cases. One is when there is an upper bound on the absolute value of the function in the region that $f(z)$ is analytic. The other is when the given function is $f(z)=\log(g(z))$ for a polynomial $g(z)$ of degree $n$. The latter is used in \secref{An algorithm for the anisotropic xxz model} when studying the $\mathrm{XXZ}$ model. We need the former version since as we will see in \thmref{extrapolation algorithm}, the partition function of quantum (or even some classical) systems is not always a polynomial in $\b$. 
\begin{prop}[Truncated Taylor series for bounded functions and polynomials]\label{prop:taylor_bounded} We denote a disk of radius $b$ centered at the origin in the complex plane by $\D_b$, that is $\D_b=\{z\in\bbC: |z|\leq b\}$.
\begin{itemize}
    \item[(1)] Let $f(z)$ be a complex function that is analytic and bounded as $|f(z)|\leq M$ when $z\in\D_b$ for a constant $b>1$. Then the error of approximating $f(z)$ by a truncated Taylor series of order $K$ for all $|z|\leq 1$ is bounded by 
\ba
\left|f(z)-\sum_{k=0}^{K} a_k z^k\right|\leq \frac{M}{b^K(b-1)},\quad |z|\leq 1.\label{eq:o2}
\ea

\item[(2)] Assume $b$ is fixed and there is a deterministic algorithm that finds the coefficients $a_k$ in time $O(N^k)$ for some  parameter $N$. Then there exists a deterministic algorithm with running time $N^{O(\log(M/\e))}$ that outputs an $\e$-additive approximation for $f(z)$.

\item[(3)] [cf. \cite{Barvinok_book}] Let $f(z)=\log(g(z))$ for some polynomial $g(z)$ of degree $N$ that does not vanish when $z\in \D_b$. The error of approximating $f(z)$ by a truncated Taylor series of order $K$ for $|z|\leq 1$ is bounded by $\frac{N}{K+1}\frac{1}{b^K(b-1)}$.
\item[(4)] [cf. \cite{Barvinok_book}] Assuming $b$ is fixed, there exists a deterministic algorithm with running time $N^{O(\log(N/\e))}$ that outputs an $\e$-additive approximation for $\log(g(z))$.
\end{itemize}
\end{prop}

\begin{proof}
The proof of (1) is a basic result in complex analysis based on the Cauchy integral theorem for analytic functions. Let $C'$ be the circle $|z|=b$ that contains both $z$ and $z=0$. We have
\ba
f(z)=&\frac{1}{2\pi i}\oint_{C'} \frac{f(w)}{w-z} dw=\frac{1}{2\pi i}\oint_{C'}\frac{f(w)}{w} \left(1-\frac{z}{w}\right)^{-1} dw \nonumber\\
&=\frac{1}{2\pi i}\oint_{C'} \frac{f(w)}{w} \left( \sum_{k=0}^{K}\left(\frac{z}{w}\right)^k+\left(\frac{z}{w}\right)^{K+1}\left(1-\frac{z}{w}\right)^{-1}\right) dw\nonumber \\
&=\sum_{k=0}^K\frac{f^{(k)}(0)}{k!}z^k+\frac{1}{2\pi i}\oint_{C'} \frac{f(w)}{w-z}\left(\frac{z}{w}\right)^{K+1}dw\nonumber,
\ea
in which we used Eq. \eqref{eq:4} to get to the last line. We can now bound the remainder as
\ba
\left|f(z)-\sum_{k=0}^K\frac{f^{(k)}(0)}{k!}z^k\right|&\leq \frac{1}{2\pi}\oint_{C'} \frac{|f(w)|}{|w-z|}\left(\left|\frac{z}{w}\right|\right)^{K+1}dw.\nonumber\\
&\leq M\frac{b}{b-1}\left(\frac{1}{b}\right)^{K+1},
\ea
where the last line follows from the fact that $|w-z|\geq b-1$, $|z|\leq 1$, and $|f(w)|\leq M$ on $C'$.

The proof of part (3) is similar to that of (1). The degree $N$ polynomial $g(z)$ has at most $N$ complex roots $\{\z_k\}_{k=1}^N$ such that $|\z_k|\geq b$. Thus, $g(z)=g(0)\prod_{l=1}^{N}(1-\frac{z}{\z_l})$ and
\ba
\log(g(z))=\log(g(0))+\sum_{l=1}^{N} \log\left(1-\frac{z}{\z_l}\right),\quad \forall z:\ |z|\leq 1.
\ea
We can expand each term like $\log(1-\frac{z}{\z_l})$ as 
\ba\label{eq:o1}
\log\left(1-\frac{z}{\z_l}\right)=-\sum_{k=1}^K \frac{z^k}{k\z_l^k}+q_{\ell}(z),
\ea
where similar to  part (1), we see that $q_{\ell}(z)$ is a term that can be bounded by
\ba
|q_{\ell}(z)|\leq \frac{1}{K+1}\frac{1}{b^K(b-1)}.
\ea
Hence, the remainder term in the Taylor expansion of $\log(g(z))$ up to order $K$ is $q(z)=\sum_{{\ell}=1}^{N} q_{\ell}(z)$, which is bounded by $|q(z)|\leq \frac{N}{K+1}\frac{1}{b^K(b-1)}$ as claimed in part (3).

In order to find the algorithms of part (2) and (4), we need to evaluate the Taylor coefficients of $f(z)$ up to some degree $K$. Since we want an $\e$-additive approximation of $f(z)$, one can see from parts (1) and (2) that it is sufficient to keep the Taylor expansion until order $K=O(\log(\frac{M}{\e}))$ for part (2) and $K=O(\log(\frac{N}{\e}))$ for part (4). To be able to evaluate the derivatives $\frac{d^k f(z)}{dz^k}$, we express them in terms of the derivatives of $g(z)$, i.e. $\frac{d^k g(z)}{dz^k}$ \footnote{We are using the same definition $f(z)=\log(g(z))$ for the function in part (1) as well.}. This can be done by noticing that 
\ba
\frac{d^k g(z)}{dz^k}=\sum_{\ell=0}^{k-1}\binom{k-1}{\ell}\frac{d^{\ell} g(z)}{dz^{\ell}}\frac{d^{k-\ell}f(z)}{dz^{k-\ell}},
\ea
so if we have access to $\frac{d^k g(z)}{dz^k}$, we can find $\frac{d^k f(z)}{dz^k}$ by solving the system of equations in time $\poly(k)$. The important step, however, is to estimate $\frac{d^k g(z)}{dz^k}$. This by assumption takes time $N^{O(k)}$ for the $k$th derivative. Thus, evaluating the Taylor expansion in parts (2) and (4) can be done in time $N^{O(\log(M/\e))}$ and $N^{O(\log(N/\e))}$, respectively. 
\end{proof}

\begin{thm}[Extrapolation algorithm for estimating the partition function]\label{thm:extrapolation algorithm}
There exists a deterministic classical algorithm that runs in time $n^{O(\log(n/\e))}$ and outputs an estimate within $\e$-multiplicative error of the partition function $Z_{\b}(H)$ at some constant $\b$ in the zero-free region  $\Omega_{\d,\b}$ (see \defref{informal absence of zeros}).
\end{thm}
\begin{proof}[\textbf{Proof of \thmref{extrapolation algorithm}}]
We apply the truncated Taylor expansion. To use that result, we first need to specify the zero-free region and then bound the running time of computing the $k$th derivative by $n^{O(k)}$.

We can without loss of generality assume that the zero-free region $\Omega_{\d,\b}$ is a rectangular region of constant width and size depicted in \fig{zeros_fisher}. The result of \propref{taylor_bounded}, however, holds when the zero-free region is a disk of radius $b$. To match these domains, we can compose the partition function with a function $\phi(z)$ that maps a disk of radius $b$ to the rectangular region $\Omega_{\d,\b}$ such that $\phi(0)=0$ and $\phi(1)=\b$ and $b$ is constant depending on $\d$. It is shown in Lemma 2.2.3 of \cite{Barvinok_book} that one can find such a $\phi(z)$ which is a constant degree polynomial. Hence, the composed partition function is non-zero and bounded on this disk and we can apply the bound \eqref{eq:o2} on the Taylor expansion. 

As mentioned in \secref{Sketch of our techniques}, for a system of $n$ qudits, we can compute the order $k$ derivatives of $Z_{\b}(H)$ in time $n^{O(k)}$. Similarly, we can evaluate the derivatives of $Z_{\b}(H)$ composed with the constant-degree polynomial $\phi(z)$ using the same running time. Keeping only $k=O(\log(n/\e))$ many terms results in a quasi-polynomial algorithm with multiplicative error $\e$.
\end{proof}
\section{Lower bound on the critical inverse temperature}\label{sec:High temperatures: Fisher zeros}
In this section, we show that at high temperatures, there are no complex zeros near the real axis. More precisely, we prove that there exists a disk of constant radius $\b_0$ centered at $\b=0$ that does not contain any zeros and the free energy is analytic inside it. The radius $\b_0$ depends only on the geometric parameters of the Hamiltonian such as the growth constant.

\begin{thm}[High temperature zeros]\label{thm:t_higH_temp}
Let $H$ be a gometrically-local Hamiltonian on qudits with range $R$, growth constant $g$, and local interactions with norm at most $h$ (see \defref{geometrically-local Hamiltonians} and \defref{growth constant}). There exists a real constant $\b_0= 1/(5egh\k)$ such that for all $\b\in \bbC$ with $|\b| \leq \b_0$, the partition function $Z_{\b}(\L)$ of $H$ does not vanish and $\log(Z_{\b}(\L))$ is analytic and bounded by $\big|\log|Z_{\b}(\L)|\big|\leq (e^2gh|\b|+\log d)n$.
\end{thm}
This gives a lower bound $\b_0\leq \b_c$ on the phase transition point $\b_c$. Also, as outlined in \thmref{extrapolation algorithm}, if we can establish an upper bound like $\big|\log|Z_{\b}(\L)|\big|\leq O(n)$ for small enough complex $\b$, we can devise an approximation algorithm for the partition function. Hence we get
\begin{cor}[Approximation algorithm for the partition function at high temperatures]\label{thm:alg_for_higH_temp}
There exists a quasi-polynomial time algorithm with running time $n^{O(\log(n/\e))}$ that outputs an $\e$-multiplicative approximation to the partition function $Z_{\b}(\L)$ of a geometrically-local Hamiltonian $H$ when $|\b|\leq \b_0$.
\end{cor}
Before getting to the proof of \thmref{t_higH_temp}, we need to gather some facts and lemmas. Given a lattice $\L\subset \mathbb{Z}^D$ with $n$ sites, we consider a series of sublattices $\L_0\subset \L_1\subset \L_2\subset \L_2\subset \dots \subset \L_n=\L$ such that each sublattice $\L_j$ has one fewer vertex than $\L_{j+1}$ and $\L_0=\emptyset$. The partition function of $\L_0$ is assigned to be $Z_{\b}(\emptyset)=1$ for any complex $\b$. Therefore, we can write 
\ba\label{eq:t2}
Z_{\b}(\L)=d^n \prod_{j=0}^{n-1}\left( \frac{1}{d} \frac{Z_{\b}(\L_{j+1})}{Z_{\b}(\L_{j})}\right),
\ea
where the factors of $d$ are added for later convenience and to account for the dimension of the removed sites. In order to show $\left|\log \left|Z_{\b}(\L)\right|\right|\leq O(n)$ for a $\b \in \bbC$, we just need to bound the logarithm of each of the terms in Eq. \eqref{eq:t2} by a constant, i.e. 
\ba \label{eq:t0}
\left|\log \Big|\frac{1}{d}\frac{Z_{\b}(\L_{j+1})}{Z_{\b}(\L_j)}\Big|\right|\leq O(1).
\ea
This bound tells us how much the partition function changes after removing a single site from the lattice. We later prove this by induction on the number of sites. However, as shown in the following lemma, this inequality is always satisfied when $\b$ is \emph{real}.
\begin{lem}[Site removal bound]\label{lem:site_removal}
The following bound holds for any  $X\subseteq\L$ and $\b \in \bbR^+$:
\ba\label{eq:t31}
\left|\log \Big|\frac{1}{d^{|X|}}\frac{Z_{\b}(\L)}{Z_{\b}(\L\setminus X)}\Big|\right|\leq gh|\b||X|.
\ea
Recall that $h$ is the maximum norm of the local terms $H_X$ in $H$ and the growth constant $g$ is chosen such that $|\sum_{X\cap \{x_0\} \neq \emptyset}H_X|\leq gh$ for all sites $x_0\in \L$. 
\end{lem}
\begin{proof}
We have  
\ba
Z_{\b}(\L)=\Tr_{\L}\big[e^{-\b (H_{\L\setminus X}+\sum_{X'\subset \L:X'\cap X \neq \emptyset} H_{X'})}\big] &\leq \Tr_{\L} \big[e^{-\b H_{\L\setminus X}} e^{-\b \sum_{X'\subset \L:X'\cap X \neq \emptyset} H_{X'}} \big]\nonumber \\
&\leq \Tr_{\L} \big[e^{-\b H_{\L\setminus X}} \big] \Norm{e^{-\b \sum_{X'\subset \L:X'\cap X \neq \emptyset}H_{X'}}}  \nonumber \\
&\leq d^{|X|}\Tr_{\L \setminus X} \big[e^{-\b H_{\L\setminus X}} \big] e^{ \Norm{\b \sum_{X'\subset \L:X'\cap X \neq \emptyset}H_{X'}} } \nonumber \\
&\leq Z_{\b}(\L\setminus X) d^{|X|} e^{gh|\b||X|},\label{eq:t18}
\ea
where $H_{\L \setminus X}$ corresponds to the terms in the Hamiltonian acting on the remaining sublattice $\L \setminus X$. We used the Golden-Thompson inequality in the first line and the H\"older inequality to get to the second line. The factor $d^{|X|}$ is added since the original trace is over the Hilbert space of $\L$ and not $\L \setminus X$. Similarly, one can show $d^{|X|} Z_{\b}(\L\setminus X)\leq Z_{\b}(\L) e^{g h|\b||X|}$. These bounds together prove the lemma.
\end{proof}

\thmref{t_higH_temp} extends bound \eqref{eq:t31} to the case where $\b$ is a small complex number. We prove this in two steps.

\paragraph{First step:} In contrast to the proof of \lemref{site_removal}, the Golden-Thompson inequality can no longer be used in the complex regime. Hence, to compare the partition function before and after removing a site $x_0$, we need to find another way of separating the contribution of the terms in the Hamiltonian that act on $x_0$. We achieve this using a \emph{cluster expansion} for the partition function that expands the operator $\exp (-\b H)$ \footnote{Note that this is different from Taylor expanding $\log (\exp(-\b H))$, which is our eventual goal.} into a sum of products of local terms in $H$. The idea of using cluster expansions to study high temperature properties of classical or quantum spin systems has been widely applied before \cite{kotecky_cluster,Dobrushin_estimates,Park_cluster,greenberg_cluster}. Here, we use a particular version of that expansion which is tailored for our application. This was first introduced in \cite{Hastings_solving_gapped_locally} and later improved and generalized in \cite{Kastoryano_locality}. In \secref{The cluster expansion for the partition function}, we modify the result of \cite{Hastings_solving_gapped_locally,Kastoryano_locality} and adapt it for complex partition functions.
    
\paragraph{Second step:} Our next step is to use the cluster expansion and show that in the partition function, the contribution of the sites acting on $x_0$ compared to the rest of the terms is bounded by a constant. We show this in \secref{Bounding the correction terms in the expansion} by induction on the number of sites. This is our main contribution and lets us prove the bound \eqref{eq:t31}.

\subsection{The cluster expansion for the partition function}\label{sec:The cluster expansion for the partition function}
When using the cluster expansion, we often need to consider products of local terms like $\prod_{j=1}^{\ell} H_{X_j}$, but since the local interaction terms $H_{X_j}$ do not necessarily commute with each other, we set an $\ell$-tuple $(X_1,\dots, X_{\ell})$ to indicate the order of multiplication. We also need to decompose the sequence $(X_1,\dots,X_{\ell})$ into the union of \emph{connected} components. Let us define what we mean by connected more formally.

\begin{definition}[Connected sets]\label{def:connected_set}
Fix a site $x_0 \in \L$. A collection of sublattices such as $\cX=\{X_1,X_2,\dots, X_k\}$ is called a connected set containing $x_0$ with size $|\cX|=k$ if the following conditions hold
\begin{itemize}
    \item[i)] All the sublattices $X_1,X_2,\dots, X_k$ have bounded size and diameter. That is $1\leq |X_i|\leq \k$ and $\diam(X_i)\leq R$. 
    \item[ii)] For any sublattice $X_i$ in $\cX$, a series of other members of $\cX$ connect this set to the site $x_0$. See \fig{connected_set_fig} for an example. More precisely we have: for any $X_i \in \cX$, there exists $I\subseteq [k]$ such that $i\in I$ and  $\forall j\in I, \exists \ell \in I: X_j \neq X_{\ell}$ yet $X_j \cap X_l \neq \emptyset$, and moreover, $x_0 \in \cup_{j\in I}X_j$. 
\end{itemize} Although $\cX$ consists of sublatices of $\L$ and not individual sites, in a slight abuse of notation, we specify a set $\cX$ that contains the site $x_0$ by $x_0 \in \cX$. We denote all the sites that a connected set $\cX$ includes by $\supp(\cX)$.
\end{definition}

\begin{figure}[t!]
\centering
	\includegraphics[width=.25
\textwidth]{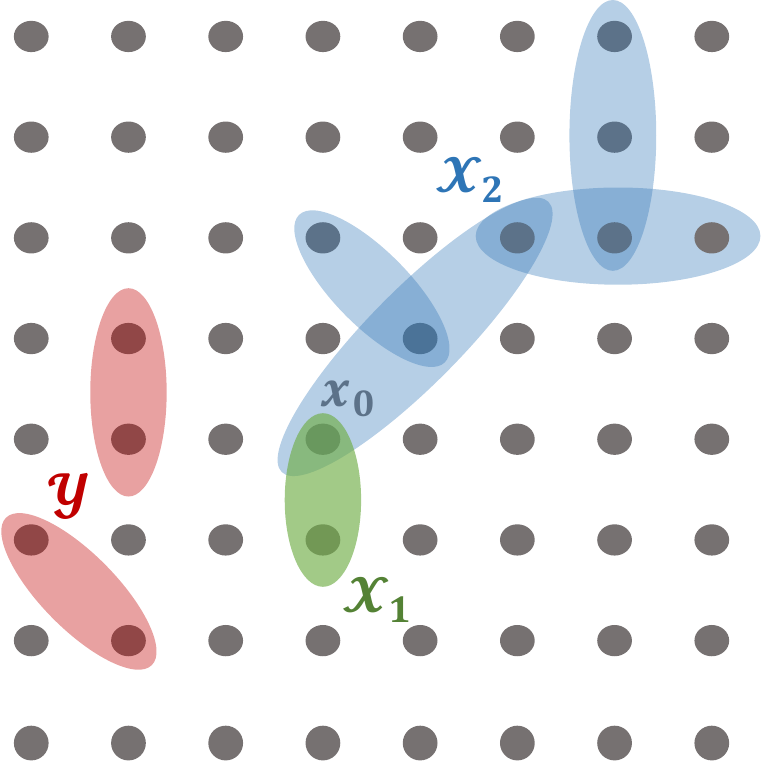}
	\caption{The sets $\cX_1, \cX_2$ are connected, contain $x_0$, and have size $|\cX_1|=1, |\cX_2|=4$. However, the set $\cY$ is not connected, does not include $x_0$, and has size $|\cY|=2$.}\label{fig:connected_set_fig}
\end{figure}
\begin{rem}
In \defref{connected_set}, we include an upper bound on the size and diameter of the subsets in $\cX$, i.e. $|X_i|\leq \k, \diam(X_i)\leq R$. This is because for geometrically-local Hamiltonians, $\norm{H_X}=0$ for $|X_i|> \k,\ \diam(X_i)> R$, so we do not need to consider those sets. 
\end{rem}

In the upcoming proofs, we need to have an upper bound on the number of the connected sets $\cX$ that contains a specific site $x_0 \in \L$. This is stated in the following lemma. 
\begin{lem}[Cf. \cite{Kastoryano_locality}]\label{lem:bound_connected}
The number of connected sets $\cX$ of size  $|\cX|$ containing the site $x_0 \in \L$ is upper bounded by $g^{|\cX|}$ where $g$ is the growth constant of the Hamiltonian $H$ (see \defref{growth constant}). In particular, for a D-dimensional lattice and $\k=2$, we have $g\leq 2eD$.  
\end{lem}
The next lemma achieves the first step in our proof by setting up the cluster expansion for the partition function. 

\begin{lem}[High temperature expansion]\label{lem:higH_temp_expansion}
For any $x_0\in \L$, the partition function of the lattice $Z_{\b}(\L)$ admits the following decomposition for $|\b| \leq \frac{1}{gh(e-1)}$:
\ba
Z_{\b}(\L)=d\cdot Z_{\b}(\L \setminus X_0)+\sum_{\substack{\cX: x_0 \in \cX\\\cX \text{\ is connected\ }}} W_{\b}(\cX)Z_{\b}(\L\setminus \supp(\cX)),\label{eq:t_cluster_expansion}
\ea
where $X_0=\{x_0\}$ and we define $W_{\b}(\cX)$ as
\ba
W_{\b}(\cX)=\sum_{p=|\cX|}^{\infty} \frac{(-\b )^{p}}{p!}\big(\sum_{\substack{(X_1,\dots, X_p) \\  \forall i\in [p]: X_i \in \cX\\ \cX = \cup_{i=1}^p \{X_i\}}} \Tr_{\supp(\cX)}\big[\prod_{j=1}^p  H_{X_j}\big]\big).\label{eq:t32}
\ea
The last sum in \eqref{eq:t32} is over all p-tuples $(X_2,X_2,\dots,X_p)$ that one can form from members of $\cX$ by repeating them at least once.
\begin{proof}[\textbf{Proof}]
We start by Taylor expanding $\exp (-\b H)$. We have 
\ba
Z_{\b}(\L)&=\Tr_{\L} \big[\sum_{k=0}^{\infty} \frac{(-\b )^k}{k!}(\sum_{X\subset \L} H_X)^k\big]\nonumber\\
&=\Tr_{\L}\big[\sum_{k=0}^{\infty} \frac{(-\b )^k}{k!}(\sum_{X\subset \L \setminus X_0}H_X + \sum_{X \subset \L: X \cap X_0 \neq \emptyset} H_X)^k\big]\nonumber\\
&=\Tr_{\L}\big[\sum_{k=0}^{\infty} \frac{(-\b )^k}{k!}(\sum_{X\subset \L \setminus X_0} H_X)^k\big]+\Tr_{\L} \big[\sum_{\ell=1}^{\infty} \frac{(-\b)^{\ell}}{\ell !} \sum_{\substack{(X_1,\dots, X_{\ell})\\ \forall i\in [\ell]: X_i \subset \L \\ \exists X_i: X_i\cap X_0 \neq \emptyset}}\prod_{j=1}^{\ell} H_{X_j}\big], \label{eq:t4}
\ea
where the trace is over the Hilbert space of $\L$ as usual. The first term in the last line is just the Taylor expansion of $d\cdot Z_{\b}(\L\setminus X_0)$. As in Eq. \eqref{eq:t18}, the factor $d$ is included because the original trace is over $\L$ and not $\L\setminus X_0$. The last term, however, does not have a closed form, and involves summing over all the products of the local interaction terms $H_{X_j}$ such that at least one of the terms has non-empty overlap with the site $x_0$. We can simplify this term by partitioning the sequence $(X_1,\dots,X_l)$ into two parts. The first part forms a connected set $\cX$ that contains the site $x_0$. The second part contains all $X_i$ that do not intersect with this connected set $\cX$. We then change the order of the summation in \eqref{eq:t4} by first summing over all $X_i$ not connected to a fixed $\cX$ and then varying the set $\cX$. We get

\ba
& \sum_{\ell=1}^{\infty} \frac{(-\b )^{\ell}}{\ell !}\sum_{\substack{(X_1,\dots, X_l)\\ \forall i\in [\ell]: X_i \subset \L \\ \exists X_i: X_i\cap X_0 \neq \emptyset}} \Tr_{\L}\big[\prod_{j=1}^{\ell} H_{X_j}\big]\nonumber \\ &=\sum_{p=|\cX|,q=0}^{\infty} \binom{p+q}{p} \frac{(-\b )^{p+q}}{(p+q)!}\sum_{\substack{(X_1,\dots, X_p) \\  \forall i\in [p]: X_i \in \cX\\ \cX = \cup_{i=1}^p \{X_i\}}} \Tr_{\supp(\cX)}\big[\prod_{j=1}^p  H_{X_j}\big] \sum_{\substack{(X_{p+1},\dots, X_{p+q}) \\ X_{p+i} \cap \supp(\cX) = \emptyset}} \Tr_{\L\setminus\supp(\cX)}\big[\prod_{j=p+1}^{p+q} H_{X_j} \big]\Big).\nonumber
\ea
The coefficient $\binom{p+q}{p} \frac{(-\b )^{p+q}}{(p+q)!}$ in the second line counts the number of ways we can distribute our choices of $X_i\in \cX$ inside the tuple $(X_1,\dots,X_{p+q})$ and is equal to $\frac{(-\b )^{p}}{(p)!}\frac{(-\b )^{q}}{(q)!}$. The last sum in the right side term vanishes for $q=0$. We can restate this sum in terms of the Taylor expansion of $Z_{\b}(\L \setminus \supp(\cX))$. This gives us the following equality
\ba
&\sum_{\ell=1}^{\infty} \frac{(-\b )^{\ell}}{\ell !}\sum_{\substack{(X_1,\dots, X_{\ell})\\ \forall i\in [\ell]: X_i \subset \L \\ \exists X_i: X_i\cap X_0 \neq \emptyset}} \Tr_{\L}\big[\prod_{j=1}^{\ell} H_{X_j}\big]\nonumber\\ &=\sum_{\substack{\cX: x_0 \in \cX\\\cX \text{\ is connected\ }}}\sum_{p=|\cX|}^{\infty} \frac{(-\b )^{p}}{p!}\big(\sum_{\substack{(X_1,\dots, X_p) \\  \forall i\in [p]: X_i \in \cX\\ \cX = \cup_{i=1}^p \{X_i\}}}\Tr_{\supp(\cX)}\big[\prod_{j=1}^p  H_{X_j}\big]\big) Z_{\b}\big(\L \setminus \supp(\cX)\big) \nonumber\\
&=\sum_{\substack{\cX: x_0 \in \cX\\\cX \text{\ is connected\ }}} W_{\b}(\cX)Z_{\b}(\L\setminus \supp(\cX)),\label{eq:t7}
\ea
which by plugging into Eq \eqref{eq:t4} gives us the expansion  \eqref{eq:t_cluster_expansion}. Note that since we manipulated infinite series, we still need to prove the convergence of the expansion \eqref{eq:t_cluster_expansion} for small enough complex $\b$. We show the \emph{absolute convergence} of this expansion by first bounding the infinite series $W_{\b}(\cX)$ and then the expression \eqref{eq:t7}. A similar expansion for a different purpose has been considered before in \cite{Hastings_solving_gapped_locally,Kastoryano_locality} where an upper bound for $W_{\b}(\cX)$ is obtained. In particular, Lemma 5 in \cite{Kastoryano_locality} implies \footnote{Note that compared to \cite{Kastoryano_locality} we pick up the extra factor $d^{|\supp(\cX)|}$ when bounding $\big|\Tr\big[\prod H_{X_j}\big]\big|$.}
\ba
|W_{\b}(\cX)|\leq d^{|\supp(\cX)|} \left(e^{|\b| h}-1\right)^{|\cX|}\label{eq:t12}.
\ea
By using the result of \lemref{bound_connected}, we see that

\ba\label{eq:t15}
\sum_{\substack{\cX: x_0 \in \cX \\\cX \text{\ is connected\ }}}\big|W_{\b}(\cX)\big|\cdot \big|Z_{\b}(\L\setminus \supp(\cX))\big|&\leq d^{n} e^{g h|\b|n} \sum_{|\cX|=1}^{\infty} g^{|\cX|} \left(e^{|\b| h}-1\right)^{|\cX|},
\ea
in which we used the upper bound $|Z_{\b}(\L\setminus \supp{\cX})| \leq d^{n-|\supp(\cX)|}e^{gh|\b|n}$ that can be shown using the H\"older inequality. This right-hand side of the inequality \eqref{eq:t15} is finitely bounded when 
\ba
g (e^{|\b|h}-1)\leq 1,
\ea

which along with the inequality $e^x\leq 1+(e-1)x$ implies an upper bound on the size of the admissible $\b$ 
\ba\label{eq:t16}
|\b|\leq \frac{1}{gh(e-1)}.
\ea
Hence, we get 
\ba\label{eq:t13}
\sum_{\substack{\cX: x_0\in \cX \\\cX \text{\ is connected\ }}}\big|W_{\b}(\cX)\big|\cdot \big|Z_{\b}(\L\setminus \supp(\cX))\big|
&\leq d^n e^{gh|\b|n} \frac{g (e^{|\b|h}-1)}{1-g (e^{|\b|h}-1)},
\ea
which for a fixed $n$, shows the absolute convergence of \eqref{eq:t_cluster_expansion} and completes the proof of the lemma.\qedhere
\end{proof}
\end{lem}

Having this lemma, we can now proceed to the second step of our proof of \thmref{t_higH_temp}.

\subsection{A zero-free region at high temperatures}\label{sec:Bounding the correction terms in the expansion}

\begin{proof}[\textbf{Proof of \thmref{t_higH_temp}}]\label{proof:t_higH_temp}
As explained in the beginning of \secref{High temperatures: Fisher zeros}, to show the partition function does not vanish for small enough $|\b|$, and moreover $|\log Z_{\b}(\L)|\leq O(n)$, it is sufficient to prove the bound in \eqref{eq:t0}. More specifically, we prove 
\ba
\left|\log \left|\frac{1}{d}\frac{Z_{\b}(\L)}{Z_{\b}(\L \setminus X_0)}\right|\right|\leq e^2gh|\b|,\quad \forall|\b|\leq \b_0=\frac{1}{5e gh\k} \label{eq:t19}
\ea
The proof of this bound is by induction on the number of lattice sites $n$. 

For the base of the induction, we assume $Z_{\b}(\emptyset)=1$ for all complex $\b$. The induction hypothesis is the bound \eqref{eq:t19}. Thus, our goal is to assume \eqref{eq:t19} for lattices of size $n-1$ and show that the same bound holds for lattices of size $n$. By using the "telescoping products" as in Eq. \eqref{eq:t0} along with the induction hypothesis, we obtain the following bound for all lattices of size at most $n-1$ including $\L \setminus X_0$, 
\ba
\left|\log \left|\frac{1}{d^{|\supp(X\setminus{X_0})|}}\frac{Z_{\b}(\L\setminus X_0)}{Z_{\b}(\L\setminus X)}\right|\right|\leq e^2gh |\b|  |\supp(X\setminus X_0)|,\quad |\b|\leq \b_0, \label{eq:t8}
\ea
where $X\subseteq \L$ is an arbitrary non-empty set. According to the decomposition of $Z_{\b}(\L)$ obtained in \lemref{higH_temp_expansion}, we have
\ba
\frac{1}{d}\frac{Z_{\b}(\L)}{Z_{\b}(\L\setminus X_0)}=1+\sum_{\substack{\cX: x_0 \in \cX \\\cX \text{\ is connected\ } }} W_{\b}(\cX)\left(\frac{1}{d}\frac{Z_{\b}(\L\setminus \supp(\cX))}{Z_{\b}(\L\setminus X_0)}\right).
\ea
Thus, we get
\ba
\left|\log \left|\frac{1}{d}\frac{Z_{\b}(\L)}{Z_{\b}(\L \setminus X_0)}\right|\right|&=\left|\log \Bigg|1+\sum_{\substack{\cX: x_0 \in \cX \\\cX \text{\ is connected\ } }} W_{\b}(\cX)\left(\frac{1}{d}\frac{Z_{\b}(\L\setminus \supp(\cX))}{Z_{\b}(\L\setminus X_0)}\right) \Bigg|\right|\nonumber\\
&\leq -\log \left(1-\sum_{\substack{\cX: x_0 \in \cX \\\cX \text{\ is connected\ } }} |W_{\b}(\cX)|\left|\frac{1}{d}\frac{Z_{\b}(\L\setminus \supp(\cX))}{Z_{\b}(\L\setminus X_0)}\right|\right)\label{eq:t10}\\
&\leq -\log \left(1-\sum_{\substack{\cX: x_0 \in \cX \\\cX \text{\ is connected\ } }} \left(e^{|\b| h}-1\right)^{|\cX|}e^{ghe^2|\b||\supp(\cX)|}\right)\label{eq:t11},
\ea
where we used the following inequality to get to Eq. \eqref{eq:t10}: for all $\z \in \bbC, |\z|\leq 1$, we have $\big|\log |1+\z|\big|\leq -\log (1-|\z|)$. The last line \eqref{eq:t11} is obtained by plugging in the bound in \eqref{eq:t12} and the induction hypothesis \eqref{eq:t8}. 

It remains to show that Eq. \eqref{eq:t11} is bounded from above by $e^2 g h|\b|$. To get the desired upper bound on \eqref{eq:t11}, it is sufficient to prove the following bound which we separately prove in \lemref{inductive bound on expansion}:
\ba
\sum_{\substack{\cX: x_0 \in \cX \\ \cX \text{\ is connected\ }}} \left(e^{|\b| h}-1\right)^{|\cX|}e^{ghe^2|\b||\supp(\cX)|} \leq e(e-1) gh |\b|,\quad |\b|\leq \b_0.\label{eq:t20}
\ea
The reason this implies the claimed upper bound on \eqref{eq:t11} is that we have 
\ba
 -\log \left(1-\sum_{\substack{\cX: x_0 \in \cX \\\cX \text{\ is connected\ } }} \left(e^{|\b| h}-1\right)^{|\cX|}e^{ghe^2|\b||\supp(\cX)|}\right)&\leq -\log \left(1-e(e-1)gh|\b|\right)\nonumber\\
&\leq e^2 gh |\b|.
\ea
To get to the last line we used the inequality $ -\log(1-\frac{e-1}{e} y)\leq y,\ \forall y\in[0,1]$ with $y=e^2gh|\b|$. Notice that $\b_0=\frac{1}{5e g h \k}$, which means $y= e^2 gh |\b| \leq 1$ for $|\b|\leq \b_0$.

This concludes the induction step and also the proof of the theorem.
\end{proof}

\begin{lem}\label{lem:inductive bound on expansion}
Consider the same setup as \thmref{t_higH_temp}. The following bound holds:
\ba
\sum_{\substack{\cX: x_0 \in \cX \\ \cX \text{\ is connected\ }}} \left(e^{|\b| h}-1\right)^{|\cX|}e^{ghe^2|\b||\supp(\cX)|} \leq e(e-1) gh |\b|,\quad |\b|\leq \b_0.\label{eq:t20_}
\ea
\end{lem}
\begin{proof}[\textbf{Proof of \lemref{inductive bound on expansion}}]
Since for a connected set $\cX$, both its size $|\cX|$ and the size of its support $|\supp(\cX)|$ show up in the summation, we need to take extra care in finding a proper upper bound. We achieve this again by induction, this time over the size of $|\cX|$.
We begin with restating the sum in \eqref{eq:t20_} in a different form. This includes adding the contribution of all connected sets $\cX$ that contain a site $x_0$ in the following order.

First, we consider the contribution of a fixed set $X\subset \L$ with size and diameter at most $\k$ and $R$ that contains $x_0$. We then sum over all the connected sets that include a site $x\in X$. It is not hard to see that by selecting all possible choices of $X$ and performing the addition in this way, we \emph{overcount} the number of connected sets $\cX$ that contain $x_0$, and therefore get an upper bound on the original sum in \eqref{eq:t20_}. More formally, for any $x_0 \in \L$, we have 

\ba
&\sum_{\substack{\cX: x_0 \in \cX \\\cX \text{\ is connected\ } }} \left(e^{|\b| h}-1\right)^{|\cX|}e^{ghe^2|\b||\supp(\cX)|} \nonumber \\
&\leq \sum_{\substack{X: x_0\in X \\|X|\leq \k\\\diam(X)\leq R}} \left( \left(e^{|\b| h}-1\right) e^{ghe^2|\b||X|} \prod_{x \in X}\left(\sum_{\substack{\cX: x \in \cX \\\cX \text{\ is connected\ } }} \left(e^{|\b| h}-1\right)^{|\cX|}e^{ghe^2|\b||\supp(\cX)|}\right)\right)\nonumber\\
&\leq \sum_{\substack{X: x_0\in X \\|X|\leq \k\\\diam(X)\leq R}} \left( \left(e^{|\b| h}-1\right) e^{ghe^2|\b||X|} \prod_{x \in X}\left(1+ e(e-1)gh|\b|\right)\right)\\
&\leq g \left(e^{|\b| h}-1\right) \left(e^{ghe^2|\b|}\left(1+ e(e-1)gh|\b|\ \right)\right)^{\k}\label{eq:t21} \\
&\leq (e-1)g|\b|he^{e(2e-1)gh|\b|\k}\label{eq:t22}\\
&\leq e(e-1)g|\b|h,\label{eq:t23}
\ea
where we used the induction hypothesis to get from the second to the third line. Eq. \eqref{eq:t21} follows from the definition of the growth constant $g$ which gives an upper bound on the number of sets $X$ containing $x_0$ with size at most $\k$. To get to Eq. \eqref{eq:t22} and \eqref{eq:t23}, we use the fact that $1+y\leq e^y$, $e^y-1\leq (e+1)y$ for $y\in[0,1]$ and $|\b|\leq \frac{1}{5egh\k}$.\qedhere
\end{proof}

\section{Analyticity implies exponential decay of correlations}\label{sec:Zero-free region implies the exponential decay of correlations}
In this section, we show that the exponential decay of correlations is a \emph{necessary} condition for the free energy to be analytic and bounded close to the real axis. Our bounds are stronger for commuting Hamiltonians on arbitrary lattices and non-commuting Hamiltonians on a $1$D chain and slightly weaker for generic geometrically-local cases.

Similar to the rest of this paper, our general strategy heavily uses extrapolation between different regimes of the inverse temperature parameter. We know that at $\b=0$, the Gibbs state is just the maximally mixed state, so the decay of correlations property trivially holds. Additionally, we show that at $\b=0$, the low-order derivatives of a function that encode the amount of correlation between two regions are zero. This combined with the absence of singularities coming from the analyticity condition puts an exponentially small bound on how fast this function (i.e. the correlations) can grow with $\b$.

The proof is reminiscent of the one for classical systems first shown by \cite{Dobrushin1}. As explained earlier, the essence of the proof is the following simple lemma from complex analysis. 

\begin{lem}[cf. \cite{Dobrushin1}]\label{lem:functions with many zero derivatives are small}
Let $f(z_1,\dots,z_m)$ be a complex function that on a bounded connected open region $\Omega \subset \bbC^{m}$ is analytic and $|f(z_1,\dots,z_m)|\leq M$. Let $k_1,\dots,k_m$ be non-negative integers summing to $K$. Suppose that $f(z_1,\dots,z_m)$ and its following derivatives are zero at some $\z_0 \in \Omega$:
\ba\label{eq:v-1}
\frac{d^{K}}{d^{k_1}z_i\dots d^{k_m}z_m} f(z_1,\dots,z_m)\Bigr|_{\z_0}=0 \quad \text{if }\  |\{i\in [m]: k_i\geq 1\}|\leq L-1,
\ea
that is, unless we take the derivative with respect to at least $L$ distinct variables $z_i$, this derivative is zero at $\z_0$. Then, for any $(z_1,\dots,z_m) \in \Omega$, there exist constants $c_1,c_2$ depending on $\z_0$ and $(z_1,\dots,z_m)$ such that $|f(z_1,\dots,z_m)|\leq c_1 M e^{-c_2L}$.
\end{lem}

\begin{proof}[\textbf{Proof of \lemref{functions with many zero derivatives are small}}]
Without loss of generality, we can restrict ourselves to the single variable case, $m=1$, by defining a path parameterized by $z\in [0,1]$ that connects $\z$ to any point $(z_1,\dots,z_m)$ of interest. We denote the function on this path by $f(z)$. Region $\Omega$ in this case is just a region in the complex plane around $[0,1]$ that has a small imaginary part such that $f(z)$ remains analytic and bounded. 

Using conformal mapping similar to what we did in \thmref{extrapolation algorithm}, we can map the unit disk onto $\Omega$, which is the set of $z\in \bbC$ such that $|z|\leq 1$. Hence, without loss of generality, we assume $f(z)$ is analytic on the unit disk. It is also not hard to see that Eq. \eqref{eq:v-1} implies the first $L$ derivatives of $f(z)$ vanish at the origin. Thus, the Taylor expansion of $f(z)$ converges and we have 
\ba
\forall z \in \Omega,\quad f(z)=\sum_{k>L} a_k z^k=z^{L}\sum_{k>L} a_k z^{k-L},
\ea
but $\sum_{k>L} a_k z^{k-L}$ is itself an analytic function, so it is either a constant or attains its maximum absolute value on the boundary. It follows from $|f(z)|\leq M$ that in either case $|\sum_{k>L} a_k z^{k-L}|\leq M$. This implies $\forall|z|\leq 1,\ |f(z)|\leq M |z|^L$, which in turn proves the theorem. 
\end{proof}

The connection between \lemref{functions with many zero derivatives are small}, the decay of correlations, and the analyticity condition becomes clear once we substitute our choice of function $f(z_1,\dots,z_m)$ and region $\Omega$. We begin by defining $\Omega$. Fixing our choice of function $f(z_1,\dots,z_m)$ is postponed until after we discuss the precise statement of the analyticity condition and the decay of correlations. 

Region $\Omega$ corresponds to the region near the real $\b$ axis where the partition function does not vanish. Given a local Hamiltonian $H=\sum_{i=1}^m H_i$, we define complex variables $z_1,\dots,z_m$ such that each $z_i$ roughly equals $\b$ plus some small \emph{complex deviation}. Hence, instead of working with functions of $\b H$ such as $\exp(-\b H)$, we consider functions of $\sum_{i=1}^m z_i H_i$ as in $\exp(-\sum_{i=1}^m z_i H_i)$. For a fixed inverse temperature $\b$ and maximum deviation $\d$, we denote the set of such tuples $(z_1,\dots,z_m)$ by $\G_{\d,\b}$. By varying $\b$ from zero to some constant $\b$ and taking the union of corresponding $\G_{\d,\b}$, the set $\Omega_{\d,\b}$ is obtained. 

As discussed earlier, the critical temperature $\b_c$ corresponds to the thermal phase transition point, where complex zeros of the partition function approach the real axis. Note that even with deviations, we do not want any of the variables $z_i$ to exceed $\b_c$. More precisely, we have the following definition.

\begin{definition}[The vicinity of the real $\b$ axis]\label{def:The vicinity of real}
Let $\G_{\d,\b}$ be the set $\{(z_1,\dots , z_m): \forall i \in [m], z_i\in \bbC, |z_i-\b|\leq \d\}$. We define $\Omega_{\d,\b}$ to be $\Omega_{\d,\b}=\bigcup_{\substack{\b' \in \bbR^+\\ \b' < \b/(1+\d)}} \G_{\d,\b'}$. 
\end{definition}

We also define the perturbed Gibbs state as follows.
\begin{definition}[Complex perturbed Gibbs state]\label{def:Complex perturbed Gibbs state}
The $\d$-perturbed Gibbs state of a local Hamiltonian $H=\sum_{i}^m H_i$ at inverse temperature $\b$ is defined as
\ba
\r_{\vec{z}}(H)=\frac{e^{-\sum_{i=1}^mz_i H_i}}{\Tr[e^{-\sum_{i=1}^m{z_i H_i}}]},\quad \vec{z}=(z_1,\dots,z_m)\in \G_{\d,\b}
\ea
where $\G_{\d,\b}$ is defined in \defref{The vicinity of real}.
\end{definition}

The analyticity condition we consider here is stronger than the ones derived in \secref{High temperatures: Fisher zeros} in the high temperature regime or used in the approximation algorithm in \secref{Algorithm for estimating the partition function}. Previously we only included systems with open boundary conditions in our analysis, but here we also need to allow for other boundary conditions. This is not restricted to the quantum case, and Dobrushin and Shlosman use similar conditions in their proof for classical systems \cite{Dobrushin1}. The precise statement of our condition is the following:

\begin{customcond}{1}[Analyticity after measurement] \label{cond:analyticity after measurement}The free energy of a geometrically-local Hamiltonian $H$ is $\d$-analytic at $\b$ if for any local operator $N\geq 0$ with $\norm{N}=1$, there exists a constant $c$ such that
\ba\label{eq:v1}
\left| \log\left( \Tr \left[ e^{ -\sum_{i=1}^m z_i H_i} N \right] \right)\right|\leq c n,\quad \forall(z_1,\dots, z_m)\in \G_{\d ,\b}.
\ea
\end{customcond}

To see the motivation for this condition, first note that for classical spin systems, the operator $N$ sets the boundary conditions. In that case, we can fix the value of certain spins in the system before computing the partition function, or more generally, finding the Gibbs distribution. A natural question then is how varying these boundary conditions affects the distribution. In particular, the \emph{uniqueness of the Gibbs distribution} refers to the case that in the limit of a large number of particles, changing distant spins has a negligible effect on the distribution of spins on a finite region. Hence, a \emph{unique} Gibbs distribution can be defined for such systems. This condition is not satisfied at all temperatures, and below the critical temperature, multiple Gibbs distributions exist.Thus, it seems natural to include the boundary conditions in the partition function when studying its complex zeros and the critical behavior of the system in general.

For quantum systems, one can think of fixing the boundary spin values by projecting them onto a specific state or more generally by post-selecting after a local measurement has been performed. Hence, $\Tr \left[\exp \left( -\sum_{i=1}^m z_i H_i\right) N \right]$ is the partition function of the normalized Gibbs state after conditioning on the measurement outcome associated with $N$. Notice that, in principle, the state of the spins after post-selection can be entangled. As we will see, this causes technical difficulties in extending the classical results to the quantum regime. 

Our goal is to show that \condref{analyticity after measurement} on the analyticity of the free energy implies the exponential decay of correlations. This condition is stated as follows.

\begin{customcond}{2}[Exponential decay of correlations]\label{cond:Exponential decay of correlations}
The correlations in the Gibbs state $\rho_{\b}(H)$ of a geometrically-local Hamiltonian decay exponentially if for any local Hermitian operators $O_1$ and $O_2$, there exist constants $\xi$ and $c$ such that
\ba\label{eq:v2}
\big|\Tr\left[ \rho_{\b}(H) O_1 O_2 \right]-\Tr\left[ \rho_{\b}(H) O_1 \right]\Tr\left[ \rho_{\b}(H) O_2 \right] \big| \leq c\norm{O_1}\norm{O_2} e^{-\dist(O_1,O_2)/\xi}.
\ea
\end{customcond}

We first prove a slightly weaker version of \condref{Exponential decay of correlations} assuming \condref{analyticity after measurement}. We then improve our bound for commuting and $1$D Hamiltonians. 

\begin{thm}[Analyticity implies exponential decay of correlations] \label{thm:The absence of zeros implies the decay of correlations in the non-commuting case} Suppose the free energy of a geometrically-local Hamiltonian is $\d$-analytic for all $\b\in[0,\b_c)$ as in \condref{analyticity after measurement}. Then the correlations between any two operators $O_1,O_2$ with $\dist(O_1,O_2)=\Omega(\log n)$ decay exponentially for all $\b\in[0,\b_c)$ as in \condref{Exponential decay of correlations}.
\end{thm}

\begin{proof}[\textbf{Proof of \thmref{The absence of zeros implies the decay of correlations in the non-commuting case}}]

We can without loss of generality assume $\norm{O_1},\norm{O_2}\leq 1$. Let $A_1=\supp(O_1)$ and $A_2=\supp(O_2)$. Each of the observables $O_1$ and $O_2$ can be decomposed into two positive semi-definite (PSD) matrices: $O_1=O^+_1-O^-_1$ and $O_2=O^+_2-O^-_2$, where $O^+_1,O^+_2$ include the positive eigenvalues of $O_1,O_2$ and $-O^-_1,-O^-_2$ include the negative ones. We can write the covariance in Eq. \eqref{eq:v2} as
\ba
&\huge| \Tr\left[ \rho_{\b}(H) O_1 O_2 \right]-\Tr\left[ \rho_{\b}(H) O_1 \right]\Tr\left[ \rho_{\b}(H) O_2 \right] \huge| \nonumber\\
&=\left|\sum_{\a,\g\in\{\pm\}}\a \g \left(\Tr\left[\rho_{\b}(H)O^{\a}_1 O^{\g}_2\right]-\Tr\left[\rho_{\b}(H) O^{\a}_1 \right] \Tr\left[ \rho_{\b}(H) O^{\g}_2 \right]\right)\right| \nonumber\\
&\leq 4\cdot \max_{\substack{N_2,N_1\geq 0:\\ \norm{N_2},\norm{N_1}\leq 1}} \left|\left(\Tr\left[\rho_{\b}(H) N_2 N_1 \right]-\Tr\left[\rho_{\b}(H) N_2\right] \Tr\left[ \rho_{\b}(H) N_1\right]\right)\right|,\label{eq:v5}
\ea
where $\supp(N_2)=A_1$ and $\supp(N_1)=A_2$. Recall that the post-selected state $\rho_{\b}(H|N)$ is defined by 
\ba 
\rho_{\b}(H|N)=\frac{\sqrt{N} \exp(-\b H)\sqrt{N}}{\Tr[\exp(-\b H)N]}.
\ea
The bound \eqref{eq:v5} can be rewritten as 
\ba
\left|\left(\Tr\left[\rho_{\b}(H) N_2 N_1 \right]-\Tr\left[\rho_{\b}(H) N_2\right] \Tr\left[ \rho_{\b}(H) N_1\right]\right)\right|
&=\left|\Tr\left[\rho_{\b}(H) N_2 \right]\right| \left|\Tr[\rho_{\b}(H|N_2) N_1]-\Tr\left[\rho_{\b}(H|\iden)N_1\right]\right|\nonumber\\
&\leq \left|\Tr[\rho_{\b}(H|N_2) N_1]-\Tr\left[\rho_{\b}(H|\iden)N_1\right]\right|.
\ea

Hence, our goal is to show 
\ba\label{eq:v6}
\left|\Tr[\rho_{\b}(H|N_2) N_1]-\Tr\left[\rho_{\b}(H|\iden)N_1\right]\right|\leq c e^{- \dist(O_1,O_2)/\xi}.
\ea
We instead show
\ba\label{eq:v7}
\left|\log\left(\frac{\Tr[\rho_{\b}(H|N_2) N_1]}{\Tr\left[\rho_{\b}(H|\iden)N_1\right]}\right)\right| \leq c e^{- \dist(O_1,O_2)/\xi}.
\ea
To see why this implies \eqref{eq:v6}, we can further upper bound the right-hand side using the inequality $x\leq -\log(1-x)$ for $x<1$ and choosing $x=c \exp(- \dist(O_1,O_2)/\xi)$. Then the fact that $|\Tr(\r_{\b}(H) N_1)|\leq 1$ implies the desired bound. We can prove a similar bound even when instead of $\iden$ there is any other PSD operator in the denominator. One way to interpret these bounds is that a local measurement on region $A_2$ is undetected from the perspective of local operators on region $A_1$.

\begin{figure}[t!]
\centering
	\includegraphics[width=.33\textwidth]{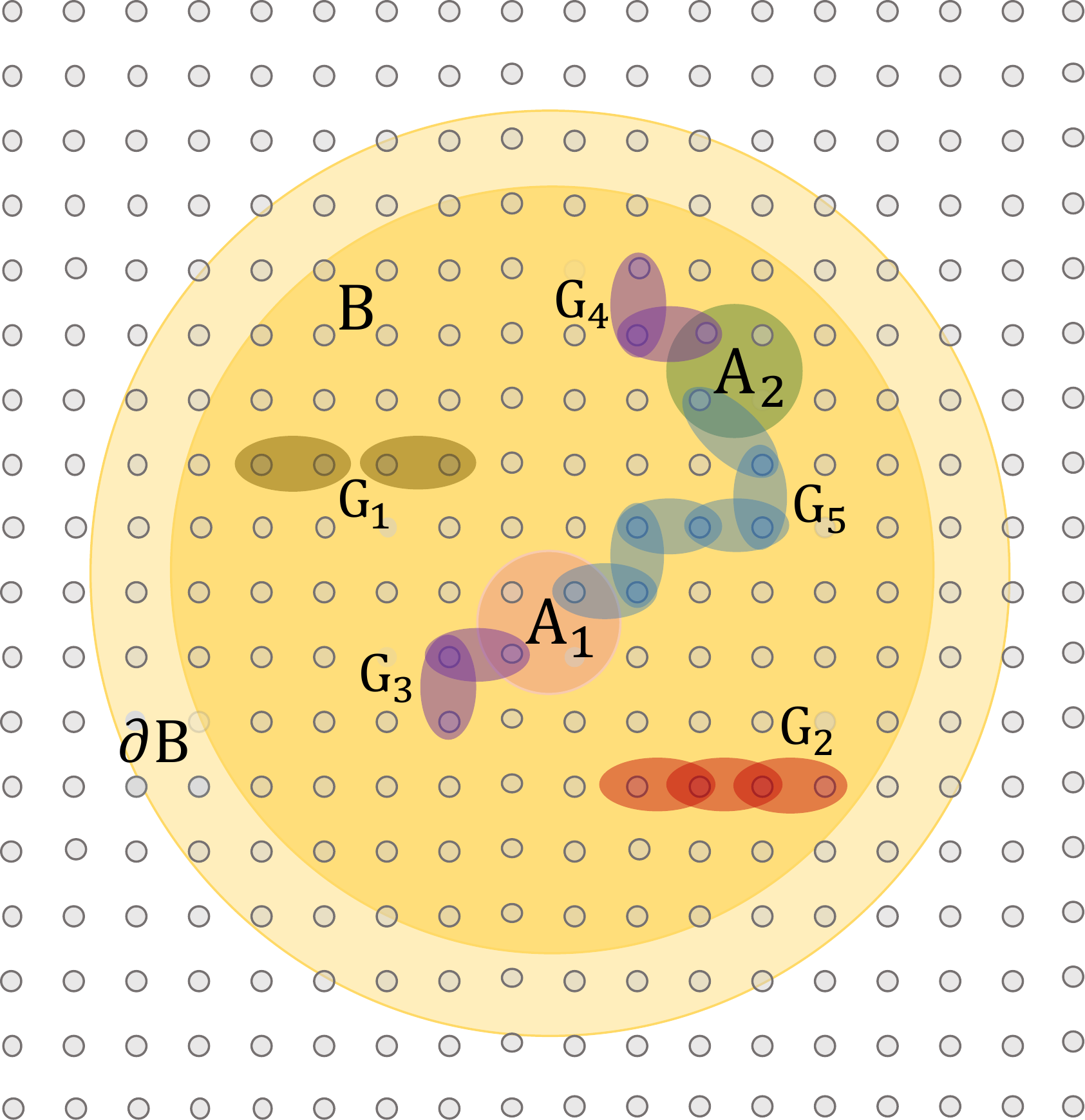}
	\caption{To study the correlations between regions $A_1,A_2$, we can restrict the Gibbs state to region $B$ while adding an operator on the boundary $\partial B$ to include the effect of the rest of the lattice. Regions $G_1,\dots, G_5$ show up when studying the derivatives of the correlation function. See the proofs of \thmref{The absence of zeros implies the decay of correlations in the non-commuting case} and \thmref{The absence of zeros implies the decay of correlations in the commuting case}.}\label{fig:regions_ABC}
\end{figure}

The proof follows from \lemref{functions with many zero derivatives are small}. We first consider a perturbed version of \eqref{eq:v7} using \defref{Complex perturbed Gibbs state}. We define $f(z_1,\dots,z_m)$ as
\ba\label{eq:v42}
f(z_1,\dots,z_m)&= \log\left(\frac{\Tr[\r_{\vec{z}}(H|N_2)N_1]}{\Tr[\r_{\vec{z}}(H|\iden)N_1]}\right) \nonumber \\
&=\log\left(\frac{\Tr[e^{- \sum_{i=1}^m z_i H_i}  N_2 N_1]}{\Tr[e^{- \sum_{i=1}^m z_i H_i} N_2]}\frac{\Tr[e^{-\sum_{i=1}^m z_i H_i}]}{\Tr[e^{- \sum_{i=1}^m z_i H_i}N_1]}\right).
\ea
This function is our choice for $f(z_1,\dots,z_m)$ in \lemref{functions with many zero derivatives are small}. In particular, we prove that assuming \condref{analyticity after measurement} is satisfied, $f(z_1,\dots,z_m)$ is analytic in $\Omega_{\d,\b_c}$, has a bounded absolute value, and has vanishing derivatives at $z_1=\dots=z_m=0$. Let us begin with the analyticity and boundedness. 

\paragraph{Analyticity and boundedness:}  From \eqref{eq:v1} we see that for any positive operator $N$, the post-selected free energy is analytic and there exists some constant $c$ such that
\ba
\left| \log\left( \Tr \left[ e^{ - \sum_{i=1}^m z_i H_i} N \right] \right)\right|\leq c n,
\ea
By using a proper choice for $N$, we see that $f(z_1,\dots,z_m)$ is a sum of analytic functions and therefore is itself analytic. We also get an upper bound on $|f(z_1,\dots,z_m)|$, that is,
\ba
\forall (z_1,\dots,z_m)\in \Omega_{\d,\b_c},\quad  \left|f(z_1,\dots,z_m)\right|&\leq \left|\log\left(\Tr[e^{- \sum_{i=1}^m z_i H_i}  N_2 N_1]\right)\right|+\left|\log\left(\Tr[e^{- \sum_{i=1}^m z_i H_i}  N_2 ]\right)\right|\nonumber\\
&+\left|\log\left(\Tr[e^{- \sum_{i=1}^m z_i H_i}N_1]\right)\right|+\left|\log\left(\Tr[e^{- \sum_{i=1}^m z_i H_i}]\right)\right|\nonumber\\
&\leq 4c n.
\ea

\paragraph{Vanishing derivatives:} It remains to show that certain derivatives of $f(z_1,\dots,z_m)$ are zero at the point $\b=0$, which is inside $\Omega_{\d,\b_c}$. The derivatives of $f(z_1,\dots z_m)$ are combinations of terms like 
\ba\label{eq:v9}
\frac{d^K}{d^{k_1}z_i\dots d^{k_m}z_m} \log\left(\Tr \left[ e^{ - \sum_{i=1}^m z_i H_i} N_2 N_1 \right] \right)\Bigr|_{z=0},
\ea
where $k_i\geq 0$ and $K=\sum_{i=1}^m k_i$. Notice that we are including the $z_i$ that are not in the derivative by letting $k_i=0$. We claim in certain instances that these terms are either zero or cancel each other. Consider all the local terms $H_i$ that we are taking a derivative with respect to their $z_i$, i.e. $k_i\geq 1$. We denote the union of the support of these terms by $G$. Recall that $A_1,A_2$ are the support of $O_1,O_2$, respectively. Region $G$ fits into one of the following cases. 

\textbf{Case 1:} $G$ is not connected and does not intersect with $ A_1 \cup A_2$ (see $G_1$ in 
\fig{regions_ABC} for an example). In this case, the terms in the derivatives are 
\ba\label{eq:v11}
\frac{d^K}{d^{k_1}z_i\dots d^{k_m}z_m} \log\left(\Tr \left[ e^{ - \sum_{i=1}^m z_iH_i} N_2 N_1 \right] \right)\Bigr|_{z=0}&=\frac{d^K}{d^{k_1}z_i\dots d^{k_m}z_m} \log \left(\Tr[N_2]\Tr[N_1] \prod_{i: k_i\geq 1} \Tr \left[ e^{-z_i H_i}\right] \right)\Bigr|_{z=0}\nonumber\\
&=\frac{d^K}{d^{k_1}z_i\dots d^{k_m}z_m} \log \left(\Tr[N_2]\Tr[N_1]\right) \nonumber\\ &+ \sum_{i: k_i\geq 1} \frac{d^K}{d^{k_1}z_i\dots d^{k_m}z_m} \log \left(\Tr \left[ e^{-z_i H_i}\right] \right)\Bigr|_{z=0}\nonumber\\
&=0.
\ea
In the first line, we used the fact that sublattices $A_1$, $A_2$, and $\supp(H_i)$ with $k_i\geq 1$ do not intersect. The last line follows because $\Tr[N_2]\Tr[N_1]$ is a constant, and $\Tr \left[ e^{-z_i H_i}\right]$ only depends on $z_i$ and its derivative with respect to other $z_i$ is zero. 

\textbf{Case 2:} $G$ is connected but does not intersect with $ A_1 \cup A_2$ (see $G_2$ in 
\fig{regions_ABC} for an example). Similar to \eqref{eq:v11}, we can still separate $\Tr[N_2]\Tr[N_1]$ from the remaining terms and their derivative is zero. Hence, we obtain
\ba\label{eq:v12_}
\frac{d^K}{d^{k_1}z_i\dots d^{k_m}z_m} \log\left(\Tr \left[ e^{ - \sum_{i=1}^m z_i H_i} N_2 N_1\right] \right)\Bigr|_{z=0}= \frac{d^K}{d^{k_1}z_i\dots d^{k_m}z_m} \log \left(\Tr \left[ e^{-\sum_{i: k_i\geq 1} z_i H_i}\right] \right)\Bigr|_{z=0}.
\ea
Although this term does not necessarily equal zero, the derivatives of $f(z_1,\dots,z_m)$ are combinations of terms like Eq. \eqref{eq:v12_}. These terms are all equal as we can separate traces involving $N_2$ and $N_1$ using the same argument as above, but they appear with opposite signs and thus cancel each other. More precisely, we have
\ba
&\frac{d^K}{d^{k_1}z_i\dots d^{k_m}z_m} f(z_1,\dots,z_m)\Bigr|_{z=0}\nonumber\\
&=\frac{d^K}{d^{k_1}z_i\dots d^{k_m}z_m} \Big(\log \left(\Tr[N_2N_1]\Tr \left[ e^{-\sum_{i: k_i\geq 1} z_i H_i}\right] \right)
- \log \left(\Tr[N_2]\Tr \left[ e^{-\sum_{i: k_i\geq 1} z_i H_i}\right] \right)\nonumber\\
&-\log \left(\Tr[N_1]\Tr \left[ e^{-\sum_{i: k_i\geq 1} z_i H_i}\right] \right)
+\log \left(\Tr \left[ e^{-\sum_{i: k_i\geq 1} z_i H_i}\right] \right)\Big)\Bigr|_{z=0}\nonumber\\
&=0.
\ea

\textbf{Case 3:} $G$ is connected and intersects with only one of $A_1$ or $A_2$ (see $G_3$ or $G_4$ in 
\fig{regions_ABC} for an example). Similar to Case 2, the derivatives of $f(z_1,\dots,z_m)$ consist of equal terms with opposite signs and therefore vanish. Here, we show the case where $G$ only intersects $A_2$. The other ones similarly follow. 
\ba
&\frac{d^K}{d^{k_1}z_i\dots d^{k_m}z_m} f(z_1,\dots,z_m)\Bigr|_{z=0}\nonumber\\
&=\frac{d^K}{d^{k_1}z_i\dots d^{k_m}z_m} \Big(\log \left(\Tr[N_1]\Tr \left[ e^{-\sum_{i: k_i\geq 1} z_i H_i}N_2\right] \right)
- \log \left(\Tr \left[ e^{-\sum_{i: k_i\geq 1} z_i H_i}N_2\right] \right)\nonumber\\
&-\log \left(\Tr[N_1]\Tr \left[ e^{-\sum_{i: k_i\geq 1} z_i H_i}\right] \right)
+\log \left(\Tr \left[ e^{-\sum_{i: k_i\geq 1} z_i H_i}\right] \right)\Big)\Bigr|_{z=0}\nonumber\\
&=0,
\ea
in which the first two and last two terms cancel each other. 

\textbf{Case 4:} $G$ is connected and intersects with both $A_1$ and $A_2$ (see $G_5$ in 
\fig{regions_ABC} for an example). Here, the cancellation that appeared in the other cases does not happen. Thus, this is the only situation in which the derivatives are non-zero.

The important observation is that for Case 4 to happen, $G$ needs to be \emph{long} enough to touch both $A_1$ and $A_2$. Hence, if the number of $z_i$ with $k_i\geq 1$ is less than roughly $\dist(O_1,O_2)$, their corresponding derivative vanishes. 
Having all the criteria needed for applying \lemref{functions with many zero derivatives are small}, i.e analyticity, boundedness, and zero derivatives, we can get the following bound on $|f(z_1,\dots,z_m)|$ for some constant $c$ and $\xi$:
\ba\label{eq:v40}
|f(\b,\dots,\b)|\leq c n e^{-\dist(O_1,O_2)/\xi},
\ea
which as explained before implies
\ba\label{eq:v41}\big|\Tr\left[ \rho_{\b}(H) O_1 O_2 \right]-\Tr\left[ \rho_{\b}(H) O_1 \right]\Tr\left[ \rho_{\b}(H) O_2 \right] \big| \leq c n e^{-\dist(O_1,O_2)/\xi}.
\ea
Due to the extra factor of $n$ in front of this bound, it implies the exponential decay of correlations only when $\dist(O_1,O_2)=\Omega(\log n)$.\qedhere
\end{proof}

\subsection{Tighter bounds for commuting Hamiltonians}
Here we show how using the commutativity of $H$ enables us to remove the extra factor of $n$ in the bound \eqref{eq:v40} that we derived for general Hamiltonians. We state this in the following theorem.

\begin{thm} \label{thm:The absence of zeros implies the decay of correlations in the commuting case} Suppose $H$ is a  geometrically-local Hamiltonian with mutually commuting terms that satisfies \condref{analyticity after measurement} for $\b\in[0,\b_c)$. Then the correlations between any two operators $O_1,O_2$ decay exponentially for all $\b\in[0,\b_c)$ as in \condref{Exponential decay of correlations}.
\end{thm}

\begin{proof}[\textbf{Proof of \thmref{The absence of zeros implies the decay of correlations in the commuting case}}]
The proof of this theorem follows similar steps to that of \thmref{The absence of zeros implies the decay of correlations in the non-commuting case}. A crucial difference, which is the only part where we use the commutativity of local terms $H_i$, is the following. The Hamiltonian $H$ in states $\r_{\b}(H|N_2)$ and $\r_{\b}(H|\iden)$ involves terms acting on all $n$ sites in lattice $\L$. In our analysis, we can essentially neglect the contribution of sites that are far from both region $A_1$ and $A_2$. In other words, as shown in \fig{regions_ABC}, let $B\subset \L$ be a ball of diameter slightly larger than $\dist(O_1,O_2)$ centered at $A_1$ that encloses region $A_2$. We restrict the Hamiltonian and states $\r_{\b}(H|N_2),\r_{\b}(H|\iden)$ to this region and include the effect of other sites by an operator acting on $\partial B$, the boundary of this enclosing region. We prove \eqref{eq:v7} for this smaller region. Without this step, we end up getting an upper bound like $c n \exp(-\dist(O_1,O_2)/\xi)$, which has an extra factor of $n$, the number of sites in $\L$, whereas with the restriction to the enclosing region, this factor is the number of sites in $B$ which is negligible compared to the exponential decay factor. 
More formally, since $H$ is a commuting Hamiltonian, we have $e^{-\b H}=e^{-\b H_B}e^{-\b (H-H_B)}$. Hence, we get
\ba
\Tr[\r_{\b}(H|N_2) N_1]&=\frac{\Tr[e^{-\b H }N_2 N_1]}{\Tr[e^{-\b H} N_2]}\nonumber\\
&=\frac{\Tr_B[e^{-\b H_B} \Tr_{\bar{B}}[e^{-\b (H-H_B)}]N_2 N_1]}{\Tr_B[e^{-\b H_B} \Tr_{\bar{B}}[e^{-\b (H-H_B)}]N_2]}\nonumber\\
&=\frac{\Tr_B[e^{-\b H_B} \s N_2 N_1]}{\Tr_B[e^{-\b H_B}\s N_2]}\nonumber\\
&=\Tr[\r_{\b}(H_B|\s N_2)N_1],\label{eq:v8} 
\ea
where
\ba\label{eq:v12}
\s=\frac{\Tr_{\bar{B}}[e^{-\b (H-H_B)}]}{ \Tr_{\bar{B}\cup \partial \bar{B}}[e^{-\b (H-H_B)}]}
\ea
is a state acting on the boundary $\partial \bar{B}$ \footnote{Based on our definition of the boundary of a region, the boundary $\partial \bar{B}$ is inside $B$.}. Thus, we can replace the operator $N_2$ by $\s \otimes N_2$ acting on a larger region $\partial \bar{B}\cup A_2$, which is still only a constant, and restrict our attention to region $B$. We can now repeat the argument of \thmref{The absence of zeros implies the decay of correlations in the non-commuting case}. Let the perturbed Hamiltonian restricted to region $B$ be $H_B(\vec{z})=\sum_{H_i: \supp(H_i)\subset B} z_i H_i$, where for simplicity, the number of local terms in $H_B$ is denoted again by $m$. By plugging \eqref{eq:v8} into \eqref{eq:v42}, we see that the function $f(z_1,\dots,z_m)$ is
\ba
f(z_1,\dots,z_m)&=\log\left(\frac{\Tr[\r_{\vec{z}}(H_B|\s N_2)N_1]}{\Tr[\r_{\vec{z}}(H_B|\s)N_1]}\right)\nonumber\\
&=\log\left(\frac{\Tr[e^{-\b H_B(\vec{z})} \s N_2 N_1]}{\Tr[e^{- H_B(\vec{z})} \s N_2]}\frac{\Tr[e^{-\b H_B(\vec{z})} \s]}{\Tr[e^{-\b H_B(\vec{z})} \s N_1]}\right)
\ea

The rest of the proof of \thmref{The absence of zeros implies the decay of correlations in the non-commuting case} applies to this function. In particular, assuming \condref{analyticity after measurement} holds, this function is bounded and analytic in $\Omega_{\d,\b_c}$, i.e. $|f(z_1,\dots,z_m)|\leq c|B|$. Similarly, one can see that the low-order derivatives of $f(z_1,\dots,z_m)$ are zero. Since the distance between $\partial{B}$ and $A_1$ is still $O(\dist(O_1,O_2))$, \lemref{functions with many zero derivatives are small} implies $|f(\b,\dots,\b)|\leq c |B| \exp(-\dist(O_1,O_2)/\xi)$. Hence, we have
\ba\label{eq:v44}\big|\Tr\left[ \rho_{\b}(H) O_1 O_2 \right]-\Tr\left[ \rho_{\b}(H) O_1 \right]\Tr\left[ \rho_{\b}(H) O_2 \right] \big| \leq c\ \dist(O_1,O_2)^{D} e^{-\dist(O_1,O_2)/\xi}.
\ea
\end{proof}

\subsection{Tighter bounds for 1D Hamiltonians}

\begin{thm} \label{thm:zeros to decay for 1D} 
 Let $H$ be a geometrically-local Hamiltonian on a $1D$ chain that satisfies \condref{analyticity after measurement}. Then, the exponential decay of correlations given in \condref{Exponential decay of correlations} also holds for this Hamiltonian. 
\end{thm}
\begin{proof}[\textbf{Proof of \thmref{zeros to decay for 1D}}]
The proof is similar to that of \thmref{The absence of zeros implies the decay of correlations in the non-commuting case} and \thmref{The absence of zeros implies the decay of correlations in the commuting case}. Recall that an important step is to introduce boundary states $\s$ that include the effect of terms in the Hamiltonian $H$ that are acting on the boundary or outside of some region $B$. Region $B$ encloses the support of operators whose correlations we want to bound. There, we use the commutativity of $H$ to find the boundary states $\s$ which does not hold in general. Here, we show how, by using the quantum belief propagation operator $\eta$ we introduced in \propref{Quantum belief propagation}, we can achieve the same boundary state in $1$D.  

We do not go through all steps of the proof of \thmref{The absence of zeros implies the decay of correlations in the non-commuting case} again. Instead, we directly show that by restricting the Hamiltonian to region $B$ and adding the boundary terms, the covariance in \eqref{eq:v2} changes negligibly. Then we apply bound \eqref{eq:v41} to this restricted covariance. Since the number of particles inside $B$ is constant, instead of the extra factor of $n$, we get a constant prefactor as desired. 

Recall that using the belief propagation equation \eqref{eq:p0} and the bound \eqref{eq:p1}, we can remove the boundary terms $H_{\partial B}$ acting between $B,\bar{B}$ from the Gibbs state and get 
\ba
\Tr[\r_{\b}(H)O_1 O_2]&=\Tr\left[\frac{Z_{\b}(H-H_{\partial B})}{Z_{\b}(H)} \eta \r_{\b}(H-H_{\partial B})\eta^{\dag} O_1 O_2\right]\nonumber\\
&=\Tr\left[\frac{Z_{\b}(H-H_{\partial B})}{Z_{\b}(H)} \eta_{\ell}\r_{\b}(H-H_{\partial B})\eta^{\dag}_{\ell}O_1 O_2\right]\nonumber\\
&+\Tr\left[\frac{Z_{\b}(H-H_{\partial B})}{Z_{\b}(H)} \eta_{\ell} \r_{\b}(H-H_{\partial B})(\eta^{\dag}-\eta_{\ell}^{\dag}) O_1 O_2\right]\nonumber\\
&+\Tr\left[\frac{Z_{\b}(H-H_{\partial B})}{Z_{\b}(H)} (\eta-\eta_{\ell}) \r_{\b}(H-H_{\partial B})\eta^{\dag} O_1 O_2\right],
\ea
where in the second line, we replaced $\eta$ with the truncated operator $\eta_{\ell}$. To simplify this equation, we absorb the coefficient $Z_{\b}(H-H_{\partial B})/Z_{\b}(H)$ into the operators $\eta,\eta_{\ell}$, and define 
\ba
\tilde{\eta}=\left(\frac{Z_{\b}(H-H_{\partial B})}{Z_{\b}(H)}\right)^{1/2}\eta,\quad  \tilde{\eta}_{\ell}=\left(\frac{Z_{\b}(H-H_{\partial B})}{Z_{\b}(H)}\right)^{1/2}\eta_{\ell}.
\ea
Hence, we have
\ba
\left|\Tr\left[\r_{\b}(H)O_1 O_2\right]-\Tr\left[\tilde{\eta}_{\ell}\r_{\b}(H-H_{\partial B})\tilde{\eta}^{\dag}_{\ell}O_1 O_2\right]\right|&\leq\left|\Tr\left[\tilde{\eta} \r_{\b}(H-H_{\partial B})(\tilde{\eta}^{\dag}-\tilde{\eta}_{\ell}^{\dag}) O_1 O_2\right]\right|\nonumber\\
&+\left|\Tr\left[(\tilde{\eta}-\tilde{\eta}_{\ell}) \r_{\b}(H-H_{\partial B})\tilde{\eta}^{\dag} O_1 O_2\right]\right|.
\ea
According to \eqref{eq:p1}, we have $\norm{\eta-\eta_{\ell}} \leq  e^{\a_1 |\partial B|-\a_2 \ell}$ and $\norm{\eta} \leq  e^{\b/2 \norm{H_{\partial B}} }$. Also, \lemref{site_removal} implies $Z_{\b}(H-H_{\partial B})/Z_{\b}(H)\leq e^{\alpha_3 |\partial B|}$ for some constant $\a_3$ that depends on the details of $H$. Using these bounds as well as the Cauchy-Schwarz and H\"older inequalities, we get the following bound for some constants $c'$ and $\a_4$:
\ba
\left|\Tr\left[\r_{\b}(H)O_1 O_2\right]-\Tr\left[\tilde{\eta}_{\ell}\r_{\b}(H-H_{\partial B})\tilde{\eta}^{\dag}_{\ell}O_1 O_2\right]\right|&\leq 2\norm{O_1} \norm{O_2} \norm{\eta-\eta_{\ell}}  \norm{\eta}  \nonumber\\
&\leq c'e^{-\a_4 \ell}.
\ea

To arrive at the last line, we used the fact that $|\partial B|$ in $1$D is just a constant that depends on the range of $H$, and we assumed the truncation length $\ell$ is sufficiently larger than $|\partial B|$.

Note that since we removed the boundary terms $H_{\partial B}$, the Gibbs state decomposes into $\r_{\b}(H-H_{\partial B})=\r_{\b}(H_{\bar{B}})\r_{\b}(H_{B})$, which allows us to write 
\ba
\Tr\left[\tilde{\eta}_{\ell}\r_{\b}(H-H_{\partial B})\tilde{\eta}^{\dag}_{\ell}O_1 O_2\right]=\Tr\left[\r_{\b}(H_B)\tilde{\s}_{\partial B}O_1 O_2\right],
\ea
in which we assume region $B$ is chosen to be wide enough so that both $O_1,O_2$ are sufficiently far from the boundary $\partial B$ compared to length $\ell$. This means $\eta_{\ell}$ does not overlap with $O_1,O_2$. We also define the \emph{unnormalized} boundary state $\tilde{\s}_{\partial B}$ by
\ba
\tilde{\s}_{\partial B}=\tilde{\eta}^{\dag}_{\ell}\tilde{\eta}_{\ell} \Tr_{\bar{B}\setminus{\supp(\eta_{\ell})}}[\r_{\b}(H_{\bar{B}})].
\ea
Notice that $\tilde{\s}_{\partial B}$ is a PSD matrix. To see why, we use the fact that $\Tr_{\bar{B}\setminus{\supp(\eta_{\ell})}}[\r_{\b}(H_{\bar{B}})]$ is a PSD matrix and hence can be written as $W W^{\dag}$ for some operator $W$ supported on $\supp(\eta_{\ell})\cap \bar{B}$. Then it is not hard to see that for any state $\ket{\phi}$, we have
\ba
\bra{\phi}\tilde{\s}_{\partial B}\ket{\phi}=\sum_{i=1}^{\dim(\supp(W))} \bra{i}W^{\dag}\bra{\phi}\tilde{\eta}_{\ell}^{\dag} \tilde{\eta}_{\ell}W\ket{i}\ket{\phi}\geq 0.
\ea
Overall, we have
\ba
\left|\Tr\left[\r_{\b}(H)O_1 O_2\right]-\Tr\left[\r_{\b}(H_B)\tilde{\s}_{\partial B}O_1 O_2\right]\right|\leq c'e^{-\a_4 \ell}.
\ea
Similarly, we can replace $\Tr[\r_{\b}(H)O_i]$ with $\Tr[\r_{\b}(H_B)\tilde{\s}_{\partial B}O_i ]$ up to an exponentially small error in $\ell$,
\ba
\left|\Tr\left[\r_{\b}(H)O_i \right]-\Tr\left[\r_{\b}(H_B)\tilde{\s}_{\partial B}O_i \right]\right|\leq c'e^{-\a_4 \ell},\quad i\in\{1,2\}.
\ea
We can now plug these expressions into the covariance \eqref{eq:v2}. Since $\norm{\Tr[\r_{\b}(H_B) \tilde{\s}_{\partial B} O_i]}  $ is just a constant, we see that there exist constants $c''$ and $\a_5$ such that
\ba
&\big|\Tr\left[ \rho_{\b}(H) O_1 O_2 \right]-\Tr\left[ \rho_{\b}(H) O_1 \right]\Tr\left[ \rho_{\b}(H) O_2 \right] \big| \nonumber\\
&=\big|\left(\Tr\left[\r_{\b}(H_B)\tilde{\s}_{\partial B}O_1 O_2\right]+c'e^{-\a_4 \ell}\right)-\left(\Tr\left[\r_{\b}(H_B)\tilde{\s}_{\partial B}O_1\right]-c'e^{-\a_4 \ell}\right)\left(\Tr\left[ \rho_{\b}(H) \tilde{\s}_{\partial B} O_2 \right]-c'e^{-\a_4 \ell}\right) \big|\nonumber\\
&\leq \big|\Tr\left[\r_{\b}(H_B)\tilde{\s}_{\partial B}O_1 O_2\right]-\Tr\left[\r_{\b}(H_B)\tilde{\s}_{\partial B}O_1\right]\Tr\left[ \rho_{\b}(H) \tilde{\s}_{\partial B} O_2 \right]\big|+c''e^{-\a_5 \ell}.\label{eq:v15}
\ea
We can consider $ \tilde{\s}_{\partial B}O_2$ to be the new operator whose correlation with $O_1$ we want to measure. The operator $\tilde{\s}_{\partial B}O_2$ is still $\dist(O_1,O_2)$ far from $O_1$. Thus, using the bound \eqref{eq:v41} proved in \thmref{The absence of zeros implies the decay of correlations in the non-commuting case}, we get
\ba
\big|\Tr\left[\r_{\b}(H_B)\tilde{\s}_{\partial B}O_1 O_2\right]-\Tr\left[\r_{\b}(H_B)\tilde{\s}_{\partial B}O_1\right]\Tr\left[ \rho_{\b}(H) \tilde{\s}_{\partial B} O_2 \right]\big|\leq c |B| e^{-\dist(O_1,O_2)/\xi}.
\ea
Combined with \eqref{eq:v15}, we have
\ba
\big|\Tr\left[ \rho_{\b}(H) O_1 O_2 \right]-\Tr\left[ \rho_{\b}(H) O_1 \right]\Tr\left[ \rho_{\b}(H) O_2 \right] \big|\leq c |B| e^{-\dist(O_1,O_2)/\xi}+c''e^{-\a_5 \ell}.
\ea
Since all the coefficients in the bound on the right-hand side are constants, it suffices to choose $\ell$ large enough compared to $\dist(O_1,O_2)$ so that it is negligible compared to the $e^{-\dist(O_1,O_2)/\xi}$ term. This is possible because we assumed $\partial{C}$ is sufficiently (but still only constantly) far from $O_1,O_2$. This allows us to get a bound that does not depend on $n$ as before, hence finishing the proof.  
\end{proof} 

\begin{rem}\label{rem:belief propagation in higher dim}
Recall that from \eqref{eq:p1} we know that the error of truncating the belief propagation operator $\eta$ is
\ba
\Norm{\eta - \eta_{\ell}} \leq e^{\a_1 |\partial B|-\a_2 \ell}.
\ea
In our setting, the dependence of the error bound  on $e^{\a_1 |\partial B|}$ makes this result only be applicable when $\L$ is a $1$D lattice. Otherwise, since $|\partial B|$ is proportional to $\diam(B)^{D-1}$, we cannot choose length $\ell$ small enough compared to $\diam(B)$. Hence, we do not get a local operator as required. 
\end{rem}

\section{Exponential decay of correlations implies analyticity}\label{sec:Exponential decay of correlations gives the absence of complex zeros}

In this section, we focus on the converse of \thmref{The absence of zeros implies the decay of correlations in the non-commuting case}. In \secref{Zero-free region implies the exponential decay of correlations}, we showed that the exponential decay of correlations is a \emph{necessary} condition for the analyticity of the free energy. In this section, we ask if this condition is also \emph{sufficient} for the analyticity. This was first established for classical systems by Dobrushin and Shlosman \cite{Dobrushin1}. It appears that the quantum generalization of that proof requires the development of new tools. The goal in this section is to identify these tools. Our contribution is to extend the result of \cite{Dobrushin1} to classical systems that are not translationally invariant and express the proof in a language that is suitable for the quantum case.

Here, for clarity, we consider a simpler version of \condref{analyticity after measurement} that is stated below:
\begin{customcond}{1'}[Analyticity of the free energy]\label{simple Absence of complex zeros} \label{cond:simple absence of zeros}
The free energy of a geometrically-local Hamiltonian $H$ is $\d$-analytic at inverse temperature $\b\in \bbR^+$ if for all $\b'\in\bbC$ such that $|\b'-\b|\leq \d$, the free energy is analytic and there exists a constants $c$ such that
\ba\label{eq:q1}
\left| \log\left( \Tr \left[ e^{ -\b' H} \right] \right)\right|\leq c n.
\ea
\end{customcond}

Recall that in \condref{analyticity after measurement}, we assumed that the free energy of a \emph{post-selected} state is analytic and bounded. In comparison, \condref{simple absence of zeros} only includes partition functions with an \emph{open boundary condition}. For algorithmic purposes, like the one in \secref{Algorithm for estimating the partition function}, this version is sufficient. However, with small modifications, the same proof can be adapted to show \condref{analyticity after measurement} with arbitrary boundary conditions. 

Our goal is to derive \condref{simple absence of zeros} assuming that the correlations in the system decay exponentially. We restate this condition for convenience.
\\

\noindent\textbf{Restatement of \condref{Exponential decay of correlations}.}\label{cond:Restate exponential decay of correlations} \emph{The correlations in the Gibbs state $\rho_{\b}(H)$ of a geometrically-local Hamiltonian decay exponentially if for any local Hermitian operators $O_1$ and $O_2$, there exist constants $\xi$ and $c$ such that}
\ba\label{eq:v2_}
\big|\Tr\left[ \rho_{\b}(H) O_1 O_2 \right]-\Tr\left[ \rho_{\b}(H) O_1 \right]\Tr\left[ \rho_{\b}(H) O_2 \right] \big| \leq c\norm{O_1}\norm{O_2} e^{-\dist(O_1,O_2)/\xi}.
\ea

Although we consider classical systems, we find it more convenient to continue using quantum notation. This also makes it easier to point out where the proof breaks for quantum systems. The reader, however, should note that the terms in the Hamiltonian are all diagonal in a product basis and the projector operators we use basically \emph{fix the value} of classical spins. 

More formally, we prove the following theorem in this section.

\begin{thm}[The decay of correlations implies analyticity for classical systems] \label{thm:The decay of correlations implies the absence of zeros} Let $H=\sum_{i=1}^m H_i$ be a geometrically-local Hamiltonian of a classical spin system, i.e. the local terms $H_i$ are all diagonal in the same product basis. For such a system, the exponential decay of correlations given in \condref{Exponential decay of correlations} implies analyticity of the free energy in \condref{simple absence of zeros}.
\end{thm}

\begin{figure}[h]
  \centering
  
\usetikzlibrary{shapes,arrows,fit,backgrounds,calc}
\tikzstyle{box} = [rectangle, very thick, rounded corners, minimum width=3cm, minimum height=1cm,text centered, draw=black, fill=white, text centered,  inner sep=6pt, inner ysep=5pt]

\tikzstyle{arrow} = [double,double equal sign distance,-implies]

\tikzstyle{block} = [very thick, draw=black, rectangle, text width=1cm, text centered, minimum height=1.2cm, minimum width=.2cm node distance=3cm,fill=white,]

\begin{tikzpicture}[node distance=2.45cm]
\node (box1) [box] [draw, align=center]{
        \condref{simple absence of zeros}\\
        \emph{Analyticity of the free energy}};
\node (box2) [box, below of=box1] [draw, align=center]{
        Complex site removal bound \eqref{eq:q5}};
\node (box3) [box, below of=box2] [draw, align=center]{
        Small relative phase with\\
         different boundary conditions \eqref{eq:q2}};
\node (box4) [box, below of=box3] [draw, align=center]{
        \condref{Exponential decay of correlations}\\
        \emph{Exponential decay of correlations}};

\draw [arrow] (box2) -- node [anchor=west] {\propref{bound on Z with projector}}(box1);
\draw [arrow] (box3) -- node [anchor=east] {\propref{fraction with projectors}}(box2);
\draw [arrow] (box4) -- node [anchor=west] {\propref{small phase from decay}}(box3);

\end{tikzpicture}

\caption{The structure of the proof of \thmref{The decay of correlations implies the absence of zeros}. We follow a series of reductions to show \condref{simple absence of zeros}.}\label{fig:structure of the proof}
\end{figure}

We prove this theorem in multiple steps that are formulated in \proprefthree{bound on Z with projector}{fraction with projectors}{small phase from decay}. An outline of the proof is given in \fig{structure of the proof}. It turns out that \propref{bound on Z with projector}  and \propref{fraction with projectors} continue to hold for commuting Hamiltonians, so we give their statements and proofs for these Hamiltonians. However, for reasons to be highlighted in its proof, \propref{small phase from decay} only holds for classical systems. 

\begin{proof}[\textbf{Proof of \thmref{The decay of correlations implies the absence of zeros}}]
The proof is immediate from the combination of \propref{fraction with projectors}, \propref{bound on Z with projector}, and \propref{small phase from decay}.
\end{proof}
\subsection{Step 1: \condref{simple absence of zeros} from the complex site removal bound}

Our first step, stated in \propref{bound on Z with projector}, is to show how a variant of the complex site removal bound that we discussed in \secref{High temperatures: Fisher zeros} allows us to find an upper bound on the absolute value of the free energy as in \condref{simple absence of zeros}. Compared to the bound \eqref{eq:t0} in \secref{High temperatures: Fisher zeros}, this variant includes setting a non-trivial boundary condition after removing a subset of lattice sites. To avoid subtleties arising from entangled boundary conditions and projectors, we need to give a slightly different proof compared to what we did before \eqref{eq:t2}. 

\begin{prop}[\condref{simple absence of zeros} from the complex site removal bound]\label{prop:bound on Z with projector}
 Let $H=\sum_{k=1}^m H_k$ be a geometrically-local Hamiltonian with mutually commuting terms on lattice $\L$. Let $P$ be a projector acting on $\partial \bar{A}$ where $A\subset \L$ is a region of constant size\footnote{Recall $\partial \bar{A}$ is the boundary of $\bar{A}$ and is inside $A$. For a $(\k,R)$-local $H$, $\partial A=\{v\in \L\setminus A: \exists v'\in A,\ \dist(v-v')\leq R\}$.}. We denote the terms in $H$ acting on $\bar{A}$ or $\partial \bar{A}$ by $H'$ and the real and imaginary parts of $\b \in \bbC$ by $\b_r$ and $\b_i$. Suppose when $|\b_i|\leq \d$ for some sufficiently small $\d$, there exists a constant $c$ such that 
 \ba
\left|\log \left(\frac{\Tr_{\bar{A}\cup A}[e^{-\b H}]}{\Tr_{\bar{A}\cup \partial \bar{A}}\left[e^{-\b H'}P\right]}\right)\right| \leq c. \label{eq:q5}
\ea
Then,
\begin{itemize}
    \item [i.] The observables supported on $A$ like $H_A$ have bounded expectations with respect to the complex perturbed Gibbs state $\r_{\b}(H)$. That is, there exists a constant $c'$ such that $|\Tr\left[H_{A}\r_{\b}(H)\right]|\leq c' \norm{H_A}.\label{eq:q19}$
    \item [ii.] \condref{simple absence of zeros} holds for this system.
\end{itemize}
\end{prop}

\begin{proof}[\textbf{Proof of \propref{bound on Z with projector}}] By using \lemref{site_removal}, we have $|\log (\Tr[e^{-\b_r H}])|\leq O(n)$. Hence to show \eqref{eq:q1}, it is sufficient to show that
\ba\label{eq:q-1}
\left|\log \left(\frac{\Tr[e^{-\b H}]}{\Tr[e^{-\b_r H}]}\right)\right|\leq c n.
\ea
The difference between the numerator and denominator of \eqref{eq:q-1} is the addition of the complex perturbations $\b_i H=\sum_{k=1}^m \b_i H_k$ to the exponent of the numerator. Instead of adding these terms all together, we can add local terms $\b_i H_k$ step by step. We do this by setting up a telescoping series of products such that in each fraction, a new term $\b_i H_k$ is added. We have
\ba\label{q}
\frac{\Tr[e^{-\b H}]}{\Tr[e^{-\b_r H}]}=\frac{\Tr[e^{-\b_r H-i\b_i \sum_{k=1}^m H_k }]}{\Tr[e^{-\b_r H-i\b_i \sum_{k=1}^{m-1} H_k}]}\frac{\Tr[e^{-\b_r H-i\b_i \sum_{k=1}^{m-1} H_k }]}{\Tr[e^{-\b_r H-i\b_i \sum_{k=1}^{m-2} H_k}]}\dots \frac{\Tr[e^{-\b_r H-i\b_i H_1 }]}{\Tr[e^{-\b_r H}]}.
\ea
Hence,
\ba
\left|\log \left(\frac{\Tr[e^{-\b H}]}{\Tr[e^{-\b_r H}]}\right)\right|\leq \sum_{j=0}^{m-1} \left|\log \left(\frac{\Tr[e^{-\b_r H-i\b_i \sum_{k=0}^{j+1} H_k}]}{\Tr[e^{-\b_r H-i\b_i \sum_{k=0}^{j} H_k}]}\right)\right|,
\ea
in which we set $H_0=0$. Since for interactions considered in this paper $m=O(n)$, we can derive the bound in \eqref{eq:q2} by showing for any $j$,
\ba\label{eq:q3}
\left|\log \left(\frac{\Tr[e^{-\b_r H-i\b_i \sum_{k=0}^{j+1} H_k}]}{\Tr[e^{-\b_r H-i\b_i \sum_{k=0}^{j} H_k}]}\right)\right|\leq O(1).
\ea
To do so, we define $\g_j(t)$ for $t\in[0,1]$ to be
\ba
\g_j(t)=\log \left(\Tr[e^{-\b_r H-i\b_i \sum_{k=0}^{j} H_k-i\b_i t H_{j+1})}]\right).
\ea
Then, the left hand side of \eqref{eq:q3} can be written as
\ba
\left|\log \left(\frac{\Tr[e^{-\b_r H-i\b_i \sum_{k=0}^{j+1} H_k}]}{\Tr[e^{-\b_r H-i\b_i \sum_{k=0}^{j} H_k}]}\right)\right|&=|\g_j(1)-\g_j(0)|\nonumber\\
&\leq \max_{t\in[0,1]} \left|\frac{d \g_j(t)}{d t}\right|\nonumber\\
&= |\b_i| \max_{t\in[0,1]} \left|\frac{\Tr[
H_{j+1}e^{-it\b_i H_{j+1}}e^{-\b_r H-i\b_i \sum_{k=0}^{j} H_k}]}{\Tr[e^{-it\b_i H_{j+1}} e^{-\b_r H-i\b_i \sum_{k=0}^{j} H_k}]}\right|.\label{eq:q4}
\ea
For a region $A\subset \L$, let $H_A$ and $H'$ be parts of the Hamiltonian acting on $A$ and $\bar{A}\cup \partial \bar{A}$, respectively. One can see that for any choice of $j$ and $t$, finding an upper bound like the one in \eqref{eq:q4} is equivalent to bounding a local expectation term like
\ba
\Tr\left[H_{A}e^{-(\b_r+it\b_i) H_A}\frac{e^{-\b H'}}{Z_{\b}(H)}\right]=\Tr[H_A \r_{
\b}(H)]
\ea
for some suitable choice of $A$. We also assume, without loss of generality, that \emph{all} local terms in $H'$ are complex perturbed. Using the H\"older inequality, we get
\ba
\left|\Tr\left[H_{A}e^{-(\b_r+it\b_i) H_A}\frac{e^{-\b H'}}{Z_{\b}(H)}\right]\right|&=\left| \Tr_A\left[H_{A}e^{-(\b_r+it\b_i) H_A}\frac{\Tr_{\bar{A}}[e^{-\b H'}]}{Z_{\b}(H)}\right]\right|\nonumber\\
&\leq \norm{H_A} e^{|\b| \norm{H_A}}d^{|\partial \bar{A}|}\NOrm{\frac{\Tr_{\bar{A}}[e^{-\b H'}]}{Z_{\b}(H)}}.
\ea
Since $|A|=O(1)$, we only need to upper bound the largest singular value of $\Tr_{\bar{A}}[e^{-\b H'}]/Z_{\b}(H)$, whose support is only on $\partial \bar{A}$, by a constant. Let $\ket{u}$ and $\ket{v}$ be the left and right singular vectors associated with the largest singular value. We claim that there exists a rank $1$ projector $P$ supported on $\partial \bar{A}$ such that
\ba
\NOrm{\frac{\Tr_{\bar{A}}[e^{-\b H'}]}{Z_{\b}(H)}}&=\Tr_{\partial \bar{A}} \left[\frac{\Tr_{\bar{A}}[e^{-\b H'}]}{Z_{\b}(H)} \ketbra{u}{v}\right]\nonumber\\
&\leq (2+\sqrt{2}) \left|\Tr_{\bar{A}\cup \partial \bar{A}}\left[\frac{e^{-\b H'}}{Z_{\b}(H)}P\right]\right|.\label{eq:q7}
\ea
This can be derived by noting that $\ketbra{u}{v}$ can be decomposed as sum of rank $1$ projectors as follows
\ba
\ketbra{u}{v}=-\frac{1+i}{2}\left(\ketbra{u}{u}+\ketbra{v}{v}\right)+i \ketbra{w^-}{w^-}+\ketbra{w^+}{w^+},
\ea
where $\ket{w^+}=\frac{1}{\sqrt{2}}(\ket{u}+\ket{v})$ and $\ket{w^-}=\frac{1}{\sqrt{2}}(\ket{u}+i\ket{v})$. 

Finally, using the premise of this proposition given in \eqref{eq:q5}, we get both \eqref{eq:q19} and \condref{simple absence of zeros}, which concludes the proof. 
\end{proof}
\subsection{Step 2: The complex site removal bound from the small relative phase condition}\label{sec:Step 2: the complex site removal bound from the small relative phase condition}
\begin{prop}\label{prop:fraction with projectors}
Consider the same setup as that of \propref{bound on Z with projector}. Let $P$ and $Q$ be projectors acting on $\partial \bar{A}$. Let $\t(\d)$ be a complex function depending on $H$, $P$, and $Q$, but constant in $n$ such that for any positive constant $c$, $c|\t(\d)|\geq \d$ for sufficiently small $\d$. We can, for instance, assume $|\t(\d)|=\sqrt{\d}$. Suppose when $|\b_i|\leq \d$ for some sufficiently small $\d$, we have 
\ba\label{eq:q2}
\frac{\Tr_{\bar{A} \cup \partial \bar{A}}[e^{-i\b_i H'}\r_{\b_r}(H'|P)]}{\Tr_{\bar{A} \cup \partial \bar{A}}[e^{-i\b_i H'}\r_{\b_r}(H'|Q)]}=1+|\partial \bar{A}|\t(\d).
\ea
Then, the complex site removal bound \eqref{eq:q5} given in \propref{bound on Z with projector} holds, i.e. 
\ba
\left|\log \left(\frac{\Tr_{\bar{A}\cup A}\left[e^{-\b H}\right]}{\Tr_{\bar{A}\cup \partial \bar{A}}\left[e^{-\b H'}P\right]}\right)\right| \leq c.\label{eq:5'}
\ea
\end{prop}

Before getting to the proof of this proposition, we first state and prove a relevant lemma. 

\begin{lem}\label{lem:unperturbed ratios}
Consider the same definitions as in \propref{fraction with projectors}. The ratio of the \emph{unperturbed} partition functions (with real $\b$) with different boundary conditions can be bounded as
\ba
 \left|\log \left(\frac{\Tr_{\bar{A}\cup \partial \bar{A}}\left[e^{-\b_r H'}Q\right]}{\Tr_{\bar{A}\cup \partial \bar{A}}\left[e^{-\b_r H'}P\right]}\right)\right|\leq c'
\ea
for some constant $c'$ depending on $|\partial{\bar{A}}|$.
\begin{proof}[\textbf{Proof of \lemref{unperturbed ratios}}]
Let $H_{\bar{A}}$ be terms in $H'$ that are acting solely on $\bar{A}$. That is, $H_{\bar{A}}$ commutes with both $P$ and $Q$. We have 
\ba
 \frac{\Tr_{\bar{A}\cup \partial \bar{A}}\left[e^{-\b_r H'}Q\right]}{\Tr_{\bar{A}\cup \partial \bar{A}}\left[e^{-\b_r H'}P\right]}= \frac{\Tr_{\bar{A}\cup \partial \bar{A}}\left[e^{-\b_r H'} Q\right]}{\Tr_{\bar{A}\cup \partial \bar{A}}\left[e^{-\b_r H_{\bar{A}}} Q\right]} \frac{\Tr_{\bar{A}\cup \partial \bar{A}}\left[e^{-\b_r H_{\bar{A}}}P\right]}{\Tr_{\bar{A}\cup \partial \bar{A}}\left[e^{-\b_r H'} P\right]}.
\ea
We can bound both of the ratios on the left side of this equality as follows:
\ba
\left|\frac{\Tr_{\bar{A}\cup \partial \bar{A}}\left[e^{-\b_r H'} Q\right]}{\Tr_{\bar{A}\cup \partial \bar{A}}\left[e^{-\b_r H_{\bar{A}}} Q\right]}\right| =\left|\Tr_{\bar{A}\cup \partial \bar{A}}\left[e^{-\b_r (H'-H_{\bar{A}})}\r_{\b}(H_{\bar{A}}|Q)\right]\right|\leq \Norm{e^{-\b_r (H'-H_{\bar{A}})}}.
\ea
Similarly, we can exchange the role of $P$ and $Q$ to get a lower bound. 
\end{proof}
\end{lem}
\begin{proof}[\textbf{Proof of \propref{fraction with projectors}}]\label{proof:fraction with projectors}
We show how assuming equation \eqref{eq:q2}, we can derive a \emph{lower} and an \emph{upper} bound for 
\ba
\left|\frac{\Tr_{\bar{A}\cup A}\left[e^{-\b H}\right]}{\Tr_{\bar{A}\cup \partial \bar{A}}\left[e^{-\b H'}P\right]}\right|. \label{eq:q10}
\ea
We decompose the expression \eqref{eq:q10} into two parts denoted by $L_1$ and $L_2$ as follows
\ba
\frac{\Tr_{\bar{A}\cup A}\left[e^{-\b H}\right]}{\Tr_{\bar{A}\cup \partial \bar{A}}\left[e^{-\b H'}P\right]}&=\Tr_A\left[e^{-i\b_i H_A}e^{-\b_r H_A}\frac{\Tr_{\bar{A}}\left[e^{-\b H'}\right]}{\Tr_{\bar{A}\cup \partial \bar{A}}\left[e^{-\b H'}P\right]}\right]\nonumber\\
&=L_1+L_2,\label{eq:q15}
\ea
where
\ba
L_1&=\Tr_A\left[e^{-\b_r H_A}\frac{\Tr_{\bar{A}}\left[e^{-\b H'}\right]}{\Tr_{\bar{A}\cup \partial \bar{A}}\left[e^{-\b H'}P\right]}\right]\label{eq:q6'}\\
L_2&=\Tr_A\left[(e^{-i\b_i H_A}-\iden)e^{-\b_r H_A}\frac{\Tr_{\bar{A}}\left[e^{-\b H'}\right]}{\Tr_{\bar{A}\cup \partial \bar{A}}\left[e^{-\b H'}P\right]}\right]\label{eq:q6}.
\ea
All the complex perturbations acting on $A$ are moved to the second part $L_2$ which is analyzed later and shown to have only a small contribution.

Let $\{\ket{\psi_k}\}$ be the set of eigenstates of the operator $H_A$ that span the Hilbert space of $A$. The term $L_1$ can be written as 
\ba
L_1&=\sum_k \bra{\psi_k}e^{-\b_r H_A}\ket{\psi_k} \left[\frac{\Tr_{\bar{A} \cup A}\left[e^{-\b H'}\ketbra{\psi_k}{\psi_k}\right]}{\Tr_{\bar{A}\cup \partial \bar{A}}\left[e^{-\b H'}P\right]}\right]\nonumber\\
&=\sum_k e_k \left[\frac{\Tr_{\bar{A}\cup \partial \bar{A}}\left[e^{-\b H'}\Tr_{A\setminus{\partial\bar{A}}}\ketbra{\psi_k}{\psi_k}\right]}{\Tr_{\bar{A}\cup \partial \bar{A}}\left[e^{-\b H'}P\right]}\right]\nonumber\\
&=\sum_{j,k} e_k r_{j,k}\left[\frac{\Tr_{\bar{A}\cup \partial \bar{A}}\left[e^{-\b H'}Q_{j,k}\right]}{\Tr_{\bar{A}\cup \partial \bar{A}}\left[e^{-\b H'}P\right]}\right],\label{eq:q8}
\ea
where the first line follows from $\{\ket{\psi_k}\}$ spanning the Hilbert space of $A$. In the second line, we denoted $\bra{\psi_k}e^{-\b_r H_A}\ket{\psi_k}$ by positive coefficients $e_k$. In the last line, we used the fact that $\Tr_{A\setminus{\partial\bar{A}}}\ketbra{\psi_k}{\psi_k}$ is a density operator on $\partial \bar{A}$ and can be decomposed into a convex combination of projectors $Q_{j,k}$ supported on $\partial \bar{A}$ with positive coefficients $r_{j,k}$. In other words,
\ba
\Tr_{A\setminus{\partial\bar{A}}}\ketbra{\psi_k}{\psi_k}=\sum_j r_{j,k} Q_{j,k}.
\ea
From the assumption of the theorem given in \eqref{eq:q2} we get
\ba
\left[\frac{\Tr_{\bar{A}\cup \partial \bar{A}}\left[e^{-\b H'}Q_{j,k}\right]}{\Tr_{\bar{A}\cup \partial \bar{A}}\left[e^{-\b H'}P\right]}\right]=\a_{j,k}\left(1+|\partial \bar{A}|\t_{j,k}(\d)\right),\label{eq:q22}
\ea
where 
\ba
\a_{j,k}=\frac{\Tr_{\bar{A}\cup \partial \bar{A}}\left[e^{-\b_r H'}Q_{j,k}\right]}{\Tr_{\bar{A}\cup \partial \bar{A}}\left[e^{-\b_r H'}P\right]}
\ea
is the ratio of the \emph{real} partition functions, and according to \lemref{unperturbed ratios},
\ba
|\log \a_{j,k}|\leq O(|\partial \bar{A}|).
\ea
Hence, we get the following expression for $L_1$:
\ba
L_1= \sum_{j,k} \a_{j,k}r_{j,k} e_k \left(1+|\partial \bar{A}|\t_{j,k}(\d)\right).\label{eq:q11}
\ea
This allows us to find a lower bound on this term. Since all coefficients $\a_{j,k}$, $r_{j,k}$, and $e_k$ are positive constants, Eq. \eqref{eq:q11} is sum of complex numbers with various magnitudes that have small complex phases at most proportional to $|\partial \bar{A}|\t_{j,k}(\d)$. The absolute value of the sum of these complex numbers is at least the sum of their real parts. In particular, since $A$ is a region of \emph{constant} size, by choosing a sufficiently small $\d$ such that $\d |\partial \bar{A}|\ll 1$, we can ensure that the real parts are all positive and add up to some non-zero value. More precisely,
\ba
\left|L_1\right|\geq  \left(\sum_{j,k} \a_{j,k}r_{j,k} e_k\right) \cos \left(c''|\partial \bar{A}|\t(\d)\right)\geq \Omega(1) \quad \text{for } \d\ll |\partial \bar{A}|.\label{eq:q12}
\ea
We can also get an \emph{upper} bound on $|L_1|$ using the expression \eqref{eq:q11}. We have
\ba
\left|L_1\right|\leq  \left(1+|\partial \bar{A}|\right)  \left(\sum_{j,k} \a_{j,k}r_{j,k} e_k\right)\leq O(1)\label{eq:q21}
\ea

Now, we look at the second term $L_2$. Similar to the previous bound, we can find a projector $Q$ and a constant $c'$ such that 
\ba
|L_2|=\left|\Tr_A\left[(e^{-i\b_i H_A}-\iden)e^{-\b_r H_A}\frac{\Tr_{\bar{A}}\left[e^{-\b H'}\right]}{\Tr_{\bar{A}\cup \partial \bar{A}}\left[e^{-\b H'}P\right]}\right]\right|&\leq \norm{e^{-i\b_i H_A}-\iden}\norm{e^{-\b_r H_A}}\frac{\NOrm{\Tr_{\bar{A}}\left[e^{-\b H'}\right]}_1}{\left|\Tr_{\bar{A}\cup \partial \bar{A}}\left[e^{-\b H'}P\right]\right|} \nonumber\\
&\leq c'\d\norm{H_A}d^{|\partial \bar{A}|}e^{|\b| \norm{H_A}}. \label{eq:q9}
\ea
We used a bound similar to \eqref{eq:q22} to get to the last line.

All bounds \eqref{eq:q12}, \eqref{eq:q21}, and \eqref{eq:q9} depend on $|A|$ which is a constant. Also, as $\d$ is made smaller, \eqref{eq:q9} becomes negligible compared to \eqref{eq:q12} or \eqref{eq:q21}. Hence, if $\d$ is chosen to be sufficiently small yet still a constant, we get the desired bounds:
\ba
O(1)\geq |L_1|+|L_2|\geq\left|\frac{\Tr_{\bar{A}}\left[e^{-\b H'}\right]}{\Tr_{\bar{A}\cup \partial \bar{A}}\left[e^{-\b H'}P\right]}\right|\geq|L_1|-|L_2|\geq \Omega(1).
\ea
\end{proof}
\subsection{Step 3: The small relative phase condition from \condref{Exponential decay of correlations}}
\begin{prop}\label{prop:small phase from decay}
 Let $H=\sum_{i=1}^m H_i$ be a geometrically-local Hamiltonian of a \emph{classical} spin system. Suppose the correlations in this system decay exponentially as in \condref{Exponential decay of correlations}. Then, the bound given in \eqref{eq:q2} holds for this system.
\end{prop}

\begin{proof}[\textbf{Proof of \propref{small phase from decay}}]
The proof is by induction. The lattice $\L$ is already divided into regions $A$ and $\bar{A}$ according to \propreftwo{bound on Z with projector}{fraction with projectors}. We further split the region $\bar{A}\cup \partial \bar{A}$ into a constant region $B$ and its complement $\bar{B}$. For reasons that will become clear shortly, it suffices to fix an arbitrary site $x$ on $\partial \bar{A}$ and choose region $B$ such that $\dist(\partial{\bar{B}},x)\gg \xi$, where $\xi$ is the correlation length in \condref{Exponential decay of correlations}. We assume inductively that \eqref{eq:q2} holds for $\bar{B}$. Then, using the decay of correlations, we show that even after adding the contribution of region $B$, Equation \eqref{eq:q2} still holds for the region $\bar{A}\cup \partial \bar{A}=B \cup \bar{B}$.

Since we are considering classical systems, the projectors $P$ and $Q$ set the value of the boundary spins, each of which attains $d$ distinct states, to some fixed values denoted by strings $s_p$ and $s_q$, where $s_{p\ \text{or}\ q}\in [d]^{|\partial \bar{A}|}$. Hence, $P=\ketbra{s_p}{s_p}$ and $Q=\ketbra{s_q}{s_q}$. Assume $s_p$ and $s_q$ differ on $t$ sites. Consider a series of strings $s_1,\dots, s_t$ such that $s_1=s_p$, $s_t=s_q$, and $s_i$ and $s_{i+1}$ differ only on one site. Denote the corresponding projectors by $P_1, P_2,\dots, P_t$. We can set up a telescoping product for \eqref{eq:q2} as follows:
\ba
&\frac{\Tr_{\bar{A} \cup \partial \bar{A}}[e^{-i\b_i H'}\r_{\b_r}(H'|P)]}{\Tr_{\bar{A} \cup \partial \bar{A}}[e^{-i\b_i H'}\r_{\b_r}(H'|Q)]}\nonumber\\
&=\frac{\Tr_{\bar{A} \cup \partial \bar{A}}[e^{-i\b_i H'}\r_{\b_r}(H'|P_1)]}{\Tr_{\bar{A} \cup \partial \bar{A}}[e^{-i\b_i H'}\r_{\b_r}(H'|P_2)]} \frac{\Tr_{\bar{A} \cup \partial \bar{A}}[e^{-i\b_i H'}\r_{\b_r}(H'|P_2)]}{\Tr_{\bar{A} \cup \partial \bar{A}}[e^{-i\b_i H'}\r_{\b_r}(H'|P_3)]}\dots \frac{\Tr_{\bar{A} \cup \partial \bar{A}}[e^{-i\b_i H'}\r_{\b_r}(H'|P_{t-1})]}{\Tr_{\bar{A} \cup \partial \bar{A}}[e^{-i\b_i H'}\r_{\b_r}(H'|P_t)]}.\label{eq:q41}
\ea
One can see that to get the desired bound in \eqref{eq:q2}, it is enough to show the following bound on these ratios:
\ba
\frac{\Tr_{\bar{A} \cup \partial \bar{A}}[e^{-i\b_i H'}\r_{\b_r}(H'|P_{i})]}{\Tr_{\bar{A} \cup \partial \bar{A}}[e^{-i\b_i H'}\r_{\b_r}(H'|P_{i+1})]}=1+\t(\d) \label{eq:q14}
\ea
for $\t(\d)$ satisfying the conditions given in \propref{bound on Z with projector}. This is why we define region $B$ around a \emph{single} site on $\partial\bar{A}$.

 To simplify the notation, we keep using $P,Q$ instead of $P_i,P_{i+1}$ for the rest of the proof bearing in mind that they differ on one site. In order to show \eqref{eq:q14}, we change the left-hand side to a slightly different expression that makes it easier to see the connection to the decay of correlations. We have 
\ba
\frac{\Tr_{\bar{A} \cup \partial \bar{A}}[e^{-i\b_i H'}\r_{\b_r}(H'|P)]}{\Tr_{\bar{A} \cup \partial \bar{A}}[e^{-i\b_i H'}\r_{\b_r}(H'|Q)]}
&=1+\frac{\Tr_{\bar{A} \cup \partial \bar{A}}\left[e^{-i\b_i H'}\left(\r_{\b_r}(H'|P)-\r_{\b_r}(H'|Q)\right)\right]}{\Tr_{\bar{A} \cup \partial \bar{A}}\left[e^{-i\b_i H'}\r_{\b_r}(H'|Q)\right]}.
\ea
To derive \eqref{eq:q2}, we can alternatively show
\ba
\frac{\Tr_{\bar{A} \cup \partial \bar{A}}\left[e^{-i\b_i H'}\left(\r_{\b_r}(H'|P)-\r_{\b_r}(H'|Q)\right)\right]}{\Tr_{\bar{A} \cup \partial \bar{A}}\left[e^{-i\b_i H'}\r_{\b_r}(H'|Q)\right]}=  \t(\d).\label{eq:q13}
\ea
 The steps that we take to prove this equation are very similar to the ones in the proof of \propref{bound on Z with projector}. Recall that $H'$ consists of the terms in $H$ that act on $\bar{A}\cup \partial \bar{A}$. Similarly, let $H''$ be part of $H'$ that acts on $\bar{B} \cup \partial \bar{B}$. We also define $T$ to be a projector (which again assigns a value from $[d]$ to the boundary spins) supported on $\partial\bar{B}$.
 
 We divide both the numerator and the denominator of $\eqref{eq:q13}$ by $\Tr_{\bar{B} \cup \partial \bar{B}}[e^{-i\b_i H''}\r_{\b_r}(H''|T)]$. This does not change the fraction but allows us to use the induction hypothesis. Similar to what we did in \eqref{eq:q15}, we split the numerator into two parts, denoted by $M_1$ and $M_2$, such that the complex perturbations acting on $B$ are all moved to $M_2$. We get
\ba
\frac{\Tr_{\bar{A} \cup \partial \bar{A}}\left[e^{-i\b_i H'}\left(\r_{\b_r}(H'|P)-\r_{\b_r}(H'|Q)\right)\right]}{\Tr_{\bar{B} \cup \partial \bar{B}}\left[e^{-i\b_i H''}\r_{\b_r}(H''|T)\right]}=M_1+M_2,\label{eq:q40}
\ea
where
\ba
M_1&= \frac{\Tr_{\bar{A} \cup \partial \bar{A}}\left[e^{-i\b_i H''} \left(\r_{\b_r}(H'|P)-\r_{\b_r}(H'|Q)\right)\right]}{\Tr_{\bar{B} \cup \partial \bar{B}}\left[e^{-i\b_i H''}\r_{\b_r}(H''|T)\right]}\nonumber\\
M_2&= \frac{\Tr_{\bar{A} \cup \partial \bar{A}}\left[e^{-i\b_i H''}(e^{-i\b_i (H'-H'')}-\iden)\left(\r_{\b_r}(H'|P)-\r_{\b_r}(H'|Q)\right)\right]}{\Tr_{\bar{B} \cup \partial \bar{B}}\left[e^{-i\b_i H''}\r_{\b_r}(H''|T)\right]}.
\ea
Now we use the crucial fact that $\r_{\b_r}(H'|P\ \text{or}\ Q)$ is a classical probability distribution that has the \emph{Markov property}. In other words, 
\begin{lem}\label{lem:conditional probability}
For any \emph{diagonal} operator $O$ supported on $\bar{B}\cup \partial \bar{B}$, we have 
\ba
\Tr_{\bar{A} \cup \partial \bar{A}}[O\ \r_{\b_r}(H'|P)]=\sum_{s\in [d]^{|\partial \bar{B}|}}\Tr_{\bar{B}\cup \partial \bar{B}}\left[O\ \r_{\b_r}(H''|P_s)\right] \Tr_{\bar{A} \cup \partial \bar{A}}[P_s\ \r_{\b_r}(H'|P)],\label{eq:q16}
\ea
where $s$ denotes the state of the spins on $\partial \bar{B}$ and $P_s$ is the corresponding projector. 
\end{lem}
This equality follows from the law of total probability. The term $\Tr_{\bar{A} \cup \partial \bar{A}}[P_s\ \r_{\b_r}(H'|P)]$ is the probability of the boundary spins being in state $s$ while $\Tr_{\bar{B}\cup \partial \bar{B}}\left[O\ \r_{\b_r}(H''|P_s)\right]$ is the expectation value of $O$ \emph{conditioned} on the state of the boundary spins. The fact that we only need to condition on the value of the \emph{boundary} spins follows from the Markov property of the Gibbs distribution. We postpone a more detailed proof of Eq. \eqref{eq:q16} until after the end of this proof. 

Using \eqref{eq:q16}, the term $M_1$ can be written as
\ba
M_1&=\sum_{s\in [d]^{|\partial \bar{B}|}}\left(\frac{\Tr_{\bar{B}\cup \partial \bar{B}}\left[e^{-i\b_i H''}\ \r_{\b_r}(H''|P_s)\right]}{\Tr_{\bar{B} \cup \partial \bar{B}}\left[e^{-i\b_i H''}\r_{\b_r}(H''|T)\right]} -1\right)\Tr_{\bar{A} \cup \partial \bar{A}}[P_s\ \r_{\b_r}(H'|P)]\nonumber\\
&-\sum_{s\in [d]^{|\partial \bar{B}|}}\left(\frac{\Tr_{\bar{B}\cup \partial \bar{B}}\left[e^{-i\b_i H''}\ \r_{\b_r}(H''|P_s)\right]}{\Tr_{\bar{B} \cup \partial \bar{B}}\left[e^{-i\b_i H''}\r_{\b_r}(H''|T)\right]} -1\right)\Tr_{\bar{A} \cup \partial \bar{A}}[P_s\ \r_{\b_r}(H'|Q)]\nonumber\\
&=\sum_{s\in [d]^{|\partial \bar{B}|}}\left(\frac{\Tr_{\bar{B}\cup \partial \bar{B}}\left[e^{-i\b_i H''}\ \r_{\b_r}(H''|P_s)\right]}{\Tr_{\bar{B} \cup \partial \bar{B}}\left[e^{-i\b_i H''}\r_{\b_r}(H''|T)\right]} -1\right)\left(\Tr_{\bar{A} \cup \partial \bar{A}}[P_s\ \r_{\b_r}(H'|P)]-\Tr_{\bar{A} \cup \partial \bar{A}}[P_s\ \r_{\b_r}(H'|Q)]\right).
\ea
For later convenience, we added and subtracted $1$ in the first equality. We can now use the induction hypothesis \emph{and} the exponential decay of correlations property to bound $M_1$. From the induction assumption \eqref{eq:q2} we see that for sufficiently small $\d$
\ba
\frac{\Tr_{\bar{B}\cup \partial \bar{B}}\left[e^{-i\b_i H''}\ \r_{\b_r}(H''|P_s)\right]}{\Tr_{\bar{B} \cup \partial \bar{B}}\left[e^{-i\b_i H''}\r_{\b_r}(H''|T)\right]}=1+|\partial \bar{B}|\t(\d)\label{eq:q17}.
\ea
Moreover, we let $x\in B$ be the site on which $P$ and $Q$ differ. Then, the assumption of the exponential decay of correlations \eqref{eq:v2_} implies
\ba
\left|\Tr_{\L \setminus\partial \bar{B}}[ \r_{\b_r}(H'|P)]-\Tr_{\L \setminus \partial \bar{B}}[ \r_{\b_r}(H'|Q)]\right|\leq c e^{-\dist(x,\partial \bar{B})/\xi}\label{eq:q18_}.
\ea
Overall, \eqref{eq:q17} and \eqref{eq:q18_} show that $|M_1|$ can be bounded as follows:
\ba
|M_1|\leq c|\t(\d)| |\partial \bar{B}|e^{-\dist(x,\partial \bar{B})/\xi}.
\ea
Similarly, one can show that 
\ba
M_2&=\nonumber\\
&\sum_{s\in [d]^{|B|}}\left(\frac{\Tr_{\bar{B}\cup \partial \bar{B}}\left[e^{-i\b_i H''}\ \r_{\b_r}(H''|P_s)\right]}{\Tr_{\bar{B} \cup \partial \bar{B}}\left[e^{-i\b_i H''}\r_{\b_r}(H''|T)\right]}\right)\Tr_{\bar{A} \cup \partial \bar{A}}[(e^{-i\b_i (H'-H'')}-\iden) P_s\ (\r_{\b_r}(H'|P)-\r_{\b_r}(H'|Q))],
\ea
which again by using \eqref{eq:q2} and $\sum_{s\in [d]^{|B|}}\Tr_{\bar{A} \cup \partial \bar{A}}[P_s\ \r_{\b_r}(H'|P)]=1$ can be bounded as
\ba
|M_2|\leq c''\d\norm{H_B}(1+|\partial \bar{B}||\t(\d)|).
\ea
We next analyze the denominator of \eqref{eq:q13} that similar to the numerator is first divided by $\Tr_{\bar{B} \cup \partial \bar{B}}\left[e^{-i\b_i H''}\r_{\b_r}(H''|T)\right]$. We can follow similar arguments to \secref{Step 2: the complex site removal bound from the small relative phase condition} to show that for sufficiently small $\d$, we can lower bound this term by a constant:
\ba
\frac{\Tr_{\bar{A} \cup \partial \bar{A}}\left[e^{-i\b_i H'}\r_{\b_r}(H'|Q)\right]}{\Tr_{\bar{B} \cup \partial \bar{B}}\left[e^{-i\b_i H''}\r_{\b_r}(H''|T)\right]}\geq \Omega(1).\label{eq:q18}
\ea

Finally, we can insert these bounds in \eqref{eq:q13} to get the following upper bound for some constants $c_1$ and $c_2$:
\ba
\left|\frac{\Tr_{\bar{A} \cup \partial \bar{A}}\left[e^{-i\b_i H'}\left(\r_{\b_r}(H'|P)-\r_{\b_r}(H'|Q)\right)\right]}{\Tr_{\bar{A} \cup \partial \bar{A}}\left[e^{-i\b_i H'}\r_{\b_r}(H'|Q)\right]}\right|\leq 
c_1|\t(\d)| |\partial \bar{B}|e^{-\dist(x,\partial \bar{B})/\xi}+c_2\d\norm{H_B}(1+|\partial \bar{B}||\t(\d)|).
\ea
Since $\t(\d)$ can be chosen as $\sqrt{\d}$ for instance, for a fixed $\dist(x,\partial \bar{B})$,  we can always choose $\d$ small enough such that 
\ba
 \frac{c_2\d\norm{H_B}(1+|\partial \bar{B}||\t(\d)|)}{c_1|\t(\d)| |\partial \bar{B}|e^{-\dist(x,\partial \bar{B})/\xi}}\leq c_3
\ea
for some constant $c_3\leq 1$. We can also choose $\dist(x,\partial \bar{B})$ sufficiently large enough so that
\ba
c_1|\t(\d)| |\partial \bar{B}|e^{-\dist(x,\partial \bar{B})/\xi}\leq |\t(\d)|.
\ea
Without the term $e^{-\dist(x,\partial \bar{B})/\xi}$ that originates from the decay of correlations property, we could not ensure that the bound $|\t(\d)|$ is recovered after the induction step. 
\end{proof}
Here, we prove \lemref{conditional probability} that we mentioned during the proof of \propref{small phase from decay}. We restate the lemma for convenience.\\ 

\noindent\textbf{Restatement of \lemref{conditional probability}.}
\emph{Consider the same setup as in \propref{small phase from decay} in which we restrict ourselves to classical Hamiltonians. For any diagonal operator $O$ supported on $\bar{B}\cup \partial \bar{B}$, we have}
\ba
\Tr_{\bar{A} \cup \partial \bar{A}}[O\ \r_{\b_r}(H'|P)]=\sum_{s\in [d]^{|\partial \bar{B}|}}\Tr_{\bar{B}\cup \partial \bar{B}}\left[O\ \r_{\b_r}(H''|P_s)\right] \Tr_{\bar{A} \cup \partial \bar{A}}[P_s\ \r_{\b_r}(H'|P)],\label{eq:q16_}
\ea
\emph{where $s$ denotes the state of the spins on $\partial \bar{B}$ and $P_s$ is the corresponding projector. }

\begin{proof}[\textbf{Proof of \lemref{conditional probability}}]
We have
\ba
\Tr_{\bar{A} \cup \partial \bar{A}}\left[O\ \r_{\b_r}(H'|P)\right]&=\Tr_{\bar{A} \cup \partial \bar{A}}\left[O_{\bar{B}\cup \partial \bar{B}}\ \frac{e^{-\b_r H'}P}{\Tr_{\bar{A} \cup \partial \bar{A}}\left[e^{-\b_r H'}P\right]}\right]\nonumber\\
&=\sum_{s\in [d]^{|\partial \bar{B}|}}\Tr_{\bar{A} \cup \partial \bar{A}}\left[O_{\bar{B}\cup \partial \bar{B}}\ \frac{e^{-\b_r H''}P_s}{\Tr_{\bar{B}\cup \partial \bar{B}}[e^{-\b H''}P_s]} \frac{P_s  e^{-\b (H'-H'')}P\ \Tr_{\bar{B} \cup \partial \bar{B}}[e^{-\b H''}P_s]}{\Tr_{\bar{A} \cup \partial \bar{A}}\left[e^{-\b_r H'}P\right]}  \right]\nonumber\\
&=\sum_{s\in [d]^{|\partial \bar{B}|}}\Tr_{\bar{B}\cup \partial \bar{B}}\left[O_{\bar{B}\cup \partial \bar{B}}\ \frac{e^{-\b_r H''}P_s}{\Tr_{\bar{B}\cup \partial \bar{B}}[e^{-\b H''}P_s]}\right] \Tr_{\bar{A} \cup \partial \bar{A}}\left[P_s\ \frac{e^{-\b H'}P}{\Tr_{\bar{A} \cup \partial \bar{A}}[e^{-\b H'}P]}\right]\nonumber\\
&=\sum_{s\in [d]^{|\partial \bar{B}|}}\Tr_{\bar{B}\cup \partial \bar{B}}\left[O_{\bar{B}\cup \partial \bar{B}}\ \r_{\b_r}(H''|P_s)\right] \Tr_{\bar{A} \cup \partial \bar{A}}[P_s\ \r_{\b_r}(H'|P)],
\ea
\end{proof}

\begin{rem}
A first step in generalizing the proof of \thmref{The decay of correlations implies the absence of zeros} to the quantum case would be to consider commuting Hamiltonians. While some parts of the proof already apply to these systems, the one in \propref{small phase from decay} does not immediately go through. One issue is that the decomposition \eqref{eq:q41} does not have a quantum counterpart. In particular, when comparing the effect of two entangled boundary projectors, we cannot write a telescoping product that reduces this to comparing local projectors. Perhaps by using the commutativity of the terms in the Hamiltonian, we can find a structure in the projectors that allows us to overcome this problem. We leave this for future work.
\end{rem}

\section{Extrapolating from high external fields and Lee-Yang zeros}\label{sec:Extrapolating from high external fields and Lee-Yang zeros}
In this section, we study spin systems whose interactions are described by two- or one-body terms. For qubits, such systems are generally described by Hamiltonians of the form
\ba\label{eq:u_two_local}
H(\mu)=-\sum_{\substack{(i,j)\in E \\ a,b\in\{x,y,z\}}} J_{ij}^{ab} \s_a \ot \s_b-\sum_{i \in V} (h^x_i X_i +h^y_i Y_i +\mu h^z_i Z_i),
\ea
where $J_{ij}^{ab}, h_i^a, \mu \in \bbR$ and $\s_a\in\{X,Y,Z,\iden\}$ are Pauli matrices. The interaction graph, as usual, is denoted by $G=(V,E)$ with $|V|=n$ and $|E|=m$. Physically, the two-body interactions $J_{ij}^{ab}$ are due to the \emph{coupling} between the spins of the particles on adjacent sites, whereas the one-body terms $h_{i}^a$ characterize the interaction of spins with some \emph{external magnetic field}.

\begin{rem}
For later convenience, we introduce an extra factor $\mu$ before the $Z_i$ terms in \eqref{eq:u_two_local}. One can think of $\mu$ as the maximum strength of the external field in the $z$-direction. As explained below, this parameter plays the same role as $\b$ in the extrapolation algorithm of \secref{Algorithm for estimating the partition function}.
\end{rem}

In \secref{Algorithm for estimating the partition function}, we developed approximation algorithms for the partition function of a quantum many-body systems by extrapolating from high to low temperatures. In this section, we again use the idea of extrapolation, but this time our parameter of interest is $\mu$, the magnitude of the one-body terms in the $z$-direction. The physical motivation for this approach is that when the system is subject to a large enough external field in a specific direction (the $z$-direction in our case), all the spins align themselves in that direction, and estimating the properties of the system becomes trivial. However, as we move to smaller fields, the other interaction terms between the particles gain significance, making the problem non-trivial. 

In order to apply the extrapolation algorithm in \secref{Algorithm for estimating the partition function}, we need to know the locus of the complex zeros of the partition function as a function of the external field $\mu$. As mentioned in \secref{intro}, these are called Lee-Yang zeros. We can exactly determine the locus of these zeros when the Hamiltonian \eqref{eq:u_two_local} describes a \emph{ferromagnetic} system, i.e. when the neighboring spins tend to align along the same direction. This is a result of Suzuki and Fisher \cite{Suzuki_XYZ}. There, by generalizing the result of Lee and Yang \cite{Lee-yang}, they show that all the complex zeros lie on the imaginary axis in the $\mu$-plane. \thmref{generalized_Lee_yang} covers this result. 

The key step is to map the quantum system to a classical spin system using the \emph{quantum-to-classical mapping} (see for example \cite{Suzuki_XYZ,bravyi_classicalmapping} ). Then, by the result of \thmref{Multivariate Hurwitz' theorem}, instead of studying the zeros of the quantum system, we can focus on the zeros of a classical system. 

The classical spin system that we obtain involves $1$-, $2$-, and $4$-body terms in its Hamiltonian. We represent the terms in the Hamiltonian with functions $V_{1,i}$, $V_{2,i}$, and $V_{4,i,j}$ that assign possibly complex numbers to their input spins. The indices of these functions refer to the number of particles that they act on and the coefficients of the original quantum Hamiltonian that they depend on.

\begin{prop}[Quantum-to-classical mapping, cf. \cite{Suzuki_XYZ}]\label{prop:quantum_to_classical}
Consider a $2$-local Hamiltonian $H$ as in Eq. \eqref{eq:u_two_local}. Let $z_i=e^{\b \mu h_i^z/\eta}$ and $\e=\b/\eta$. This Hamiltonian can be mapped to a $4$-local classical spin model involving $n'=n\eta$ spins $s\in\{\pm1\}$ with the interactions of the form $V_{1,i}: \{\pm 1\}\rightarrow \bbC$, $V_{2,i}:\{\pm 1\}^{2}\rightarrow \bbC$, and $V_{4,i,j}:\{\pm 1\}^{4}\rightarrow \bbC$ such that $\exp(V_{1,i}(s_a))=z_i^{s_a}$ and
\ba
&\sum_{s_a,s_b\in \{\pm 1 \}} \exp\left(V_{2,i}(s_a,s_b)\right)\ket{s_a}\bra{s_b}=\begin{psmallmatrix}1 & \e (h_i^{x}+ih^{y}_i) \\ \e (h_i^{x}-ih_i^{y}) & 1\end{psmallmatrix},\nonumber\\
&\sum_{s_a,s_b,s_{a'},s_{b'}\in \{\pm 1 \}} \exp\left(V_{4,i,j}(s_a,s_b,s_{a'},s_{b'})\right)\ket{s_{a},s_{a'}}\bra{s_b,s_{b'}}=\nonumber\\
&\begin{psmallmatrix}1+\e J_{ij}^{zz} & \e (-i J_{ij}^{zy}+J_{ij}^{zx}) & \e(-i J_{ij}^{yz}+J_{ij}^{xz}) & \e(J_{ij}^{xx}-J_{ij}^{yy}-i J_{ij}^{xy}-i J_{ij}^{yx}) \\ \e (i J_{ij}^{zy}+J_{ij}^{zx}) & 1-\e J_{ij}^{zz} & \e(J_{ij}^{xx}+J_{ij}^{yy}+i J_{ij}^{xy}-i J_{ij}^{yx}) & \e (i J_{ij}^{yz}-J_{ij}^{xz}) \\ \e (i J_{ij}^{yz}+J_{ij}^{xz}) & \e(J_{ij}^{xx}+J_{ij}^{yy}-i J_{ij}^{xy}+i J_{ij}^{yx}) & 1-\e J_{ij}^{zz} & \e (i J_{ij}^{zy}-J_{ij}^{zx}) \\ \e (J_{ij}^{xx}-J^{yy}_{ij}+i J_{ij}^{xy}+i J_{ij}^{yx}) & \e (-i J_{ij}^{yz}-J_{ij}^{xz}) & \e (-i J_{ij}^{zy}-J_{ij}^{zx})&1+\e J_{ij}^{zz}\end{psmallmatrix}.\label{eq:w2}
\ea
The partition function of this classical system is of the form
\ba
Z_{c\ell}(\mu)=\sum_{s_1,\dots,s_{\eta}\in\{\pm1\}} \exp\left( \sum_{\substack{i\in V\\a\in E_{1,i}}} V_{1,i}(s_a)+\sum_{\substack{i\in V 
\\ (a,b)\in E_{2,i}}} V_{2,i}(s_a,s_b)+\sum_{\substack{(i,j)\in E \\(a,b,a',b')\in E_{4,i,j}}}V_{4,i,j}(s_a,s_b,s_{a'},s_{b'})\right),
\ea
where $E_{1,i}$, $E_{2,i}$, and $E_{4,i,j}$ are certain unordered subsets of vertices that depend on the choice of $i,j$ (see the remark below), and we included the effective temperature of the classical system in the coefficients $V_{1,i}$, $V_{2,i}$, and $V_{4,i,j}$. Moreover, in the limit $\eta \rightarrow \infty$, the partition function of the classical model uniformly converges to that of the quantum system.
\end{prop}
\begin{rem}\label{rem:classical interaction graph}
The details of the interaction (hyper)graph of the classical system in \propref{quantum_to_classical} is not important for our purposes. We can think of this graph as $\eta$ copies of the original interaction graph $G=(V,E)$ stacked on top of each other. These copies are coupled together by the application of $V_{1,i},V_{2,i}$, and $V_{4,i,j}$. While the interaction terms like $V_{1,i}$ apply to all vertices, the terms $V_{2,i}$ act on a vertex in one of the copies of $G$ and its clones in the neighboring graphs. The set $E_{2,i}$ denotes the set of all such two vertices that $V_{2,i}$ acts on. Similarly, the set $E_{4,i,j}$ corresponds to all four vertices that interact through $V_{4,i,j}$.
\end{rem}

In \propref{quantum_to_classical}, the dependency on $\mu$ only appears in the $1$-body terms $V_{1,i}$ and parameters $z_i$. Also, since we do not rely on sampling algorithms, we do not restrict ourselves to \emph{stoquastic} Hamiltonians as in \cite{Bravyi_stoquastic} or \cite{Bravyi_ferro}, but we later put constraints on the coefficients $J_{ij}^{ab}$ to make the Hamiltonian ferromagnetic.
\subsection{Complex zeros of ferromagnetic systems}\label{sec:Complex zeros of the ferromagnetic systems}
We now state a \emph{generalized} Lee-Yang theorem that characterizes the locus of the complex zeros of certain classical spin systems.

\begin{thm}[Generalized Lee-Yang theorem, cf. \cite{Suzuki_XYZ}]\label{thm:generalized_Lee_yang}
Consider the classical spin system described in \propref{quantum_to_classical} or more generally one that satisfies the following conditions: 
\begin{itemize}
    \item[(i)] 
    \ba\label{eq:w5}
  V_{2,i}(-s_a,-s_b)&=V_{2,i}^*(s_a,s_b)\nonumber\\
  V_{4,i,j}(-s_{a},-s_{b},-s_{a'},-s_{b'})&=V_{4,i,j}^*(s_a,s_b,s_{a'},s_{b'})
    \ea
    \item[(ii)]
    \ba \label{eq:w5_}
    |\exp\left(V_{2,i}(+1,+1)\right)|&\geq \frac{1}{4} \sum_{s_a,s_b\in\{\pm 1\}} |\exp\left(V_{2,i}(s_i,s_j)\right)|\nonumber\\
    |\exp\left(V_{4,i,j}(+1,+1,+1,+1)\right)|&\geq \frac{1}{4} \sum_{s_a,s_b,s_{a'},s_{b'}\in\{\pm 1\}} |\exp\left(V_{4,i,j}(s_a,s_b,s_{a'},s_{b'})\right)|.
    \ea
\end{itemize}
 Let $Z_{c\ell}(\mu)$ be the partition function of this system as a function of $\mu$ for a fixed $\b$. Then, the zeros of this partition function, i.e. the solutions of $Z_{c\ell}(\mu)=0$, are all on the imaginary axis in the complex $\mu$-plane, that is, $\re(\mu)=0$. 
\end{thm}
\begin{proof}[\textbf{Proof of \thmref{generalized_Lee_yang}}]
Refer to \cite{Suzuki_XYZ} for the detailed proof of this proposition. Here we only sketch one of the main ideas in their proof. 

For simplicity and in order to roughly see why conditions (i) and (ii) are sufficient for the zeros of the partition function to lie on the imaginary axis, we neglect the $V_{4,i,j}$ terms and focus on the $V_{1,i}$ and $V_{2,i}$ interactions. Recall that $z_i=e^{\b \mu h_i^z/\eta}$ and $\exp(V_{1,i}(s_a))=z_i^{s_a}$. Then, $Z_{ij}(z_i,z_j)=\sum_{s_a,s_b\in \{\pm 1\}} z_i^{s_a} z_j^{s_b} \exp(V_{2,i}(s_i,s_j))$ is proportional to the partition function of the system when all spins except $s_a$ and $s_b$ are fixed to some certain values $\{s_k\}_{k\neq a,b}$. 

Consider the solutions of $Z_{ij}(z_i,z_j)=0$. It is shown in \cite{Suzuki_XYZ} that if such a solution satisfies $|z_j|> 1$ and $|z_i|> 1$, then we can find another solution such that $|z_j|=1$ and $|z_i|> 1$ (a similar result holds for $|z_i|, |z_j|<1$).

Here, we show that when $|z_j|=1$, we also necessarily have  $|z_i|=1$. Since $z_i$ and $z_j$ depend on $\mu$ through $z_i=e^{\b \mu h_i^z/\eta}$, we see that the partition function can only vanish when $\re(\mu)=0$. Although we do not show it here, it turns out that this condition is actually sufficient to show that the whole partition function, without any fixed spins, also has complex zeros only on the imaginary axis. 

We have
\ba\label{eq:5}
Z_{ij}(z_i,z_j)=\left(\sum_{s_b\in\{\pm 1\}} z_j^{s_b} \exp\left(V_{2,i}(+1,s_b)\right)\right)z_i+\left(\sum_{s_b\in\{\pm 1\}} z_j^{s_b} \exp\left(V_{2,i}(-1,s_b)\right)\right)z_i^{-1}.
\ea
Using the condition (i) in (ii) we see that 
\ba
\left|\exp(V_{2,i}(+1,+1))\right|\geq \left|\exp(V_{2,i}(+1,-1))\right|.
\ea
If we consider $|z_j|= 1$, this implies $\sum_{s_b\in\{\pm 1\}} z_j^{s_b} \exp(V_{2,i}(+1,s_b))\neq 0$. We use this in Eq. \eqref{eq:5} to find the solutions of $Z_{ij}(z_i,z_j)=0$ for some $|z_j|=1$. We get 
\ba
|z_i|^2=\frac{|\sum_{s_b\in\{\pm 1\}} z_j^{s_b} \exp\left(V_{2,i}(-1,s_b)\right)|}{|\sum_{s_b\in\{\pm 1\}} z_j^{s_b} \exp\left(V_{2,i}(+1,s_b)\right)|},
\ea
but another application of condition (i) implies $|z_i|=1$ as desired. The rest of the proof for the whole partition function involves a recursive use of this conclusion and shows that the location of the zeros remains on the imaginary axis when different interactions are summed over in the partition function. 
\end{proof}

\begin{rem}
Instead of $\mu$, it is common to consider the partition function as a function of $e^{\mu}$. In this case, the complex zeros are located on the unit circle in the $e^{\mu}$-plane. Hence, the Lee-Yang theorem is often called the \emph{circle theorem}.
\end{rem}

The connection between \thmref{generalized_Lee_yang} and quantum ferromagnetic systems is established through the following theorem. 

\begin{thm}[Zeros of ferromagnetic systems, cf. \cite{Suzuki_XYZ}]\label{thm:t_higH_field}
Let $H(\mu)$ be a 2-local Hamiltonian as in Eq. \eqref{eq:u_two_local} with $J_{ij}^{xz},J_{ij}^{zx},J_{ij}^{yz},J_{ij}^{zy}=0$ defined over an arbitrary interaction graph that is not necessarily geometrically local. Suppose $h^{z}_i\geq 0$, and additionally, the following constraint is satisfied by the coefficients:
\ba\label{eq:stoq_constraints}
J_{ij}^{zz}\geq \frac{1}{2}\left[\left(J_{ij}^{xx}-J_{ij}^{yy}\right)^2 + \left(J_{ij}^{xy}+ J_{ij}^{yx}\right)^2\right]^{\frac{1}{2}}
+\frac{1}{2} \left[\left(J_{ij}^{xx} + J_{ij}^{yy}\right)^2 + \left(J_{ij}^{xy}- J_{ij}^{yx}\right)^2\right]^{\frac{1}{2}}.
\ea
Then, the partition function of this system only vanishes when $\re(\mu)=0$. 

When $J_{ij}^{xy}=J_{ij}^{yx}=0$, this condition simplifies to 
\ba\label{eq:w4}
J_{ij}^{zz}\geq |J_{ij}^{yy}|,\quad J_{ij}^{zz}\geq |J_{ij}^{xx}|.
\ea
This characterizes the ferromagnetic Heisenberg model given by
\ba\label{eq:Heisenbrg}
H=-\sum_{(i,j)\in E} \left(J_{ij}^{xx}X_i X_j+J_{ij}^{yy}Y_i Y_j+J_{ij}^{zz}Z_i Z_j\right)-\sum_{i\in V}\left(h_i^x X_i+h_i^y Y_i+\mu h_i^z Z_i\right).
\ea
\end{thm}
\begin{proof}[\textbf{Proof of \thmref{t_higH_field}}]
The proof follows by applying \propref{quantum_to_classical} to map the quantum system \eqref{eq:u_two_local} to the classical system in \eqref{eq:w2}. One can see that if the quantum system satisfies \eqref{eq:stoq_constraints}, then the corresponding classical system satisfies the conditions \eqref{eq:w5} and \eqref{eq:w5_}. Hence, the generalized Lee-Yang theorem in \thmref{generalized_Lee_yang} shows that the zeros of the classical system are located on the imaginary axis. As the error $\e$ in the mapping goes to zero, we get a family of classical partition functions that approach the quantum partition function. \thmref{Multivariate Hurwitz' theorem} implies that the complex zeros of the quantum and classical systems coincide in the limit of $\e\rightarrow 0$. Thus, the complex zeros of the quantum system are also located on the imaginary axis. 
\end{proof}

\begin{rem}
One can extend the result of \thmref{generalized_Lee_yang} to include interactions between spins greater than spin $1/2$. It is shown in \cite{suzuki_zeros_higher_spin} that the partition function of the Heisenberg model with spin $s$ particles can be mapped to that of a spin $1/2$ Heisenberg model as in \eqref{eq:Heisenbrg}. Therefore, the Lee-Yang theorem holds for these systems too. 
\end{rem} 
\subsection{An algorithm for the anisotropic XXZ model} \label{sec:An algorithm for the anisotropic xxz model}
In \secref{Complex zeros of the ferromagnetic systems}, we studied the location of the complex zeros of a $2$-local Hamiltonian when the external magnetic field $\mu$ is varied. Here, we focus on a specific subclass of those Hamiltonians for which we can find an approximation algorithm. Particularly, we consider the \emph{anisotropic} $\mathrm{XXZ}$ model which has the following Hamiltonian:
\ba\label{eq:XXZ}
H(\mu)=-\sum_{(i,j)\in E} \left(J_{ij}(X_i X_j+Y_i Y_j)+J_{ij}^{zz}Z_i Z_j\right)-\frac{\mu}{2} \sum_{i\in V}(Z_i+\iden).
\ea
Compared to the Heisenberg model, the $\mathrm{XXZ}$ model assigns equal coefficients to the $X_i X_j$ and $Y_i Y_j$ terms and does not include $X_i$ and $Y_i$ terms. An important property of this model that we use in our algorithm is that 
\ba\label{eq:w21}
\left[H(\mu),\frac{\mu}{2} \sum_{i\in V}(Z_i+\iden) \right]=0.
\ea
To see this, notice that $[X_iX_j+Y_iY_j,Z_i+Z_j]=0$.

Let $\ket{s_1,s_2,\dots, s_n}$ be an assignment of spins $\pm 1$ to all the vertices. Any such vector is an eigenstate of $1/2 \sum_{i=1}^n (Z_i+\iden)$, that is,
\ba
\frac{1}{2}\sum_{i=1}^n (Z_i+\iden) \ket{s_1,s_2,\dots, s_n} = \frac{1}{2}(\sum_{i=1}^n s_i+n)\ket{s_1,s_2,\dots, s_n}.
\ea
Let $\cH_k$ denote the eigenspace of $1/2\sum_{i=1}^n (Z_i+\iden)$ that corresponds to the $k$th eigenvalue. This subspace is spanned by the binary strings of length $n$ with Hamming weight $k$. We have: 
\ba
\forall \ket{v}\in \cH_k,\quad \frac{1}{2}\sum_{i=1}^n (Z_i+\iden) \ket{v}=k\ket{v}.
\ea
We can partition the Hilbert space of the $n$ vertices $\cH$ into $\cH=\oplus_k \cH_{k}$. The dimension of each of these subspaces $\dim(\cH_k)$ is $\binom{n}{k}$. Since \eqref{eq:w21} holds, the partition function of this model can be written as a polynomial in $z=\exp(\b \m)$. 

\begin{lem}\label{lem:XXZ polynomial}
The partition function of the anisotropic $\mathrm{XXZ}$ model given in \eqref{eq:XXZ} can be written as 
\ba
Z_{\b}(H(\mu))=\sum_{k=0}^n q_k z^k,\label{eq:w22}
\ea
where $z=e^{\b \m}$ and the coefficients $q_k$ are defined by
\ba
q_k=\Tr_{\cH_k}[e^{\b\sum_{(i,j)\in E} \left(J_{ij}(X_i X_j+Y_i Y_j)+J_{ij}^{zz}Z_i Z_j\right)}].
\ea
\end{lem}
\begin{proof}[\textbf{Proof of \lemref{XXZ polynomial}}]
We have
\ba
Z_{\b}(H(\mu))&=\Tr_{\cH}[e^{-\b H(\mu)}]\nonumber\\
&=\Tr_{\cH}[e^{\b\sum_{(i,j)\in E} \left(J_{ij}(X_i X_j+Y_i Y_j)+J_{ij}^{zz}Z_i Z_j\right)}e^{\b \mu/2 \sum_{i\in V}(Z_i+\iden)}.]\nonumber\\
&=\sum_{k=0}^n e^{\b \mu k}\Tr_{\cH_k}[e^{\b\sum_{(i,j)\in E} \left(J_{ij}(X_i X_j+Y_i Y_j)+J_{ij}^{zz}Z_i Z_j\right)}]\nonumber\\
&=\sum_{k=0}^n q_k z^k.\label{eq:w28}\nonumber
\ea
\end{proof}

Now, we are ready to state an algorithm for this model.
\begin{thm}[Approximation algorithm for the partition function of the $\mathrm{XXZ}$ model]\label{thm:alg for xxz}
There is an algorithm that runs in $n^{O(\log (n/\e))}$ time and outputs an $\e$-multiplicative approximation to the partition function of the anisotropic $\mathrm{XXZ}$ model in the ferromagnetic regime, i.e. when $J_{ij}^{zz}\geq |J_{ij}|$ and $\mu$ is an arbitrary constant.
\end{thm}

\begin{proof}[\textbf{Proof of \thmref{alg for xxz}}]
By \lemref{XXZ polynomial}, the partition function is a polynomial of degree $n$ given in \eqref{eq:w22}. The location of its zeros is given by \thmref{t_higH_field}. Hence, we can apply the truncated Taylor series of \propref{taylor_bounded} to obtain an approximation algorithm for $Z_{\b}(H(\mu))$.

According to \lemref{XXZ polynomial}, the partition function of this system is
\ba
Z_{\b}(H(\mu))= \sum_{k=0}^n q_k z^k.\nonumber
\ea
The running time of the extrapolation algorithm is dominated by the calculation of the coefficients $q_k$ of the Taylor expansion, where $q_k$ is
\ba
q_k=\Tr_{\cH_k}[e^{\b\sum_{(i,j)\in E} \left(J_{ij}(X_i X_j+Y_i Y_j)+J_{ij}^{zz}Z_i Z_j\right)}]
\ea
and $\dim (\cH_k)=\binom{n}{k}$. In general, we can decompose the Hilbert space of the system as $\cH=\oplus_k \cH_k$. The local Hamiltonian $H$ is block diagonal in this basis. Since $H$ is sum of local terms, it takes time $n^{O(k)}$ to compute the entries of $H$ and diagonalize it in the block corresponding to the subspace $\cH_k$. Then we can find the trace of the exponential of this block also in time $n^{O(k)}$. Since we only need $k=O(\log(n))$ in the truncated Taylor expansion, we achieve an overall running time of $n^{O(\log(n/\e))}$. 
\end{proof}

Even though \thmref{t_higH_field} applies to a broader class of $2$-local Hamiltonians such as the Heisenberg model, our method does not immediately give an algorithm for those Hamiltonians. The reason is a technical difficulty in representing the partition function of these models as a polynomial in $\exp(\b \mu)$. This turns out not to be an issue for the $\mathrm{XXZ}$ model since the $1$-body terms $\sum_i Z_i$ commute with the rest of the Hamiltonian. 

One might wonder why we could not use the extrapolation algorithm directly for the classical system that we get after the mapping in \propref{quantum_to_classical}. After all, the partition function of this system is also a polynomial of degree $\poly(n)$ and the location of its zeros is the same as that of the quantum system. It seems that at least naively applying this idea does not work. This is because the point that we want to extrapolate to in the classical system is $\mu/\eta$ instead of $\mu$. For the error of the mapping to be $1/\poly(n)$, we need $\eta$ to be $\poly(n)$. Thus, the ending point of the extrapolation is vanishingly close to the imaginary axis where the zeros are located. This makes the running time blow up and become exponential instead of quasi-polynomial.

Note that sampling algorithms like the ones used in \cite{bravyi_classicalmapping,Bravyi_ferro} do not encounter this problem. The running time of these algorithms remains efficient even if the parameters of the classical Hamiltonian scale with the number of particles $n$. There are unfortunately no randomized algorithms based on sampling known for the $4$-local classical Hamiltonian obtained in the mapping of \propref{quantum_to_classical}. We leave extending our result to cover all the Hamiltonians considered in \thmref{t_higH_field} as a challenge for future work.

\section*{Acknowledgements} We thank Fernando Brand\~ao, Kohtaro Kato, Zeph Landau, Milad Marvian, and John Wright for helpful discussions.
This work was funded by NSF grants CCF-1452616, CCF-1729369, PHY-1818914; ARO contract
W911NF-17-1-0433; and a
Samsung Advanced Institute of Technology Global Research Partnership.

\bibliographystyle{alpha}
\bibliography{main}

\newcommand{\etalchar}[1]{$^{#1}$}
\begin{thebibliography}{CPGW15}

\bibitem[AGIK09]{aharonov_line}
Dorit Aharonov, Daniel Gottesman, Sandy Irani, and Julia Kempe.
\newblock The power of quantum systems on a line.
\newblock {\em Communications in Mathematical Physics}, 287(1):41--65, 2009.

\bibitem[ALVV17]{arad_1d}
Itai Arad, Zeph Landau, Umesh Vazirani, and Thomas Vidick.
\newblock Rigorous {RG} algorithms and area laws for low energy eigenstates in
  1{D}.
\newblock {\em Comm. Math. Phys.}, 356(1):65--105, 2017.

\bibitem[Ara69]{araki_1d}
Huzihiro Araki.
\newblock Gibbs states of a one dimensional quantum lattice.
\newblock {\em Communications in Mathematical Physics}, 14(2):120--157, 1969.

\bibitem[BaH13]{brandao2013product}
Fernando Brand\~{a}o and Aram Harrow.
\newblock Product-state approximations to quantum ground states.
\newblock In {\em {P}roceedings of the 2013 {ACM} {S}ymposium on {T}heory of
  {C}omputing---{STOC} 2013}, pages 871--880. ACM, New York, 2013.

\bibitem[Bar15]{barvinok_clique}
Alexander Barvinok.
\newblock Computing the partition function for cliques in a graph.
\newblock {\em Theory of Computing. An Open Access Journal}, 11:339--355, 2015.

\bibitem[Bar16a]{Barvinok_book}
Alexander Barvinok.
\newblock {\em Combinatorics and complexity of partition functions}, volume~30
  of {\em Algorithms and Combinatorics}.
\newblock Springer, Cham, 2016.

\bibitem[Bar16b]{Barvinok_permanent}
Alexander Barvinok.
\newblock Computing the permanent of (some) complex matrices.
\newblock {\em Foundations of Computational Mathematics}, 16(2):329--342, 2016.

\bibitem[BDOT08]{Bravyi_stoquastic}
Sergey Bravyi, David DiVincenzo, Roberto Oliveira, and Barbara Terhal.
\newblock The complexity of stoquastic local {H}amiltonian problems.
\newblock {\em Quantum Information \& Computation}, 8(5):361--385, 2008.

\bibitem[BG17]{Bravyi_ferro}
Sergey Bravyi and David Gosset.
\newblock Polynomial-time classical simulation of quantum ferromagnets.
\newblock {\em Physical Review Letters}, 119(10):100503, 2017.

\bibitem[BGM19]{bravyi2019_meanvalues}
Sergey Bravyi, David Gosset, and Ramis Movassagh.
\newblock Classical algorithms for quantum mean values.
\newblock {\em arXiv preprint arXiv:1909.11485}, 2019.

\bibitem[BK16]{Brandao_gibbs_preparing}
Fernando Brand{\~a}o and Michael~J Kastoryano.
\newblock Finite correlation length implies efficient preparation of quantum
  thermal states.
\newblock {\em Communications in Mathematical Physics}, pages 1--16, 2016.

\bibitem[Bra15]{bravyi_classicalmapping}
Sergey Bravyi.
\newblock Monte {C}arlo simulation of stoquastic {H}amiltonians.
\newblock {\em Quantum Information \& Computation}, 15(13-14):1122--1140, 2015.

\bibitem[CPGW15]{cubitt_undecidablegapp}
Toby~S Cubitt, David Perez-Garcia, and Michael~M Wolf.
\newblock Undecidability of the spectral gap.
\newblock {\em Nature}, 528(7581):207, 2015.

\bibitem[CS17]{chowdhury_gibbs_sampling}
Anirban~Narayan Chowdhury and Rolando Somma.
\newblock Quantum algorithms for {G}ibbs sampling and hitting-time estimation.
\newblock {\em Quantum Information \& Computation}, 17(1-2):41--64, 2017.

\bibitem[Dob96]{Dobrushin_estimates}
Roland Dobrushin.
\newblock Estimates of semi-invariants for the {I}sing model at low
  temperatures.
\newblock {\em Translations of the American Mathematical Society-Series 2},
  177:59--82, 1996.

\bibitem[DS87]{Dobrushin1}
Roland Dobrushin and Senya Shlosman.
\newblock Completely analytical interactions: constructive description.
\newblock {\em Journal of Statistical Physics}, 46(5-6):983--1014, 1987.

\bibitem[EM17]{Mehraban_permanent}
Lior Eldar and Saeed Mehraban.
\newblock Approximating the permanent of a random matrix with vanishing mean.
\newblock {\em arXiv preprint arXiv:1711.09457}, 2017.

\bibitem[Fis65]{Fisher}
Michael Fisher.
\newblock The nature of critical points.
\newblock {\em Lecture notes in Theoretical Physics}, 7c:1--159, 1965.

\bibitem[Gam03]{gamelin_book_complex}
Theodore Gamelin.
\newblock {\em Complex analysis}.
\newblock Springer Science \& Business Media, 2003.

\bibitem[Gre69]{greenberg_cluster}
William Greenberg.
\newblock Correlation functionals of infinite volume quantum spin systems.
\newblock {\em Communications in Mathematical Physics}, 11(4):314--320, 1969.

\bibitem[Has06]{Hastings_solving_gapped_locally}
Matthew Hastings.
\newblock Solving gapped hamiltonians locally.
\newblock {\em Physical Review B}, 73(8):085115, 2006.

\bibitem[Has07]{hastings_belief_propagation}
Matthew Hastings.
\newblock Quantum belief propagation: An algorithm for thermal quantum systems.
\newblock {\em Physical Review B}, 76(20):201102, 2007.

\bibitem[JS93]{Jerrum_ising}
Mark Jerrum and Alistair Sinclair.
\newblock Polynomial-time approximation algorithms for the {I}sing model.
\newblock {\em SIAM Journal on Computing}, 22(5):1087--1116, 1993.

\bibitem[KBa16]{kastoryano_commuting}
Michael Kastoryano and Fernando Brand\~{a}o.
\newblock Quantum {G}ibbs samplers: the commuting case.
\newblock {\em Communications in Mathematical Physics}, 344(3):915--957, 2016.

\bibitem[KBa19]{Brandao_gibbs_cmi}
Kohtaro Kato and Fernando Brand\~{a}o.
\newblock Quantum approximate {M}arkov chains are thermal.
\newblock {\em Communications in Mathematical Physics}, 370(1):117--149, 2019.

\bibitem[KGK{\etalchar{+}}14]{Kastoryano_locality}
Martin Kliesch, Christian Gogolin, MJ~Kastoryano, A~Riera, and J~Eisert.
\newblock Locality of temperature.
\newblock {\em Physical Review X}, 4(3):031019, 2014.

\bibitem[Kim17]{Kim_gibbs}
Isaac Kim.
\newblock Markovian matrix product density operators: Efficient computation of
  global entropy.
\newblock {\em arXiv preprint arXiv:1709.07828}, 2017.

\bibitem[KKB19]{Kohtaro_cmi_cluster}
Tomotaka Kuwahara, Kohtaro Kato, and Fernando Brandao.
\newblock Private communication.
\newblock 2019.

\bibitem[KP86]{kotecky_cluster}
Roman Koteck{\`y} and David Preiss.
\newblock Cluster expansion for abstract polymer models.
\newblock {\em Communications in Mathematical Physics}, 103(3):491--498, 1986.

\bibitem[KS18]{gibbs_one_dim}
Tomotaka Kuwahara and Keiji Saito.
\newblock Polynomial-time classical simulation for one-dimensional quantum
  gibbs states.
\newblock {\em arXiv preprint arXiv:1807.08424}, 2018.

\bibitem[LSS19a]{Liu_zeros}
Jingcheng Liu, Alistair Sinclair, and Piyush Srivastava.
\newblock Fisher zeros and correlation decay in the {I}sing model.
\newblock In {\em 10th {I}nnovations in {T}heoretical {C}omputer
  {S}cience---{ITCS}}. 2019.

\bibitem[LSS19b]{Liu_ising}
Jingcheng Liu, Alistair Sinclair, and Piyush Srivastava.
\newblock The {I}sing partition function: zeros and deterministic
  approximation.
\newblock {\em Journal of Statistical Physics}, 174(2):287--315, 2019.

\bibitem[LW05]{levin_wen}
Michael Levin and Xiao-Gang Wen.
\newblock String-net condensation: A physical mechanism for topological phases.
\newblock {\em Physical Review B}, 71(4):045110, 2005.

\bibitem[LY52]{Lee-yang}
Tsung-Dao Lee and Chen-Ning Yang.
\newblock Statistical theory of equations of state and phase transitions. ii.
  lattice gas and {I}sing model.
\newblock {\em Physical Review}, 87(3):410, 1952.

\bibitem[MB18]{mann2018_approximation}
Ryan Mann and Michael Bremner.
\newblock Approximation algorithms for complex-valued {I}sing models on bounded
  degree graphs.
\newblock {\em arXiv preprint arXiv:1806.11282}, 2018.

\bibitem[MSVC15]{Cirac_gibbs_peps}
Andras Molnar, Norbert Schuch, Frank Verstraete, and Ignacio Cirac.
\newblock Approximating gibbs states of local {H}amiltonians efficiently with
  projected entangled pair states.
\newblock {\em Physical Review B}, 91(4):045138, 2015.

\bibitem[Par82]{Park_cluster}
Yong~Moon Park.
\newblock The cluster expansion for classical and quantum lattice systems.
\newblock {\em Journal of Statistical Physics}, 27(3):553--576, 1982.

\bibitem[PR18]{Regts_ising}
Han Peters and Guus Regts.
\newblock Location of zeros for the partition function of the {I}sing model on
  bounded degree graphs.
\newblock {\em arXiv preprint arXiv:1810.01699}, 2018.

\bibitem[PW09]{poulin2009sampling}
David Poulin and Pawel Wocjan.
\newblock Sampling from the thermal quantum gibbs state and evaluating
  partition functions with a quantum computer.
\newblock {\em Physical Review Letters}, 103(22):220502, 2009.

\bibitem[SF71]{Suzuki_XYZ}
Masuo Suzuki and Michael Fisher.
\newblock Zeros of the partition function for the {H}eisenberg, ferroelectric,
  and general {I}sing models.
\newblock {\em Journal of Mathematical Physics}, 12(2):235--246, 1971.

\bibitem[Sly10]{sly_hardcore}
Allan Sly.
\newblock Computational transition at the uniqueness threshold.
\newblock In {\em 2010 {IEEE} 51st {A}nnual {S}ymposium on {F}oundations of
  {C}omputer {S}cience---{FOCS} 2010}, pages 287--296. IEEE Computer Soc., Los
  Alamitos, CA, 2010.

\bibitem[SS12]{sly_sun}
Allan Sly and Nike Sun.
\newblock The computational hardness of counting in two-spin models on
  {$d$}-regular graphs.
\newblock In {\em 2012 {IEEE} 53rd {A}nnual {S}ymposium on {F}oundations of
  {C}omputer {S}cience---{FOCS} 2012}, pages 361--369. IEEE Computer Soc., Los
  Alamitos, CA, 2012.

\bibitem[SST14]{Sinclair_antiferro}
Alistair Sinclair, Piyush Srivastava, and Marc Thurley.
\newblock Approximation algorithms for two-state anti-ferromagnetic spin
  systems on bounded degree graphs.
\newblock {\em Journal of Statistical Physics}, 155(4):666--686, 2014.

\bibitem[Suz69]{suzuki_zeros_higher_spin}
Masuo Suzuki.
\newblock On the distribution of zeros for the {H}eisenberg model.
\newblock {\em Progress of Theoretical Physics}, 41(6):1438--1449, 1969.

\bibitem[Wei04]{weitz2004mixing}
Dror Weitz.
\newblock {\em Mixing in time and space for discrete spin systems}.
\newblock ProQuest LLC, Ann Arbor, MI, 2004.
\newblock Thesis (Ph.D.)--University of California, Berkeley.

\bibitem[Wei06]{Weitz}
Dror Weitz.
\newblock Counting independent sets up to the tree threshold.
\newblock In {\em {P}roceedings of the 38th {A}nnual {ACM} {S}ymposium on
  {T}heory of {C}omputing---{STOC} 2006}, pages 140--149. ACM, New York, 2006.

\end{thebibliography}

\end{document}